\documentclass[journal,12pt,onecolumn,draftclsnofoot,]{IEEEtran}
\usepackage{subcaption}
\usepackage{graphicx}
\usepackage{amsmath}
\usepackage{amssymb }
\usepackage{graphicx}
\usepackage{amsmath}
\usepackage{amssymb}
\usepackage{amsthm}
\usepackage{lipsum}
\usepackage{amsmath}   
\usepackage{amssymb}   
\usepackage{amsfonts}   
\usepackage{lipsum}   
\usepackage{float}   
\usepackage{cuted}     
\usepackage{amsthm}
\usepackage[dvipsnames]{xcolor}

\newtheorem{thm}{Theorem}
\newtheorem{lem}[thm]{Lemma}
\newtheorem{prop}[thm]{Proposition}
\newtheorem{cor}{Corollary}

\theoremstyle{definition}

\makeatletter
\let\NAT@parse\undefined
\makeatother
\usepackage{hyperref}

\usepackage[ruled,vlined]{algorithm2e}
\def\BibTeX{{\rm B\kern-.05em{\sc i\kern-.025em b}\kern-.08em
    T\kern-.1667em\lower.7ex\hbox{E}\kern-.125emX}}

\begin{document}

\title{Two-Timescale Design for RIS-aided Cell-free Massive MIMO Systems with Imperfect CSI}
\author{Mahdi Eskandari,
        Kangda Zhi,
        Huiling Zhu,
        Cunhua Pan, \\
        and~Jiangzhou Wang,~\IEEEmembership{Fellow,~IEEE,}
        
\thanks{Mahdi Eskandari is with the School of Engineering, University
of Kent, UK. (email: me377@kent.ac.uk). Kangda Zhi is with the School of
Electronic Engineering and Computer Science at Queen Mary University of London, UK. (email: k.zhi@qmul.ac.uk). Huiling Zhu is with the School of Engineering, University
of Kent, UK. (email: H.Zhu@kent.ac.uk). Cunhua
Pan is with the National Mobile Communications Research Laboratory, Southeast University,
Nanjing 211111, China. (email: cpan@seu.edu.cn). Jiangzhou Wang is with the School of Engineering, University
of Kent, UK. (email: j.z.wang@kent.ac.uk)}
}

\maketitle

\begin{abstract}
The objective of this paper is to evaluate the effectiveness of a two-timescale transmission design in cell-free massive multi-input multiple-output (MIMO) systems incorporating reconfigurable intelligent surfaces (RISs) under the assumption of imperfect channel state information (CSI).
We examine the Rician channel model and formulate the passive beamforming for the RISs based on statistical channel state information (S-CSI).
To that end, we put forth a linear minimum mean square error (LMMSE) estimator with the aim of estimating the aggregation of channels from the users to the APs within each channel coherence interval. Meanwhile, the active beamforming for the radio units (APs) is executed using the maximum ratio combining (MRC) approach, which utilizes the instantaneous aggregated channels, that result from the combination of the direct and reflected channels from the RISs. 
Subsequently, we derive the closed-form expressions of the achievable uplink spectral efficiency (SE), which is a function of S-CSI elements such as distance-dependent path loss, Rician factors as well as the number of RIS elements and AP antennas.
We then optimize the phase shifts of the RISs to maximize the sum SE of the users, utilizing the soft actor-critic (SAC) which is a deep reinforcement learning (RL) method, and relying on the derived closed-form expressions.
Numerical evaluations affirm that, despite the presence of imperfect CSI, the deployment of RIS in cell-free systems can lead to significant performance improvement.
\end{abstract}
\begin{IEEEkeywords}
Reconfigurable intelligent surface (RIS), deep reinforcement learning, soft actor-critic, S-CSI, cell-free MIMO.
\end{IEEEkeywords}

\IEEEpeerreviewmaketitle

\section{Introduction}
It has become increasingly important for academia and industry to study cell-free massive multi-input multi-output (MIMO) as a potential technology for 5G and 6G networks \cite{zhang2019cell, ngo2017cell, yang2018energy}. Cell-free systems consist of a large number of access points (APs) distributed over a small coverage area to serve a small number of users \cite{wang2020uplink, interdonato2019ubiquitous, interdonato2021enhanced}. Unlike cellular massive MIMO, cell-free systems do not have cells, and therefore do not suffer from inter-cell interference. For more sophisticated operations, each AP is equipped with limited computing power units in the serving area and connected to the CPU via a low-latency backhaul network. There are several disadvantages associated with cell-free wireless networks, including cost, power consumption of access points, blockage, and managing a large number of APs. In addition, the AP deployment in cell-free systems can be expensive and energy-intensive compared to a standard cellular network.  Signal degradation or loss can occur due to blockages caused by objects such as buildings or trees \cite{wang2020uplink}.

Reconfigurable intelligent surface (RIS) is an emerging technology that allows radio waves to be shaped at the electromagnetic level without using digital signal processing methods or requiring power amplifiers \cite{liu2021reconfigurable}. Radio frequency (RF) chains and power amplifiers are not required in the RIS because each element scatters (reflects) the incident signal. RIS has the potential to address some of the limitations of cell-free wireless networks. It can be used to steer and focus signals around obstacles, improving coverage and capacity, while reducing the number of distributed APs required, which can in turn lower the power consumption of the network \cite{wu2021intelligent}. 

For simplicity, the main attention has far been concentrated on designing the phase shifts under the assumption of perfect channel state information (CSI). See \cite{zhao2020intelligent, zappone2020overhead, abrardo2021intelligent, zheng2019intelligent}. In \cite{wei2021channel}, the authors discussed the fundamental issues of perfect channel estimation in RIS-aided systems. The impact of the channel estimation overhead on spectral efficiency (SE) and energy efficiency (EF) was investigated in \cite{zappone2020overhead}. In \cite{zappone2020overhead}, \cite{abrardo2021intelligent} and \cite{eskandari2022statistical} to reduce the impact of channel estimation overhead, the authors investigated the design of the RIS in the presence of statistical CSI. 

As for the integration of the cell-free massive MIMO and RIS, recent works have formulated and solved optimization problems with different objectives under the assumption of perfect and instantaneous CSI \cite{zhou2020achievable, zhang2020capacity, bashar2020performance, zhang2020reconfigurable, zhang2021beyond}. The work \cite{zhang2021beyond} considered a downlink RIS-aided cell-free network and jointly designed the active and passive beamforming vectors with the assumption of perfect channel knowledge available at the CPU. The work \cite{zhang2021joint} considered a downlink RIS-aided cell-free network, in a typical wideband scenario, the authors formulated the problem of joint precoding design at the APs and RISs to maximize the
weighted sum-rate (WSR) of all users to improve the network capacity. Unlike the previous works that the active and passive beamforming vectors were designed in a centralized manner in the CPU, the work \cite{huang2020decentralized} proposed a downlink distributed scheme for active and passive beamforming where the joint active and passive beamformers are designed locally at each AP and based on the communications between neighbouring APs.

All the aforementioned works considered the instantaneous CSI (I-CSI) knowledge available at the APs and the CPU. In this paper, the uplink two-timescale design of a RIS-aided cell-free massive MIMO system is proposed subject to imperfect aggregated CSI.  
To the best of our knowledge, only a few research works investigated RIS-aided massive
MIMO systems based on the two-timescale design \cite{zhi2022power, 8746155, wu2022two, zhi2022two}.  The authors of \cite{zhi2022power} proposed to
employ the RIS for serving users that are located in the out-of-coverage areas of massive MIMO
systems. In \cite{8746155}, the authors evaluated the performance of a RIS-assisted large-scale antenna system by formulating an upper bound of the ergodic SE and investigated the effect of the phase shifts on the ergodic SE. The authors of \cite{wu2022two} investigated the two-timescale design of the simultaneously transmitting and reflecting reconfigurable intelligent surfaces (STAR-RISs) where the phase shifts of the STAR-RIS is designed based on S-CSI whereas the active beamforming at the BS is designed based on I-CSI. All the aforementioned works assumed the availability of the perfect channel knowledge for the two-timescale design. The work \cite{zhi2022two} investigates the uplink two-timescale design of the single-cell massive MIMO system in the presence of imperfect channel knowledge. In \cite{zhi2022two}, the authors have designed the phase shifts of the RIS based on statistical CSI (S-CSI) where the MRC scheme based on I-CSI was used for uplink detection. Additionally, to the best of our knowledge, the two-timescale design of the cell-free RIS-aided massive MIMO system with imperfect CSI has not been studied yet.

In this paper, an uplink two-timescale design of a RIS-aided cell-free massive MIMO system is proposed subject to imperfect aggregated CSI. Specifically, in each channel coherence time interval, the linear minimum mean square error (LMMSE) method is utilized to perform the aggregated channel estimation. Hence, instead of estimating the user-RIS, RIS-AP, and user-AP channels independently, the aggregated channel from user to AP link is estimated which has a pilot overhead as conventional massive MIMO systems. Furthermore, the low-complexity MRC scheme is applied in each AP to detect the transmitted signal from each user.  
The phase shift of the RISs is designed using soft actor-critic (SAC) which is of the family of deep reinforcement learning (RL) algorithms.
SAC is an RL algorithm that combines the actor-critic framework with the entropy-regularized reinforcement learning approach. It aims to optimize both the expected reward and the entropy of the policy, making it well-suited for problems with uncertain or complex environments. The algorithm has a softmax function to calculate the policy, which allows it to handle continuous action spaces. Additionally, SAC uses an off-policy learning method, which means it can learn from historical data and does not require a fixed behaviour policy. SAC has been shown to achieve state-of-the-art performance in a variety of RL tasks and is considered to be robust and versatile algorithm for RL.

The main contributions of the paper are
summarized as follows:
\begin{itemize}
    \item We first derive the aggregated channel estimation through the uplink pilot transmission. The channel estimation is done in each AP and it is not being sent to the CPU. Thus the only information that the CPU utilizes is  the S-CSI. Also, the effect of pilot contamination on channel estimation is studied. 
    \item A two-step signal detection approach is utilized. In the first step, each AP  detects the uplink signal from each user locally based on the estimated channel obtained in the channel estimation phase. Each AP used MRC to decode the signal at the first step. Then, the locally detected signal of all the APs is sent to the CPU for final detection. Since the CPU has only access to the S-CSI, the second step of detection is done using S-CSI only.  
    \item In order to reduce the complexity of designing the phase shifts of the RISs, the phase shift of the RISs is designed at the CPU based on the slow-varying S-CSI. In this case, the CPU is responsible to design the phase shift RIS panels and the CPU uses SAC to design the phase shifts of the RISs. 
\end{itemize}

The remainder of this paper is organized as follows. Section II describes the system model
of the considered RIS-aided cell-free Massive MIMO system. Section III derives the LMMSE channel
estimator. In Section IV, we derive closed-form
expressions for the uplink ergodic SE. In Section
V, we introduce the statistical CSI-based design for RISs in the cell-free massive MIMO system based on SAC. Section VI provides the numerical results and Section
VII concludes the paper.

\textbf{Notation:} $\mathbf{X}$ denotes a matrix, and column vectors are denoted by boldface uppercase letters $\mathbf{x}$.
The transpose, conjugate, conjugate transpose, and inverse of matrix $\mathbf{X}$ are denoted
by $\mathbf{X}^T$ , $\mathbf{X}^*$, $\mathbf{X}^H$ and $\mathbf{X}^{-1}$, respectively.
The
trace, expectation, and covariance operators are denoted by $\mathrm{Tr}\{.\}$, $\mathbb{E}\{.\}$, and $\mathrm{Cov}\{.\}$,
respectively. 
$\mathbb{C}^{M \times N}$ denotes the space of $M \times N$ complex matrices. $\mathbf{I}_N$ denotes
the $N \times N$ identity matrix
Also, $\mathcal{X}$ is a set. The $\mod (., .)$,$\rVert.\rVert$, and $\lfloor . \rfloor$ denotes the modulus
operation, the Euclidean norm and the floor function rounds a number down to the nearest integer, respectively. A diagonal matrix with the vector $\mathbf{x}$ on its diagonal entries is shown by $\mathrm{diag}\{ \mathbf{x}\}$. A complex Gaussian distributed vector $\mathbf{x}$ with mean $\Bar{\mathbf{x}}$ and covariance
matrix $\mathbf{C}$ is denoted by $\mathbf{x} \sim \mathcal{CN} (\Bar{\mathbf{x}}, \mathbf{C})$.

\section{System Model}

\begin{table}[t]
\caption{Major Notations} 
\centering 
\begin{tabular}{c c} 
\hline\hline 
Notation & Description \\ [0.5ex] 
\hline 
$J$ & Number of APs \\
$N$ & Number of antennas of each AP  \\
$K$ & Number of users \\
$R$ & Number of RISs \\
$M$ & Number of elements of each RIS \\
$\theta_{r, j}$ &  The phase shift of the $m$th RIS element of the $r$th RIS \\
$\boldsymbol{\Phi}_r $ & Phase shift matrix of the $r$th RIS \\
$\mathbf{H}_{r, j}$ & The channel matrix between the $r$th RIS and the $j$th AP  \\
$\mathbf{h}_{j, k}$ & The channel vector between the $j$th AP and $k$th user  \\
$\mathbf{g}_{r,k}$ & The channel vactor between the $r$th RIS and the $k$th user  \\
$\mathbf{g}_{j, r,k}$ & The cascaded channel between AP $j$ and user $k$  \\
$\mathbf{q}_{r, j, k}$ & The aggregated channel from user $k$ to AP $j$  \\
$\tau_c$ & Coherence time duration  \\
$\tau_p$ & pilot training duration  \\
$\tau_u$ & Data transmission duration  
 \\[1ex]  
\hline  
\end{tabular}
\label{table:cqi} 
\end{table}


\begin{figure}[t] \label{model}
 \centering
    \begin{subfigure}{0.80\textwidth} \centering
            \includegraphics[scale=0.40]{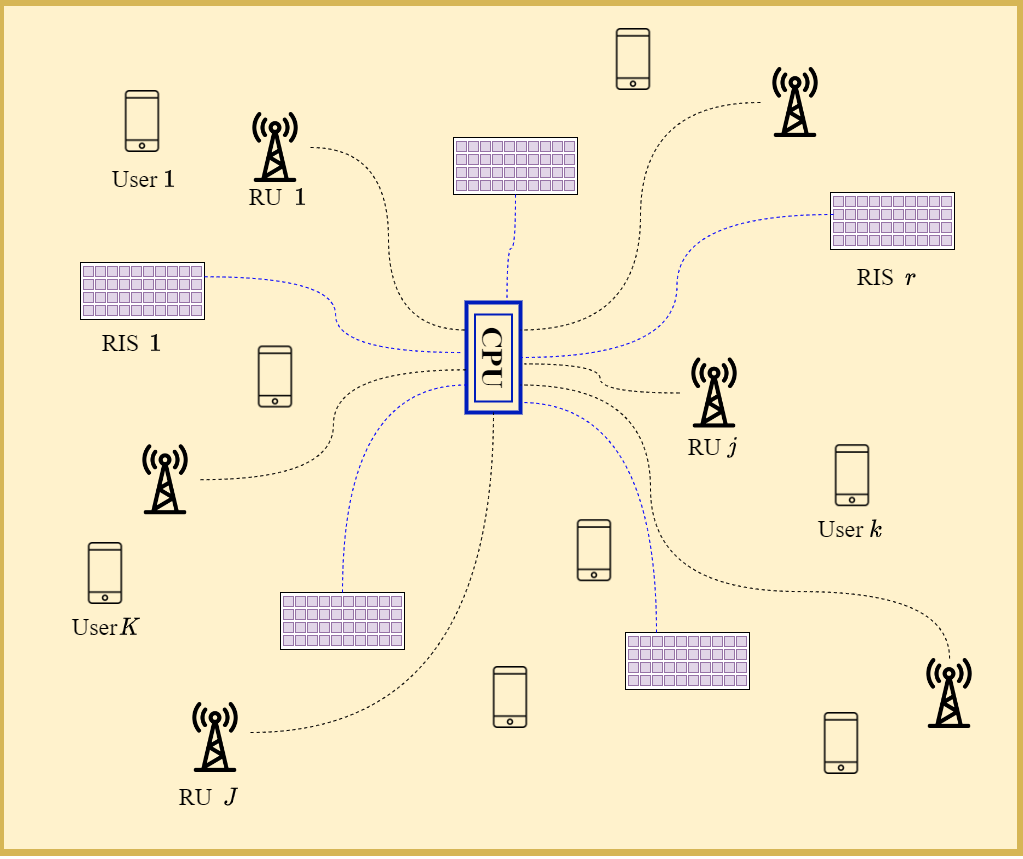}
        \caption{}
        \label{Prinzipbild}
    \end{subfigure} \newline
    \begin{subfigure}{0.80\textwidth}  \centering
    \includegraphics[scale=0.22]{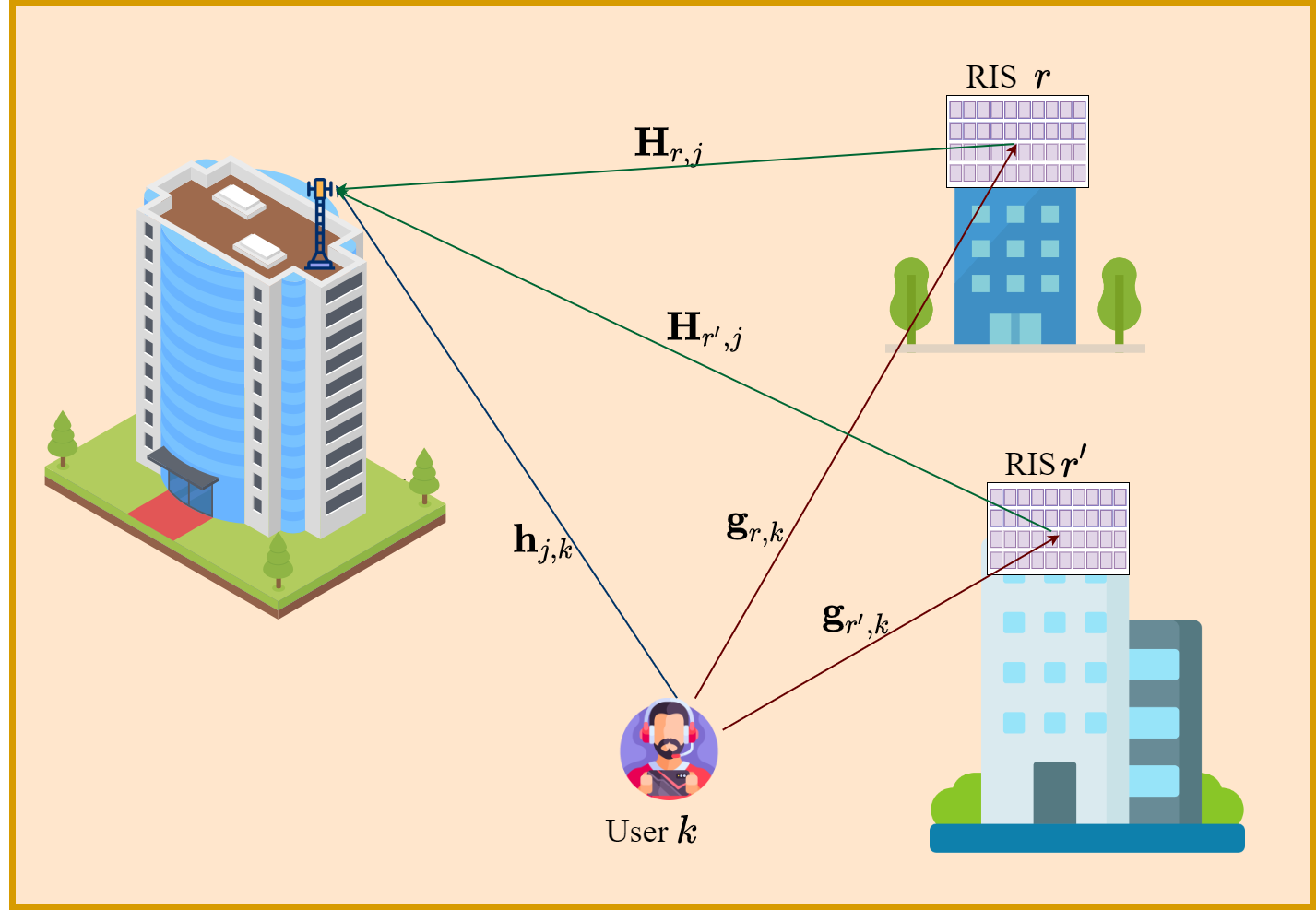}
    \caption{}
        \label{realSensor}
    \end{subfigure}
        \caption{(a) RIS aided cell-free MIMO system. (b) Detailed illustration of the channels}
\end{figure}

\subsection{Network Description}
We consider a cell-free massive MIMO system with $J$ radio units (APs) with $N$ antennas and $K$ single antenna users that are randomly located in an area. The transmission is assisted with the help of $R$ RIS panels each of which has $M$ elements. See Fig.~\ref{Prinzipbild} and Fig.~\ref{realSensor}. We denote the set of APs, RISs, and users as $\mathcal{J} = \{1, 2, \dots, J\}$, $\mathcal{R} = \{1, 2, \dots, R\}$, and $\mathcal{K} = \{1, 2, \dots, K\}$, respectively. Table~\ref{table:cqi} lists the full list of notations. It is assumed that all of the $K$ users are served by all the APs and RISs in the same time-frequency resource. The transmission from APs to the users (downlink) and the transmission from the users to the APs (uplink) are all in TDD operation. Hence, each coherence interval is divided into three sub-intervals:~i)~uplink channel training.~ii)~uplink payload data transmission.~iii)~downlink payload data transmission. Each AP is responsible for estimating the channel to all the RISs and users. The estimated channel is then used to precode the transmitted signal to the users in the downlink transmission and decode the received signal from the users in the uplink transmission. Thanks to the TDD mode operation, the channels are reciprocal which means the channel gains of the uplink and downlink are the same. The usage of TDD mode requires perfect hardware chain calibration \cite{kaltenberger2010relative,ngo2017cell}. Investigation of the impact of imperfect hardware chain calibration is left for future work. 

The RISs are modeled with a phase shift matrix denoted by $\boldsymbol{\Phi}_r = \mathrm{diag}(e^{j \theta_{r, 1}}, e^{j \theta_{r, 2}}, \dots, e^{j \theta_{r, M}})$ where $\theta_{r, m} \in [0, 2\pi)$ is the phase shift of the $m$th element of the $r$th RIS. By assuming $\mathbf{H}_{r, j}\in \mathbb{C}^{M \times N}$, $\mathbf{h}_{j, k} \in \mathbb{C}^{N \times 1}$
 and $\mathbf{g}_{r, k} \in \mathbb{C}^{M \times 1}$ as the channel between the $r$th RIS and the $j$th AP, that between $j$th AP and $k$th user and the channel between $r$th RIS and $k$th user, the cascaded channel between AP $j$ and user $k$ can be written as $\mathbf{g}_{r,j,k} = \mathbf{H}_{r, j} \boldsymbol{\Phi}_r \mathbf{g}_{r, k}$. Furthermore, by denoting $\mathbf{q}_{r, j, k} = \mathbf{g}_{r, j, k} + \mathbf{h}_{j, k}$, the received signal at the $j$th AP is given by
 \begin{align} \label{rec}
     \mathbf{y}_j = \sqrt{\rho} \sum_{k = 1}^{K} \sum_{r = 1}^{R}  \mathbf{q}_{r, j, k} x_k + \mathbf{\check{n}}_j = \sum_{r = 1}^{R} \mathbf{Q}_{j, r} \mathbf{\check{x}} + \mathbf{\check{n}}_j,
 \end{align}
 where $\rho$ is the average transmit power of each user, $\mathbf{\check{x}} = \left[x_1, x_2, \dots, x_K   \right]^T$ are the transmitted symbols of all the $K$ users, and $\mathbf{\check{n}}_j \sim \mathcal{CN}(\mathbf{0}, \sigma^2 \mathbf{I}_N)$ is the additive white Gaussian noise (AWGN). 

 \subsection{Channel Model}
Since the users are located randomly and due to the existence of random blockages such as trees and buildings, the direct channel between APs and the users may be blocked. As in \cite{zhi2021two}, we adopt the Rayleigh fading model to describe the channel between AP $j$ and user $k$ as follows
\begin{equation}
    \mathbf{h}_{j, k} = \sqrt{\gamma_{j, k}} \Tilde{\mathbf{h}}_{j, k}, \hspace{3mm} j \in \mathcal{J}, k \in \mathcal{K},
    \label{D_channel}
\end{equation}
where $\gamma_{j, k}$ is the distance dependent path-loss of the AP $j$ and user $k$ link. $\Tilde{\mathbf{h}}_{j, k}$ is the fast fading non-line-of-sight (NLoS) component of the channel with entries that follow independent and identically distributed (i.i.d) complex Gaussian random variable with zero mean and unit variance. 

Considering that the RIS is commonly deployed at the top of high-rise buildings, it has a high probability for the RISs to have a line-of-sight (LoS) link to the APs and the users. Therefore, as in \cite{zhi2021two, han2019large}, we adopt the Rician fading model for the AP $j$-RIS $r$ and RIS $r$-user $k$ channels as follows
\begin{align} 
    \mathbf{H}_{r, j} &= \sqrt{\frac{\beta_{r, j}}{1 + \kappa_{r, j}}} \left( \sqrt{\kappa_{r, j}} \Bar{\mathbf{H}}_{r, j} + \Tilde{\mathbf{H}}_{r, j}   \right) , \hspace{8mm} r \in \mathcal{R}, j \in \mathcal{J}, \\ 
    \mathbf{g}_{r, k} &= \sqrt{\frac{\alpha_{r, k}}{1 + \epsilon_{r, k}}} \left( \sqrt{\epsilon_{r, k}} \Bar{\mathbf{g}}_{r, k} + \Tilde{\mathbf{g}}_{r, k}   \right) , \hspace{13mm} r \in \mathcal{R}, k \in \mathcal{K},
\end{align} 
where $\beta_{r, j}$ and $\alpha_{r, j}$ represent the path-loss coefficients between the $r$th RIS and the $j$th AP and that between the $r$th RIS and the $k$th user, respectively. $\kappa_{r, j}$ and $\epsilon_{r, j}$ are the Rician factors for the links the $r$th RIS and the $j$th AP between, and the $r$th RIS and the $k$th user, respectively. $\Bar{\mathbf{H}}_{r, j}$ and $\Bar{\mathbf{g}}_{r, k}$ denote the LoS components of the channels whereas $\Tilde{\mathbf{H}}_{r, j}$ and $\Tilde{\mathbf{g}}_{r, k}$ are the NLoS fast fading components of the channels between the $r$th RIS and the $j$th AP, and between the $r$th RIS and the $k$th user, respectively.

We assume a uniform squared planner array (USPA) model for the RISs and a uniform linear array (ULA) model for the APs. This is a reasonable assumption since for the cell-free massive MIMO systems, the number of antennas at the APs is significantly less than the number of antennas at the BS in the cellular systems. Thus, the antennas at the APs can be packed in a small area even with ULA structure \cite{wang2020uplink}. Hence, $\Bar{\mathbf{H}}_{r, j}$ and $\Bar{\mathbf{g}}_{r, k}$ can be modelled as follows
\begin{align}
    \Bar{\mathbf{H}}_{r, j} &= \mathbf{a}_N(\varphi_{r, j}) \mathbf{a}_M^H(\psi_{r, j}^a, \psi_{r, j}^e) \\
    \Bar{\mathbf{g}}_{r, k} &= \mathbf{a}_M(\phi_{r, k}^a, \phi_{r, k}^e),
\end{align}
where $\varphi_{r, j}$ is the azimuth angle-of-arrival (AoA) for the signal received at the AP~$j$ from RIS~$r$. $\psi_{r, j}^a$ ($\psi_{r, j}^e$) is the azimuth (elevation) angle-of-departure (AoD) from the $r$th RIS towards the $j$th AP. $\phi_{r, k}^a$ ($\phi_{r, k}^e$) is the azimuth (elevation) AoA to the $r$th RIS from the $k$th user. Furthermore, $\mathbf{a}_{X}(\vartheta^a, \vartheta^e) \in \mathbb{C}^{X \times 1}$ is the $x$th entry of the array response vector\cite{zhi2021two}
\begin{equation}
    [\mathbf{a}_{X}(\vartheta^a, \vartheta^e)]_x = \exp{\left( j \frac{2\pi d}{\lambda} \left( \lfloor (x-1)/\sqrt{X} \rfloor \sin \vartheta^a \sin \vartheta^e  + \left( (x-1)\mod\sqrt{X} \right) \cos \vartheta^e \right)   \right)},
    \label{arv}
\end{equation}
where $d$ and $\lambda$ denote the element spacing and the wavelength, respectively. Note that for the ULA-based model for the APs, the array response vector can be derived by setting $\vartheta^e = \frac{\pi}{2}$ in (\ref{arv}). Also, $\lfloor (x-1)/\sqrt{X} \rfloor$ should be replaced with $x$ for the ULA-based model. In order to simplify the notations, $\mathbf{a}_M(\phi_{r, k}^a, \phi_{r, k}^e)$, $\mathbf{a}_N(\varphi_{r, j})$ and $\mathbf{a}_M^H(\psi_{r, j}^a, \psi_{r, j}^e)$ will be written as $\Bar{\mathbf{g}}_{r, k}$, $\mathbf{a}_{N, r, j}$, and $\mathbf{a}_{M, r, j}$, respectively. Then,  the aggregated channel from user $k$ to AP $j$ will be written as
\begin{align} \label{channelss}
    \mathbf{q}_{r, j, k} &= \mathbf{g}_{r, j, k} + \mathbf{h}_{j, k} = \mathbf{H}_{r, j} \boldsymbol{\Phi}_r \mathbf{g}_{r, k} + \mathbf{h}_{j, k} \\ \nonumber
    &= \underbrace{\eta_{r, j, k}^{(1)} \Bar{\mathbf{H}}_{r, j}  \boldsymbol{\Phi}_r  \Bar{\mathbf{g}}_{r, k}}_{\mathbf{\check{g}}_{r, j, k}^{(1)}} +  \underbrace{\eta_{r, j, k}^{(2)} \Bar{\mathbf{H}}_{r, j} \boldsymbol{\Phi}_r \Tilde{\mathbf{g}}_{r, k}}_{\mathbf{\check{g}}_{r, j, k}^{(2)}} \\ \nonumber
    &+ \underbrace{\eta_{r, j, k}^{(3)} \Tilde{\mathbf{H}}_{r, j} \boldsymbol{\Phi}_r \Bar{\mathbf{g}}_{r, k}}_{\mathbf{\check{g}}_{r, j, k}^{(3)}} + \underbrace{\eta_{r, j, k}^{(4)} \Tilde{\mathbf{H}}_{r, j} \boldsymbol{\Phi}_r \Tilde{\mathbf{g}}_{r, k}}_{\mathbf{\check{g}}_{r, j, k}^{(4)}} + \eta_{j, k}^{(5)} \Tilde{\mathbf{h}}_{j, k} \\ \nonumber
    & = \mathbf{\check{g}}_{r, j, k} + \mathbf{h}_{j, k},
\end{align}
where $\mathbf{\check{g}}_{r, j, k} = \sum_{u = 1}^{4} \mathbf{\check{g}}_{r,  j,  k}^{(u)}$
and $\zeta_{r, j, k} \triangleq \frac{\beta_{r, j} \alpha_{r, k}}{(1 + \kappa_{r, j})(1 + \epsilon_{r, k})}$, $\eta_{r, j, k}^{(1)} \triangleq   \sqrt{\zeta_{r, j, k} \kappa_{r, j} \epsilon_{r, k}}$, $\eta_{r, j, k}^{(2)} \triangleq \sqrt{\zeta_{r, j, k} \kappa_{r, j}}$, $\eta_{r, j, k}^{(3)} \triangleq \sqrt{\zeta_{r, j, k} \epsilon_{r, k}}$,  $\eta_{r, j, k}^{(4)} \triangleq \sqrt{\zeta_{r, j, k}}$ and $\eta_{j, k}^{(5)} \triangleq \sqrt{\gamma_{j, k}}$.
Also, $\mathbf{\check{g}}_{r, j, k}$ and $\mathbf{h}_{j, k}$ are mutually independent.

\section{Channel Estimation}
We use the LMMSE method to estimate the aggregated channel. In this case, it is assumed that in the uplink, $\tau_p$ time slots are used for training and $\tau_u = \tau_c - \tau_p$ time slots dedicated for data transmission where $\tau_c$ is the length in which the channel response is constant. We use $\tau_p$ mutually orthogonal pilot sequences for channel estimation. Let us denote $\boldsymbol{\varrho}_k \in \mathbb{C}^{\tau_p \times 1}$ as the pilot sequence for the $k$th user with $\rVert \boldsymbol{\varrho}_k \rVert^2 = 1$. Furthermore, let us define $\mathcal{P}_k$ as the index subset of the users that use the same pilot sequences as user $k$ including itself. The received training signal at the $j$th AP is then given by
\begin{equation}
    \mathbf{Y}_j = \sqrt{\rho \tau_p} \sum_{k = 1}^{K} \sum_{r = 1}^{R}  \mathbf{q}_{r, j, k} \boldsymbol{\varrho}_k^H + \mathbf{\check{N}}_j,
\end{equation}
where $\mathbf{\check{N}}_j \in \mathbb{C}^{N \times \tau_p}$ is the additive white Gaussian noise matrix, each element of which follows the Gaussian distribution with zero mean and variance of $\sigma^2$. In order to estimate $\hat{\mathbf{q}}_{r, j, k}$, the AP $j$ multiplies the received signal with $\frac{\boldsymbol{\varrho}_k}{\sqrt{\rho \tau_p}}$ as follows
\begin{align} \label{LS}
    \mathbf{y}_{j, k} =  \mathbf{Y}_j \frac{\boldsymbol{\varrho}_k}{\sqrt{\rho \tau_p}} &= \sum_{r = 1}^{R}  \mathbf{q}_{r, j, k} +  \sum_{l \in \mathcal{P}_k \backslash \{k\}} \sum_{r = 1}^{R}  \mathbf{q}_{r, j, l}\boldsymbol{\varrho}_k^H \boldsymbol{\varrho}_l + \Tilde{\mathbf{n}}_j \\ 
    &= \mathbf{q}_{j, k} + \sum_{l \in \mathcal{P}_k \backslash \{k\}} \mathbf{q}_{j, l} + \Tilde{\mathbf{n}}_j,
    \label{y_m_k}
\end{align}
where $\mathbf{q}_{j, k} \triangleq  \sum_{r = 1}^{R}  \mathbf{q}_{r, j, k} = \sum_{r=1}^{R}\mathbf{\check{g}}_{r, j, k} +  \mathbf{h}_{j, k}$ and $\Tilde{\mathbf{n}}_j \triangleq \mathbf{\check{N}}_j \frac{\boldsymbol{\varrho}_k}{\sqrt{\rho \tau_p}}$.

\begin{thm}\label{lmmse}
The LMMSE estimation of the channel $\mathbf{q}_{j,  k}$ with the observation vector of $\mathbf{y}_{j,  k}$ is given by
\begin{align} \label{est_r1}
     \hat{\mathbf{q}}_{j,  k} &=  \mathbf{\check{g}}_{j,  k}^{(1)} + \sum_{u = 2}^{4} \mathbf{A}_{j,  k}\mathbf{\check{g}}_{j,  k}^{(u)} + \mathbf{A}_{j,  k} \mathbf{h}_{j,  k} + \sum_{l \in \mathcal{P}_k \backslash \{k\}} \mathbf{A}_{j,  k} \left( \sum_{u = 2}^{4}  \mathbf{\check{g}}_{j,  l}^{(u)} +  \mathbf{h}_{j,  l} \right) \\ \nonumber
    &+  \mathbf{A}_{j,  k} \Tilde{\mathbf{n}}_j,
\end{align}
where 
\begin{align}
     \mathbf{A}_{j,  k}  &= \Big(\sum_{r=1}^R   \xi_{r, j,  k} \mathbf{a}_{N, r, j} \mathbf{a}_{N, r, j}^H + \chi_{j,  k} \mathbf{I}_N \Big)\Big( \sum_{r=1}^R \omega_{r, j,  k} \mathbf{a}_{N, r, j} \mathbf{a}_{N, r, j}^H + \nu_{j,  k }  \mathbf{I}_N  \Big)^{-1},
\end{align}
and
\begin{align}
    \xi_{r, j,  k} &\triangleq M(\eta_{r, j, k}^{(2)})^2, \\
    \mu_{j,  k} &\triangleq \sum_{r=1}^R ( M ( (\eta_{r, j, k}^{(3)})^2 + (\eta_{r, j, k}^{(4)})^2 ) + (\eta_{j, k}^{(5)})^2 ), \\
     \delta_{j,  k, l} &\triangleq \sum_{r = 1}^R \eta_{r, j, k}^{(3)} \eta_{r, j, l}^{(3)} \mathrm{Tr}\{ \Bar{\mathbf{g}}_{r, k} \Bar{\mathbf{g}}_{r, l}^H  \}, \\
     \omega_{r, j,  k} &\triangleq \xi_{r, j,  k} + \sum_{l \in \mathcal{P}_k \backslash \{k\}} \xi_{r, j,  l}, \\
     \nu_{j,  k} &\triangleq  \mu_{j,  k} +  \sum_{l \in \mathcal{P}_k \backslash \{k\}}\mu_{j,  l} + 2 \sum_{l \in \mathcal{P}_k \backslash \{k\}} \delta_{j,  k, l} + \frac{\sigma^2}{\rho \tau_p}, \\
     \chi_{j,  k} &\triangleq \mu_{j,  k} + \sum_{l \in \mathcal{P}_k \backslash \{k\}}\delta_{j,  k, l}.
\end{align}

\end{thm}
\begin{proof}
See Appendix A.
\end{proof}
Theorem \ref{lmmse} shows that we only need to estimate the aggregated channel matrices including the reflected and the direct channel of the link between user $k$ and all the RISs to AP $j$. In this case, the estimated channel has the same dimension as the AP-user channel as in conventional cell-free massive MIMO systems without RIS.

\section{Uplink Data Transmission}
During the uplink data transmission, the received signal at the $j$th AP is given by 
\begin{equation}
    \mathbf{y}_j = \sqrt{\rho_s}\sum_{i = 1}^{K} \mathbf{q}_{j, i} x_i + \mathbf{\check{n}}_j,
    \label{received}
\end{equation}
where $\rho_s$ is the uplink data transmission power of the users. 
APs are all connected to the CPU through wired or wireless backhauls. The received signal can be processed locally at each AP or can be sent to the CPU for further processing. The benefit of centralized processing is the high computing capacity that is available at the CPU, enabling complex computations to be performed on the network. However, the downside is the limited overhead of sending data to the CPU via the backhaul. By contrast, the AP can process and detect the signal locally using its own estimated channel and then send it to the CPU to be aggregated with the detected signal from other APs to decode the data sent by each user. This scheme has the advantage of low signalling overhead since the CPU does not need to know the estimated channel for each AP. 
Please refer to \cite{bjornson2019making, bjornson2019cell} for further details. Thus, in this paper, it is assumed that the signal is processed locally at the APs and then passed to the CPU for final decoding. 

Let $\mathbf{v}_{j, k} \in \mathbb{C}^{N \times 1}$ be the local combining vector that the $j$th AP uses to decode the signal sent from the $k$ user, the the local estimate of $x_k$ at AP $j$ is represented by $\check{x}_{j, k}$ and given by
\begin{equation}
   \check{x}_{j, k} \triangleq  \mathbf{v}_{j, k}^H \mathbf{y}_j = \sqrt{\rho_s} \mathbf{v}_{j, k}^H \mathbf{q}_{j, k} x_k + \sqrt{\rho_s} \sum_{i = 1, i \neq k}^K {\mathbf{v}_{j, i}^H \mathbf{q}_{j, i} x_i} + \mathbf{v}_{j, i}^H \mathbf{\check{n}}_j.
   \label{local_est}
\end{equation}

The locally estimated signal $ \check{x}_{j, k}$ is then sent to the CPU to obtain the final decoding. In this case, the CPU, obtains a linear average using the deterministic weights $a_{j, k}$ to decode $\check{x}_{j, k}$. From (\ref{local_est}), we can have
\begin{align}
    \hat{x}_k = \sum_{j = 1}^{J}  a_{j, k} \check{x}_{j, k} =  \sum_{j = 1}^{J}  \sqrt{\rho_s} a_{j, k} \mathbf{v}_{j, k}^H \mathbf{q}_{j, k} x_k + \sum_{j = 1}^{J} \sum_{i = 1, i \neq k}^K \sqrt{\rho_s} a_{j, k} {\mathbf{v}_{j, k}^H \mathbf{q}_{j, i} x_i} + \underbrace{\sum_{j = 1}^{J} a_{j, k} \mathbf{v}_{j, k}^H \mathbf{\check{n}}_j}_{\triangleq n_k^\prime}.
    \label{CPU_rec}
\end{align}
With the assumption of $\mathbf{w}_{k, i} \triangleq \left[\mathbf{v}_{1, k}^H \mathbf{q}_{1, i}, \dots, \mathbf{v}_{J, k}^H \mathbf{q}_{J, i}  \right]^T $ and $\mathbf{a}_k \triangleq \left[ a_{1, k}, \dots,  a_{J, k}  \right]^T$, (\ref{CPU_rec}) can be re-written as
\begin{equation}
    \hat{x}_k = \mathbf{a}_k^T \mathbf{w}_{k, k} x_k + \sum_{i = 1, i \neq k}^K \mathbf{a}_k^T \mathbf{w}_{k, i} x_i + n_k^\prime.
\end{equation}
The weighting vector $\mathbf{a}_k$ can be optimized at the CPU in order to maximize the uplink SE (SE). The simplest option would be equal weight averaging, i.e., $\mathbf{a}_k = \left[1/K, \dots, 1/K \right]^T$.


\begin{prop}
Since the CPU only has access to the channel statistics, the SE for the $k$th user is given by
\begin{align} \label{SE}
     \mathrm{SE}_k = \left(1 - \frac{\tau_p}{\tau_c} \right) \log_2\left(1 + \mathrm{SINR}_k \right),
\end{align}
with the effective $\mathrm{SINR}$ given by
\begin{equation}
    \mathrm{SINR}_k = \frac{\rho |\mathbf{a}_k^H \mathbb{E}\{ \mathbf{w}_{k, k} \} |^2}{ \rho ( \mathbb{E}\{ |\mathbf{a}_k^H \mathbf{w}_{k, k}|^2 \} -  |\mathbf{a}_k^H \mathbb{E}\{ \mathbf{w}_{k, k}\}|^2) + \rho \sum_{i=1, i \neq k}^{K} \mathbb{E}\{ |\mathbf{a}_k^H \mathbf{w}_{k, i}|^2 \} + \sigma^2\mathbf{a}_k^H \mathbf{V}_k \mathbf{a}_k},
    \label{sinr}
\end{equation}
where $\mathbf{V}_k = \mathrm{diag}\left( \mathbb{E}\{\rVert \mathbf{v}_{1, k}\rVert^2\}, \dots, \mathbb{E}\{\rVert\mathbf{v}_{J, k}\rVert^2 \} \right)$.
\end{prop}
\begin{proof}
See \cite[Appendix A]{bjornson2019making}
\end{proof}

\begin{cor}
The effective $\mathrm{SINR}$ in (\ref{sinr}) is maximized by 
\begin{equation}
    \mathbf{a}_k = \left( \rho (\mathbb{E}\{ \mathbf{w}_{k, k} \mathbf{w}_{k, k}^H \}  -  \mathbb{E}\{ \mathbf{w}_{k, k} \} \mathbb{E}\{ \mathbf{w}_{k, k}^H \}) + \rho \sum_{i=1, i \neq k}^{K}  \mathbb{E}\{ \mathbf{w}_{k, i} \mathbf{w}_{k, i}^H \} + \sigma^2 \mathbf{V}_k \right)^{-1} \mathbb{E}\{ \mathbf{w}_{k, k} \},
\end{equation}
and the $\mathrm{SINR}$ becomes
\begin{equation}
     \mathrm{SINR}_k =  \mathbf{s}_k^H \mathbf{I}_{k, i}^{-1} \mathbf{s}_k,
     \label{SINR_here}
\end{equation}
where
\begin{align}
    \mathbf{s}_k &=  \sqrt{\rho} \mathbb{E}\{ \mathbf{w}_{k, k} \} \\
    \mathbf{I}_{k, i} &= \rho (\mathbb{E}\{ \mathbf{w}_{k, k} \mathbf{w}_{k, k}^H \}  -  \mathbb{E}\{ \mathbf{w}_{k, k} \} \mathbb{E}\{ \mathbf{w}_{k, k}^H \}) + \rho \sum_{i=1, i \neq k}^{K}  \mathbb{E}\{ \mathbf{w}_{k, i} \mathbf{w}_{k, i}^H \} + \sigma^2 \mathbf{V}_k
\end{align}
\end{cor}
\begin{proof}
It follows from \cite[Lemma B.10]{bjornson2017massive} by noticing that
 (\ref{sinr}) is a generalized Rayleigh quotient with respect to $\mathbf{a}_k$. 
\end{proof}

\begin{thm} \label{exp_thm}
With the following definitions 
\begin{align}
    g_{r, j} &\triangleq  \mathbf{a}_{N, r, j}^H \mathbf{A}_{j,  k} \mathbf{a}_{N, r, j} = g_{r, j}^H \\ 
    f_{r, j, k}(\boldsymbol{\Phi}_r) &\triangleq \mathbf{a}_{M, r, j}^H \boldsymbol{\Phi}_r \Bar{\mathbf{g}}_{r, k} \\
     h_{r,j,k} &\triangleq \mathbf{a}_{N, r, j}^H   \mathbf{A}_{j,  k}^H \mathbf{A}_{j,  k} \mathbf{a}_{N, r, j} \\ 
     k_{j,k, h,i}(\boldsymbol{\Phi}) &\triangleq \Bar{\mathbf{g}}_{r_1, k}^H \boldsymbol{\Phi}_{r_1}^H \Bar{\mathbf{H}}_{r_1, j}^H
  \Bar{\mathbf{H}}_{r_2, j}  \boldsymbol{\Phi}_{r_2}  \Bar{\mathbf{g}}_{r_2, i}
  \Bar{\mathbf{g}}_{r_3, i}^H \boldsymbol{\Phi}_{r_3}^H \Bar{\mathbf{H}}_{r_3, h}^H   
   \Bar{\mathbf{H}}_{r_4, h}  \boldsymbol{\Phi}_{r_4}  \Bar{\mathbf{g}}_{r_4, k} \\
   z_{j,k,h}({r_{1,2,3,1}}) &\triangleq \mathrm{Tr}\{\Bar{\mathbf{H}}_{r_1, j}^H \mathbf{A}_{j,  k}^H \mathbf{a}_{N, r_2, j}  \mathbf{a}_{N, r_3, h}^H \mathbf{A}_{h,  k} \Bar{\mathbf{H}}_{r_1, h} \} \\
   a_{j, k} &\triangleq \mathrm{Tr}\{ \mathbf{A}_{j,  k} \} \\
    v_{j, h, k} &\triangleq \mathrm{Tr}\{ \mathbf{A}_{j,  k} \} \mathrm{Tr}\{ \mathbf{A}_{h,  k} \} \\ 
    w_{j, k} &\triangleq \mathrm{Tr}\{ \mathbf{A}_{j,  k}^H \mathbf{A}_{j,  k}\} 
\end{align}
\begin{align}
     m_{j, k, h, i}({r_{x, y}}) &\triangleq (\mathbf{\check{g}}_{r_x, j, k}^{(1)})^H \mathbf{\check{g}}_{r_y, h, i}^{(1)}
\end{align}
\begin{align}
    b_{j, h}(r_{x, y}) &\triangleq \mathbf{a}_{N, r_x, j}^H \mathbf{a}_{N, r_y, h}
\end{align}
\begin{align}
    c_{r, k, i} &\triangleq \Bar{\mathbf{g}}_{r, k}^H  \Bar{\mathbf{g}}_{r, i}
\end{align}
\begin{align}
   d_{j,k,h,k}(r_{2,4}) &\triangleq \mathrm{Tr} \{  \Bar{\mathbf{H}}_{r_1, j}^H  \mathbf{A}_{j,  k}^H \Bar{\mathbf{H}}_{{r_2}, j}   \Bar{\mathbf{H}}_{{r_2}, h}^H 
  \mathbf{A}_{h,  k}  \Bar{\mathbf{H}}_{{r_4}, h}  \}
\end{align}
\begin{align}
s_{j, h}(r_{1, 2,4}) &\triangleq \mathbf{a}_{N, r_1, j}^H  \Bar{\mathbf{H}}_{r_2, j} \Bar{\mathbf{H}}_{r_2, h} ^H
 \mathbf{a}_{N, r_4, h}
\end{align}
\begin{align}
   h_{j,k}(r_{1, 2}) \mathrm{Tr}\{ \mathbf{A}_{j,  k}^H
     \mathbf{a}_{N, r_1, j} \mathbf{a}_{N, r_2, j}^H \mathbf{A}_{j,  k} \}
\end{align}
and by using the MRC detector and
\begin{align}\nonumber
    \mathbb{E} \{ \mathbf{w}_{k, k} \} &=  \mathbb{E}\left\{ \left[\hat{\mathbf{q}}_{1, k}^H \mathbf{q}_{1, k}, \dots, \hat{\mathbf{q}}_{j, k}^H \mathbf{q}_{j, k}, \dots, \hat{\mathbf{q}}_{J, k}^H \mathbf{q}_{J, k}  \right]^T \right\} \\ \label{}
    &=  \Big[\mathbb{E} \{\hat{\mathbf{q}}_{1, k}^H \mathbf{q}_{1, k}\}, \dots, \mathbb{E} \{\hat{\mathbf{q}}_{j, k}^H \mathbf{q}_{j, k}\}, \dots, \mathbb{E}\{\hat{\mathbf{q}}_{J, k}^H \mathbf{q}_{J, k}\}  \Big]^T, \label{Ewkk1} 
\end{align}
the $j$th entry of (\ref{Ewkk1}) could be written as follows
\begin{align} \nonumber
    \mathbb{E} \{\hat{\mathbf{q}}_{j, k}^H \mathbf{q}_{j, k}\} &= \sum_{r=1}^R  \Bigg(   \sum_{r^\prime=1}^R \Big(\eta_{r, j, k}^{(1)} \eta_{r^\prime, j, k}^{(1)} \Bar{\mathbf{g}}_{r, k}^H \boldsymbol{\Phi}_r^H \Bar{\mathbf{H}}_{r, j}^H \Bar{\mathbf{H}}_{r^\prime, j}  \boldsymbol{\Phi}_{r^\prime}  \Bar{\mathbf{g}}_{r^\prime, k} \Big) + (\eta_{r, j, k}^{(2)})^2 M g_{r, j}\\ \nonumber
     &+  (\eta_{r, j, k}^{(3)})^2 M \mathrm{Tr}\{ \mathbf{A}_{j,  k}^H \}+ (\eta_{r, j, k}^{(4)})^2 M \mathrm{Tr}\{ \mathbf{A}_{j,  k}^H \} + (\eta_{j, k}^{(5)})^2 \mathrm{Tr}\{ \mathbf{A}_{j,  k}^H \} \\ \label{exp_qhq} 
     &+ \sum_{l \in \mathcal{P}_k \backslash \{k\}} \Big( \eta_{r, j, k}^{(3)} \eta_{r, j, l}^{(3)}  \mathrm{Tr}\{ \mathbf{A}_{j,  k}^H \}\Bar{\mathbf{g}}_{r, l}^H  \Bar{\mathbf{g}}_{r, k} \Big) \Bigg).
\end{align}
Furthermore, by assuming
\begin{align}\label{matrixE}
    \mathbb{E} \{ \mathbf{w}_{k, i} \mathbf{w}_{k, i}^H \} &=   \begin{bmatrix}
        \mathbb{E}\{|\hat{\mathbf{q}}_{1, k}^H \mathbf{q}_{1, i}|^2\} & \dots &  \mathbb{E}\{\hat{\mathbf{q}}_{1, k}^H \mathbf{q}_{1, i} (\hat{\mathbf{q}}_{j, k}^H \mathbf{q}_{j, i})^H\} & \dots &   \mathbb{E}\{\hat{\mathbf{q}}_{1, k}^H \mathbf{q}_{1, i}  (\hat{\mathbf{q}}_{J, k}^H \mathbf{q}_{J, i})^H \}\\
        \vdots & \ddots \\
        \mathbb{E}\{\hat{\mathbf{q}}_{j, k}^H \mathbf{q}_{j, i}  (\hat{\mathbf{q}}_{1, k}^H \mathbf{q}_{1, i})^H\} & \dots & \mathbb{E}\{| \hat{\mathbf{q}}_{j, k}^H \mathbf{q}_{j, i}|^2\} & \dots &  \mathbb{E}\{\hat{\mathbf{q}}_{j, k}^H \mathbf{q}_{j, i}  (\hat{\mathbf{q}}_{J, k}^H \mathbf{q}_{J, i})^H\} \\
        \vdots & & & \ddots \\
        \mathbb{E}\{\hat{\mathbf{q}}_{J, k}^H \mathbf{q}_{J, i}  (\hat{\mathbf{q}}_{1, k}^H \mathbf{q}_{1, i})^H\} & \dots & \mathbb{E}\{\hat{\mathbf{q}}_{J, k}^H \mathbf{q}_{J, i}  (\hat{\mathbf{q}}_{j, k}^H \mathbf{q}_{j, i})^H\} & \dots & \mathbb{E}\{|\hat{\mathbf{q}}_{J, k}^H \mathbf{q}_{J, i}|^2\}
    \end{bmatrix},
\end{align}
the $(h, k)$th entry of (\ref{matrixE}) with $h \neq k$ is given by

\begin{align} \label{E_jkjihkhi}
     \mathbb{E}\{& \hat{\mathbf{q}}_{j, k}^H \mathbf{q}_{j, i}  (\hat{\mathbf{q}}_{h, k}^H \mathbf{q}_{h, i})^H\} = \sum_{r_1=1}^R\sum_{r_2=1}^R\sum_{r_3=1}^R\sum_{r_4=1}^R E_{j,k,h,i}^{\mathrm{CROSS-BS}}(\boldsymbol{\Phi}) + \sum_{l \in \mathcal{P}_k \backslash \{k\}} E_{j,l,h,i}^{\mathrm{PC-CROSS-BS}}(\boldsymbol{\Phi}),
\end{align}
where
\begin{align} \nonumber
    E_{j,k,h,i}^{\mathrm{CROSS-BS}}(\boldsymbol{\Phi}) &= \eta_{r_1, j, k}^{(1)} \eta_{r_2, j, i}^{(1)}  \eta_{r_3, h, i}^{(1)}   \eta_{r_4, h, k}^{(1)}
 k_{j,k, h,i}(\boldsymbol{\Phi}) \\\nonumber
   &+ \eta_{r_1, j, k}^{(2)} \eta_{r_2, j, i}^{(1)} \eta_{r_3, h, i}^{(1)} \eta_{r_1, h, k}^{(2)} f_{r_2, j, i}(\boldsymbol{\Phi}_{r_2}) f_{r_3, h, i}^H(\boldsymbol{\Phi}_{r_3})  z_{j,k,h}({r_{1,2,3,1}}) \\\nonumber
   &+ \eta_{r_1, j, k}^{(3)} \eta_{r_1, j, i}^{(3)}  \eta_{r_3, h, i}^{(1)}    \eta_{r_4, h, k}^{(1)}  a_{j, k} f_{r_4, h, k}(\boldsymbol{\Phi}_{r_4}) f_{r_3, h, i}^H(\boldsymbol{\Phi}_{r_3}) c_{r_1, k, i}
 b_{h, h}(r_{3, 4}) \\\nonumber
 &+ \eta_{r_1, j, k}^{(1)}  \eta_{r_2, j, i}^{(2)}  \eta_{r_2, h, i}^{(2)} \eta_{r_4, h, k}^{(1)} f_{r_1, j, k}^H(\boldsymbol{\Phi}_{r_1}) f_{r_4, h, k}(\boldsymbol{\Phi}_{r_4})
 s_{j, h}(r_{1, 2,4}) \\\nonumber
 &+ \eta_{r_1, j, k}^{(1)} \eta_{r_2, j, i}^{(1)}  \eta_{r_3, h, i}^{(3)} \eta_{r_3, h, k}^{(3)} a_{h, k} f_{r_1, j, k}^H(\boldsymbol{\Phi}_{r_1}) f_{r_2, j, i}(\boldsymbol{\Phi}_{r_2})
 b_{j, j}(r_{1, 2}) c_{r_3, i, k} \\\nonumber
   &+ \eta_{r_1, j, k}^{(3)} \eta_{r_1, j, i}^{(3)} \eta_{r_3, h, i}^{(3)} \eta_{r_3, h, k}^{(3)} v_{j, h, k} c_{r_1, k, i} c_{r_3, i, k}  + \eta_{{r_1}, j, k}^{(2)} \eta_{{r_2}, j, i}^{(2)} \eta_{{r_2}, h, i}^{(2)} \eta_{{r_1}, h, k}^{(2)} d_{j,k,h,k}(r_{2,4}) \\\nonumber
  &+ \eta_{{r_1}, j, k}^{(4)}  \eta_{{r_1}, j, i}^{(4)}  \eta_{{r_1}, h, i}^{(4)}  \eta_{{r_1}, h, k}^{(4)}M v_{j, h, k}  + \eta_{{r_1}, j, k}^{(3)}  \eta_{{r_1}, j, i}^{(4)} \eta_{{r_1}, h, i}^{(4)} \eta_{{r_1}, h, k}^{(3)} M  v_{j, h, k} \\ 
  &+ \eta_{{r_1}, j, k}^{(4)} \eta_{{r_1}, j, i}^{(3)} \eta_{{r_1}, h, i}^{(3)}  \eta_{{r_1}, h, k}^{(4)} M  v_{j, h, k},
\end{align}
in which $\boldsymbol{\Phi}$ is the set of all reflecting coefficients and $E_{j,l,h,i}^{\mathrm{PC-CROSS-BS}}(\boldsymbol{\Phi})$ is given by
\begin{align} \nonumber
    E_{j,l,h,i}^{\mathrm{PC-CROSS-BS}}(\boldsymbol{\Phi}) &= \eta_{r_1, j, l}^{(2)} \eta_{r_2, j, i}^{(1)}  \eta_{r_3, h, i}^{(1)}  \eta_{r_1, h, l}^{(2)} f_{r_2, j, i}(\boldsymbol{\Phi}_{r_2}) f_{r_3, h, i}^H(\boldsymbol{\Phi}_{r_3})  z_{j,k,h}({r_{1,2,3,1}}) \\ \nonumber
    &+ \eta_{r_1, j, l}^{(3)} \eta_{r_1, j, i}^{(3)}  \eta_{r_3, h, i}^{(1)} \eta_{r_4, h, k}^{(1)} a_{j, k} f_{r_4, h, k}(\boldsymbol{\Phi}_{r_4}) f_{r_3, h, i}^H(\boldsymbol{\Phi}_{r_3}) c_{r_1, l, i}
b_{h, h}(r_{3, 4}) \\\nonumber
    &+ \eta_{r_1, j, k}^{(1)} \eta_{r_2, j, i}^{(1)}  \eta_{r_3, h, i}^{(3)} \eta_{r_3, h, l}^{(3)} a_{h, k} f_{r_1, j, k}^H(\boldsymbol{\Phi}_{r_1}) f_{r_2, j, i}(\boldsymbol{\Phi}_{r_2})
b_{j, j}(r_{1, 2}) c_{r_3, i, l} \\\nonumber
   &+ \eta_{r_1, j, l}^{(3)} \eta_{r_1, j, i}^{(3)}  \eta_{r_3, h, i}^{(3)} \eta_{r_3, h, l}^{(3)} v_{j, h, k} c_{r_1, k, i} c_{r_3, i, k} + \eta_{{r_1}, j, l}^{(4)}  \eta_{{r_1}, j, i}^{(4)}  \eta_{{r_1}, h, i}^{(4)}  \eta_{{r_1}, h, l}^{(4)}M v_{j, h, k} \\
&+ \eta_{{r_1}, j, l}^{(3)}  \eta_{{r_1}, j, i}^{(4)} \eta_{{r_1}, h, i}^{(4)} \eta_{{r_1}, h, l}^{(3)} M  v_{j, h, k}  + \eta_{{r_1}, j, l}^{(4)} \eta_{{r_1}, j, i}^{(3)} \eta_{{r_1}, h, i}^{(3)}  \eta_{{r_1}, h, l}^{(4)} M
 v_{j, h, k}.
\end{align}
Also, the $j$th diagonal entry of (\ref{matrixE}) could be written as follows
\begin{align} \nonumber
      \mathbb{E}\{|\hat{\mathbf{q}}_{j, k}^H  \mathbf{q}_{j, i}|^2\} &=  E_{j, k, i}^{\mathrm{C}} +  \sum_{r_1=1}^R\sum_{r_2=1}^R\sum_{r_3=1}^R\sum_{r_4=1}^R \Big( E_{j,k,j,i}^{\mathrm{CROSS-BS}}(\boldsymbol{\Phi}) +E_{r, j, k, i}^{\mathrm{CROSS-USER}}(\boldsymbol{\Phi})+ E_{j, k, i}^{\mathrm{NOISE}}(\boldsymbol{\Phi}) +  \\\label{E_jkji}
      &+ \sum_{l \in \mathcal{P}_k \backslash \{k\}}   E_{j,l,j,i}^{\mathrm{PC-CROSS-BS}}(\boldsymbol{\Phi})  +  E_{j, l, k, i}^{\mathrm{CROSS-USER-PC}}(\boldsymbol{\Phi}) 
     \Big),
\end{align}
with 
\begin{align} \label{sub_1}
    E_{j, k, i}^{\mathrm{C}} = (\eta_{j, k}^{(5)} \eta_{j, i}^{(5)})^2  w_{j, k} + \frac{\sigma^2}{\rho \tau_p} (\eta_{ j, i}^{(5)})^2  w_{j, k} +  \sum_{l \in \mathcal{P}_k \backslash \{k\}} \Big( (\eta_{j, l}^{(5)} \eta_{j, i}^{(5)})^2   w_{j, k}  \Big),
\end{align}
\begin{align} \nonumber
E_{r, j, k, i}^{  \mathrm{CROSS-USER}}(&\boldsymbol{\Phi}) =  \eta_{r_1, j, k}^{(1)} (\eta_{r_2, j, i}^{(3)})^2  \eta_{r_4, j, k}^{(1)}  M f_{r_1, j, k}^H(\boldsymbol{\Phi}_r) f_{r_4, j, k}(\boldsymbol{\Phi}_{r_4}) b_{j, j}(r_{1, 4})   \\ \nonumber
      &+ \eta_{r_1, j, k}^{(1)} (\eta_{r_2, j, i}^{(4)})^2  \eta_{r_4, j, k}^{(1)} f_{r_1, j, k}^H(\boldsymbol{\Phi}_r) f_{r_4, j, k}(\boldsymbol{\Phi}_{r_4}) b_{j, j}(r_{1, 4}) \\ \nonumber
     &+ \eta_{r_1, j, k}^{(1)} (\eta_{j, i}^{(5)})^2  \eta_{r_4, j, k}^{(1)} f_{r_1, j, k}^H(\boldsymbol{\Phi}_r) f_{r_4, j, k}(\boldsymbol{\Phi}_{r_4})  b_{j, j}(r_{1, 4}) \\\nonumber
    &+ (\eta_{r_1, j, k}^{(2)} \eta_{r_2, j, i}^{(3)})^2   M^2 h_{r_1,j,k} +  (\eta_{r_1, j, k}^{(2)} \eta_{r_2, j, i}^{(4)})^2  M h_{r_1,j,k} + (\eta_{r_1, j, k}^{(2)} \eta_{j, i}^{(5)})^2 M h_{r_1,j,k} \\ \nonumber
    &+   (\eta_{r_1, j, k}^{(3)})^2 \eta_{r_2, j, i}^{(1)}  \eta_{r_3, j, i}^{(1)} M m_{j, i, j, i}({r_{3, 2}}) w_{j, k}  + (\eta_{r_1, j, k}^{(3)} \eta_{r_2, j, i}^{(2)})^2     M^2 h_{r_2, j, k} \\\nonumber
    &+ (\eta_{r_1, j, k}^{(3)} \eta_{j, i}^{(5)})^2 M  w_{j, k}  + (\eta_{r_1, j, k}^{(4)})^2 \eta_{r_2, j, i}^{(1)} \eta_{r_3, j, i}^{(1)}  M m_{j, i, j, i}({r_{3, 2}}) w_{j, k} \\\nonumber
     &+ (\eta_{j, k}^{(5)})^2   \eta_{r_2, j, i}^{(1)} \eta_{r_3, j, i}^{(1)} m_{j, i, j, i}({r_{3, 2}}) w_{j, k}  + (\eta_{r_1, j, k}^{(4)} \eta_{r_2, j, i}^{(2)})^2   M^2 h_{r_2,j,k} + (\eta_{r_1, j, k}^{(4)} \eta_{j, i}^{(5)})^2  M  w_{j, k} \\ 
     &+  (\eta_{j, k}^{(5)} \eta_{r_1, j, i}^{(2)})^2 M h_{r_1,j,k}    + (\eta_{j, k}^{(5)} \eta_{r_1, j, i}^{(3)})^2  M  w_{j, k} + (\eta_{j, k}^{(5)} \eta_{r_1, j, i}^{(4)})^2  w_{j, k},
\end{align}
and
\begin{align} \nonumber
    E_{r, j, k, i}^{\mathrm{NOISE}}(\boldsymbol{\Phi}) &= 
         \frac{\sigma^2}{\rho \tau_p} \eta_{r_1, j, i}^{(1)} \eta_{r_2, j, i}^{(1)} f_{r_1, j, i}(\boldsymbol{\Phi}_{r_1}) f_{r_2, j, i}^H(\boldsymbol{\Phi}_{r_2})   h_{j,k}(r_{1, 2}) \\ 
     &+ \frac{\sigma^2}{\rho \tau_p} (\eta_{r_1, j, i}^{(2)})^2 M h_{i, j} + \frac{\sigma^2}{\rho \tau_p} (\eta_{r_1, j, i}^{(3)})^2 M  w_{j, k}  + \frac{\sigma^2}{\rho \tau_p} (\eta_{r_1, j, i}^{(4)})^2   w_{j, k},
\end{align}
\begin{align} \nonumber
   E_{j, l, k, i}^{\mathrm{CROSS-USER-PC}}(&\boldsymbol{\Phi}) = 
    (\eta_{r_1, j, l}^{(2)}  \eta_{r_2, j, i}^{(3)})^2 M^2 h_{r_1,j,k}  + (\eta_{r_1, j, l}^{(2)}  \eta_{r_2, j, i}^{(4)})^2 M h_{r_1,j,k} + (\eta_{r_1, j, l}^{(2)} \eta_{j, i}^{(5)})^2  M h_{r_1,j,k} \\ \nonumber
     &+  (\eta_{r_1, j, l}^{(3)})^2 \eta_{r_2, j, i}^{(1)}  \eta_{r_3, j, i}^{(1)} M m_{j, i, j, i}({r_{3, 2}}) w_{j, k}   + (\eta_{r_1, j, l}^{(3)} \eta_{r_2, j, i}^{(2)})^2     M^2 h_{r_2, j, k} \\ \nonumber
     &+  (\eta_{r_1, j, l}^{(3)} \eta_{r_2, j, i}^{(3)})^2    M^2  w_{j, k}  +  (\eta_{r_1, j, l}^{(3)} \eta_{j, i}^{(5)})^2 M  w_{j, k} \\ \nonumber
    &+ (\eta_{r_1, j, l}^{(4)})^2 \eta_{r_2, j, i}^{(1)} \eta_{r_3, j, i}^{(1)}  M m_{j, i, j, i}({r_{3, 2}}) w_{j, k}  +  (\eta_{r_1, j, l}^{(4)} \eta_{r_2, j, i}^{(2)})^2   M^2 h_{r_2,j,k} \\ \nonumber
    &+  (\eta_{r_1, j, l}^{(4)} \eta_{r_2, j, i}^{(4)})^2  M^2  w_{j, k} + (\eta_{r_1, j, l}^{(4)} \eta_{j, i}^{(5)})^2  M  w_{j, k} + (\eta_{j, l}^{(5)})^2   \eta_{r_2, j, i}^{(1)} \eta_{r_3, j, i}^{(1)} m_{j, i, j, i}({r_{3, 2}}) w_{j, k} \\ 
    &+  (\eta_{j, l}^{(5)} \eta_{r_1, j, i}^{(2)})^2 M h_{r_1,j,k}   + (\eta_{j, l}^{(5)} \eta_{r_1, j, i}^{(3)})^2  M  w_{j, k}  + (\eta_{j, l}^{(5)} \eta_{r_1, j, i}^{(4)})^2  w_{j, k}.
\end{align}

 Moreover, the off-diagonal entries of $ \mathbb{E} \{ \mathbf{w}_{k, k} \mathbf{w}_{k, k}^H \}$ are given by
 
 \begin{align} \nonumber
    \mathbb{E}\{ \hat{\mathbf{q}}_{j, k}^H \mathbf{q}_{j, k}  (\hat{\mathbf{q}}_{h, k}^H \mathbf{q}_{h, k})^H\} &= \sum_{r_1=1}^R\sum_{r_2=1}^R\sum_{r_3=1}^R\sum_{r_4=1}^R E_{j,k,h,k}^{\mathrm{CROSS-BS}}(\boldsymbol{\Phi}) + E_{j,k,h}^{\mathrm{SELF-UE}}( \mathbf{r}) \\\label{term_jkjkhkhk}
    &+ \sum_{l \in \mathcal{P}_k \backslash \{k\}}  E_{j,l,h,k}^{\mathrm{PC-CROSS-BS}}(\boldsymbol{\Phi}) + E_{j,l,h,k}^{\mathrm{PC - SELF-UE}}( \mathbf{r}),
\end{align}
where
\begin{align} \nonumber
    E_{j,k,h}^{\mathrm{SELF - UE}}(& \mathbf{r}) = (\eta_{r_1, j, k}^{(2)})^2  \eta_{r_3, h, k}^{(1)} \eta_{r_4, h, k}^{(1)}
M g_{r_1, j} m_{h, k, h, k}({r_{3, 4}}) + \eta_{r_1, j, k}^{(1)}  \eta_{r_2, j, k}^{(1)}  (\eta_{r_3, h, k}^{(2)})^2 M g_{r_3, h} m_{j, k, j, k}({r_{1, 2}}) \\\nonumber
      &+ (\eta_{r_1, j, k}^{(4)})^2   \eta_{r_3, h, k}^{(1)} \eta_{r_4, h, k}^{(1)}
      M a_{j, k}
m_{h, k, h, k}({r_{3, 4}}) + \eta_{r_1, j, k}^{(1)}  \eta_{r_2, j, k}^{(1)} (\eta_{r_3, h, k}^{(4)})^2  M a_{h, k} m_{j, k, j, k}({r_{1, 2}})\\ \nonumber
      &+ (\eta_{j, k}^{(5)})^2 \eta_{r_3, h, k}^{(1)} \eta_{r_4, h, k}^{(1)}
      a_{j, k}
m_{h, k, h, k}({r_{3, 4}}) + \eta_{r_1, j, k}^{(1)}  \eta_{r_2, j, k}^{(1)}  (\eta_{h, k}^{(5)})^2 a_{h, k}
m_{j, k, j, k}({r_{1, 2}}) \\ \nonumber
      &+ (\eta_{j, k}^{(5)}  \eta_{h, k}^{(5)})^2
       v_{j, h, k}  + ( \eta_{r_1, j, k}^{(2)}   \eta_{r_3, h, k}^{(4)})^2 M^2 a_{h, k} g_{r_1, j} + (\eta_{r_1, j, k}^{(4)}  \eta_{r_3, h, k}^{(2)} )^2 M^2 a_{j, k} g_{r_3, j} \\\nonumber
        &+  (\eta_{j, k}^{(5)}  \eta_{r_3, h, k}^{(2)})^2  M g_{r_3, h} a_{j, k} + (\eta_{r_1, j, k}^{(2)} \eta_{h, k}^{(5)})^2 M g_{r_1, k} a_{h, k}  + ( \eta_{r_1, j, k}^{(4)} \eta_{h, k}^{(5)})^2 M v_{j, h, k} \\\nonumber
        &+ ( \eta_{j, k}^{(5)}   \eta_{r_3, h, k}^{(4)})^2 M v_{j, h, k} + (\eta_{j, k}^{(5)}   \eta_{r_3, h, k}^{(3)})^2 M v_{j, h, k} + (\eta_{r_1, j, k}^{(4)}  \eta_{r_3, h, k}^{(3)})^2 M^2
     v_{j, h, k} \\ 
     &+ (\eta_{r_1, j, k}^{(3)}  \eta_{r_3, h, k}^{(4)})^2 M^2
     v_{j, h, k} + (\eta_{r_1, j, k}^{(2)}  \eta_{r_3, h, k}^{(3)})^2
     M^2 a_{h, k} g_{r_1, j} + (\eta_{r_1, j, k}^{(3)}  \eta_{r_3, h, k}^{(2)})^2
     M^2 a_{j, k} g_{r_3, h},
\end{align}
where $\mathbf{r}$ means that the variable is a function of RISs but not the phase shift of the RISs. Furthermore, we have
\begin{align}  \nonumber
     E_{j,l,h,k}^{\mathrm{PC - SELF - UE}}(\mathbf{r}) &= (\eta_{j, k}^{(5)})^2  \eta_{r_3, h, k}^{(3)} \eta_{r_3, h, l}^{(3)} M v_{j, h, k}  + (\eta_{h, k}^{(5)})^2   \eta_{r_3, j, k}^{(3)} \eta_{r_3, j, l}^{(3)}  M v_{j, h, k} \\ \nonumber
     &+  (\eta_{r_1, j, k}^{(4)})^2  \eta_{r_3, h, k}^{(3)} \eta_{r_3, h, l}^{(3)}  M^2
    v_{j, h, k}  + \eta_{r_1, j, l}^{(3)} \eta_{r_1, j, k}^{(3)}   (\eta_{r_3, h, k}^{(4)})^2 M^2
     v_{j, h, k} \\ 
     &+ (\eta_{r_1, j, k}^{(2)})^2  \eta_{r_3, h, k}^{(3)} \eta_{r_3, h, l}^{(3)}
     M^2 a_{h, k} g_{r_1, j}  + \eta_{r_1, j, l}^{(3)} \eta_{r_1, j, k}^{(3)}  (\eta_{r_3, h, k}^{(2)})^2
     M^2 a_{j, k} g_{r_3, h}.
\end{align}
 Additionally, the diagonal entries of $ \mathbb{E} \{ \mathbf{w}_{k, k} \mathbf{w}_{k, k}^H \}$ are as follows
\begin{align} \nonumber
      \mathbb{E}\{&|\hat{\mathbf{q}}_{j, k}^H  \mathbf{q}_{j, k}|^2\} =  E_{j, k, j}^{\mathrm{C}} +  \sum_{r_1=1}^R\sum_{r_2=1}^R\sum_{r_3=1}^R\sum_{r_4=1}^R \Big( E_{j,k,j,k}^{\mathrm{CROSS-BS}}(\boldsymbol{\Phi}) +E_{r, j, k, k}^{\mathrm{CROSS-USER}}(\boldsymbol{\Phi})+ E_{j,k,k}^{\mathrm{SELF-UE}}(\mathbf{r}) \\\label{E_jkjk}
      &+E_{j, k, k}^{\mathrm{NOISE}}(\boldsymbol{\Phi})  + \sum_{l \in \mathcal{P}_k \backslash \{k\}}   E_{j,l,j,k}^{\mathrm{PC-CROSS-BS}}(\boldsymbol{\Phi})  +  E_{j, l, k, k}^{\mathrm{CROSS-USER-PC}}(\boldsymbol{\Phi})  +  E_{j,l,j,k}^{\mathrm{PC - SELF-UE}}(\mathbf{r}).
\end{align}
 The only remaining part to be calculated is $\mathbf{V}_k = \mathrm{diag}\left( \mathbb{E}\{\rVert \mathbf{v}_{1, k}\rVert^2\}, \dots, \mathbb{E}\{\rVert\mathbf{v}_{J, k}\rVert^2 \} \right)$. The $j$th entry of $\mathbf{V}_k$ could be written as
\begin{align}
     \mathbb{E}\{\rVert\mathbf{v}_{j, k}\rVert^2 \} &=  \mathbb{E}\{\rVert  \hat{\mathbf{q}}_{j,  k} \rVert^2 \} = \mathbb{E}\{ \hat{\mathbf{q}}_{j,  k}^H \hat{\mathbf{q}}_{j,  k} \}.
\end{align}
According to the orthogonality property of the LMMSE estimator, we have $\mathbb{E}\{ \hat{\mathbf{q}}_{j,  k}^H \hat{\mathbf{q}}_{j,  k} \} = \mathbb{E}\{ \hat{\mathbf{q}}_{j,  k}^H \mathbf{q}_{j,  k} \}$ where $\mathbb{E}\{ \hat{\mathbf{q}}_{j,  k}^H \mathbf{q}_{j,  k} \}$ is calculated in (\ref{Ewkk1}).

\end{thm}

\begin{proof}
See Appendix B.
\end{proof}

\section{Phase Shift Design of the RISs}
The purpose of this section is to optimize the phase shifts of the RISs in order to maximize the $\mathrm{SE}$ derived in Theorem \ref{exp_thm}. Unlike conventional designs of RISs based on instantaneous CSI, the SE in Theorem \ref{exp_thm} only depends on statistical CSI, so we do not have to update the phase shift design until the long-term CSI varies. This results in reduced computational complexity of the proposed optimization criterion. Thus, the optimization problem is formulated as maximizing the sum rate of all the users with respect to the phase shift of all the RIS panels. The problem can be formulated as follows
\begin{align} \label{opt_general}
     \max_{\boldsymbol{\Phi}_r}& \hspace{5mm} \sum_{k=1}^K\mathrm{SE}_k \\\nonumber
      \mathrm{s.t.}& \hspace{3mm} \theta_{r, m} \in [0, 2\pi), r \in \mathcal{R}, m = 1, \dots, M.
\end{align}

It is obvious that due to the high complexity of the optimization problem in (\ref{opt_general}), mathematical approaches fail to solve this problem. Alternatively, we use the deep reinforcement approach for solving the optimization problem in (\ref{opt_general}).

\subsection{Using Soft Actor-Critic to solve (\ref{opt_general})}
In this subsection, SAC is used to solve the optimization problem in (\ref{opt_general}). To this end, the SAC is introduced briefly. 

There is a significant problem with sample efficiency when it comes to some of the most successful RL algorithms of recent years, such as Trust Region Policy Optimization (TRPO) \cite{schulman2015trust}, Proximal Policy Optimization (PPO) \cite{schulman2017proximal}, and Asynchronous Actor-Critical Agents (A3C) \cite{mnih2016asynchronous}. This is due to the fact that they learn in an on-policy manner, which means they need completely new samples each time a policy is updated. Using experience replay buffers, Q-learning-based off-policy methods such as Deep Deterministic Policy Gradient (DDPG) \cite{lillicrap2015continuous} and Twin Delayed Deep Deterministic Policy Gradient (TD3PG) \cite{joshi2021twin} can efficiently learn from past samples. In order to ensure the convergence of these methods, however, they have to be tuned carefully because they are extremely sensitive to hyperparameters. The Soft Actor-Critic algorithm builds on the legacy of the latter type of algorithm and adds methods for addressing convergence brittleness.
\subsection{Background}
Consider an infinite-horizon discounted Markov decision process (MDP), defined by the tuple $(\mathcal{S}, \mathcal{A}, P, r, \gamma)$, where $\mathcal{S}$ is the finite set of states, $\mathcal{A}$ is the finite set of actions, $P: \mathcal{S} \times \mathcal{A} \times \mathcal{S} \rightarrow[0, \infty)$
 is the transition probability distribution of the next state $\mathbf{s}_{t+1} \in \mathcal{S}$ given the current state $\mathbf{s}_{t} \in \mathcal{S}$ and action $\mathbf{a}_t \in \mathcal{A}$, Furthermore, $r: \mathcal{S} \rightarrow \mathbb{R}$ is the reward function and $\gamma$ is discount factor of the reward $r$. 
 
SAC is an RL algorithm defined for continuous action spaces. The main feature of SAC is that it uses a modified RL objective function instead of only seeking to maximize the long-term discounted reward. SAC is one of the states of the art reinforcement learning (RL) algorithms of the family of maximum entropy RL \cite{haarnoja2018soft}. As mentioned, the aim of standard RL is to maximize the expected sum of the rewards $\sum_t \mathbb{E}_{(\mathbf{s}_t, \mathbf{a}_t) \sim \rho_\pi}[r(\mathbf{s}_t, \mathbf{a}_t)]$ \cite{sutton2018reinforcement}, where $\rho_\pi$ is the short form of $\rho_\pi(\mathbf{s}_t, \mathbf{a}_t)$ is the state-action marginal of
the trajectory distribution induced by a policy $\pi(\mathbf{a}_t | \mathbf{s}_t)$.
In SAC, the objective function is generalized to the maximum entropy objective which favors stochastic policies by augmenting the objective with the expected entropy of the policy which is
\begin{equation}
    J(\pi) = \sum_{t = 0}^{T} \mathbb{E}_{(\mathbf{s}_t, \mathbf{a}_t) \sim \rho_\pi}[r(\mathbf{s}_t, \mathbf{a}_t) + \alpha \mathcal{H}(\pi(. | \mathbf{s}_t))],
\end{equation}
where
\begin{equation}
    \mathcal{H}(\pi(. | \mathbf{s}_t)) = -\int_{a \in \mathcal{A}} \pi( a | \mathbf{s}_t) \log(\pi( a | \mathbf{s}_t)) da.
\end{equation}
Also, $\alpha$ is the temperature parameter and determines the importance of the entropy term against the reward term. Hence, $\alpha$ controls the stochasticity of the optimal policy. For the rest of this paper, we will omit the temperature explicitly, as it can be subsumed into the reward by scaling it by $\alpha^{-1}$. Note that maximum entropy RL gradually approaches the conventional RL as $\alpha \rightarrow 0$.

A number of benefits can be derived from this objective from both a conceptual and practical standpoint. First, the policy encourages broader exploration, while abandoning clearly unpromising directions. In addition, the policy captures a variety of near-optimal behavior modes. The policy commits the same probability mass to multiple actions in problem settings when multiple actions seem equally attractive.

SAC makes use of three separate neural networks (NNs). A state-value function $V$ parameterized by the NN weight vectors $\boldsymbol{\Psi}$, a soft Q-function $Q$ parameterized by the weights $\boldsymbol{\Theta}$ and a policy function $\pi$ parameterized by the weights $\boldsymbol{\Phi}$. The mentioned function approximations should be trained as follows.
\begin{enumerate}
    \item Value network: The value function is defined as follows
    \begin{align}
        V(\mathbf{s}_t) = \mathbb{E}_{\mathbf{a}_t \sim \pi} \{ Q(\mathbf{s}_t, \mathbf{a}_t) - \log \pi(\mathbf{a}_t \rvert \mathbf{s}_t) \},
    \end{align}
    where $Q(\mathbf{s}_t, \mathbf{a}_t)$ is the Q-function and will be defined in the next part. The value network
    should be trained by minimizing the following error
    \begin{equation} \label{value_new}
        J_V(\boldsymbol{\Psi}) = \mathbb{E}_{\mathbf{s}_t \sim \mathcal{D}} \left[ \frac{1}{2} \left( V_{\boldsymbol{\Psi}}(\mathbf{s}_t) - \mathbb{E}_{\mathbf{a}_t \sim \pi_{\boldsymbol{\Phi}}} \left[ Q_{\boldsymbol{\Theta}}(\mathbf{s}_t, \mathbf{a}_t) - \log \pi_{\boldsymbol{\Theta}}(\mathbf{a}_t \rvert \mathbf{s}_t)\right]  \right)^2 \right].
    \end{equation}
    The meaning of (\ref{value_new}) is that across all the states that are sampled from the replay buffer $\mathcal{D}$, the value network should be trained with respect to the squared difference between the prediction of the value network and the expected prediction of the Q-function plus the entropy of the policy network.
    Furthermore, the below approximation of the gradient of the $J_V(\boldsymbol{\Psi})$ is used to update the parameters of the $V$ function
    \begin{equation} \label{n_v}
        \hat{\nabla}_{\boldsymbol{\Psi}}J_V(\boldsymbol{\Psi}) =  \nabla_{\boldsymbol{\Psi}}V_{\boldsymbol{\Psi}}(\mathbf{s}_t) (V_{\boldsymbol{\Psi}}\left(\mathbf{s}_t) - Q_{\boldsymbol{\Theta}} + \log \pi_{\boldsymbol{\Phi}}(\mathbf{a}_t \rvert \mathbf{s}_t) \right).
    \end{equation}
\item Q-network: The Q-network is trained by minimizing the following error  
\begin{equation}
    J_{Q}(\boldsymbol{\Theta}) = \mathbb{E}_{(\mathbf{s}_t, \mathbf{a}_t) \sim \mathcal{D}} \left[ \frac{1}{2} \left(  Q_{\boldsymbol{\Theta}}(\mathbf{s}_t, \mathbf{a}_t) - \hat{Q}_{\boldsymbol{\Theta}}(\mathbf{s}_t, \mathbf{a}_t) \right)^2 \right],
\end{equation}
 where   
\begin{equation}
    \hat{Q}_{\boldsymbol{\Theta}}(\mathbf{s}_t, \mathbf{a}_t) = r(\mathbf{s}_t, \mathbf{a}_t) + \gamma \mathbb{E}_{\mathbf{s}_{t+1} \sim p} \left[ V_{\bar{\boldsymbol{\Psi}}}(\mathbf{s}_{t+1}) \right],
\end{equation}
which can be optimized with a stochastic gradient
\begin{equation} \label{grad_q}
    \hat{\nabla}_{\boldsymbol{\Theta}} J_{Q}(\boldsymbol{\Theta}) = \nabla_{\boldsymbol{\Theta}} Q_{\boldsymbol{\Theta}}(\mathbf{s}_t, \mathbf{a}_t) \left( Q_{\boldsymbol{\Theta}}(\mathbf{s}_t, \mathbf{a}_t) - r(\mathbf{s}_t, \mathbf{a}_t) - \gamma V_{\bar{\boldsymbol{\Psi}}}(\mathbf{s}_{t+1})  \right).
\end{equation}
In (\ref{grad_q}), target value network $V_{\bar{\boldsymbol{\Psi}}}$ is used for the update, where $\bar{\boldsymbol{\Psi}}$ can be an exponentially moving average of the network weights. This maintains training stability \cite{mnih2015human}.
\item Policy network: The policy network is trained by minimizing the following error
\begin{equation} \label{policy}
    J_{\pi}(\boldsymbol{\Phi}) = \mathbb{E}_{\mathbf{s}_t \sim \mathcal{D}} \left[ D_{\mathrm{KL}} \left( \pi_{\boldsymbol{\Phi}}(. \rvert \mathbf{s}_t) \Bigg{\rVert} \frac{\exp({Q_{\boldsymbol{\Theta}}(\mathbf{s}_t, .))}}{Z_{\boldsymbol{\Theta}}(\mathbf
    {s}_t)}\right) \right],
\end{equation}
where $D_\mathrm{KL}(P \rVert Q)$ is the Kullback–Leibler (KL) divergence between two probability distributions $P$ and $Q$ and is defined as follows \cite{cover1999elements}.
\begin{equation}
    D_\mathrm{KL}(P \rVert Q) = \int_{-\infty}^{+\infty} p(x) \log \left( \frac{p(x)}{q(x)} \right) dx,
\end{equation}
where $p$ and $q$ denote the  probability densities of $P$ and $Q$, respectively. Note that the KL divergence quantifies how much one probability distribution differs from another probability distribution. If two distributions perfectly match, the KL divergence would be $0$, otherwise it can take values between $0$ and $\infty$.

Hence, the aim of the objective function (\ref{policy}) is to make the distribution of the policy function like the distribution of the exponentiation of the Q-function normalized by another function $Z$. 

In order to minimize the objective function, the authors of \cite{haarnoja2018soft} use a reparameterization trick. In this method, the error back-propagation is ensured by making the sampling process differentiable from the policy. The policy is parameterized as follows:
\begin{equation}
    \mathbf{a}_t = f_{\mathbf{\Phi}}(\epsilon_t; \mathbf{s}_t),
\end{equation}
where the epsilon term is a noise vector sampled from a Gaussian distribution.
In this case, the objective function could be written as follows
\begin{equation} \label{trick}
    J_{\pi}(\boldsymbol{\Phi}) = \mathbb{E}_{\mathbf{s}_t \sim \mathcal{D}, \epsilon_t \sim \mathcal{N}(0, \bar{\sigma}_\epsilon)} \left[  \log \pi_{\boldsymbol{\Phi}}(f_{\mathbf{\Phi}}(\epsilon_t; \mathbf{s}_t) \rvert \mathbf{s}_t) - Q_{\boldsymbol{\Theta}}(\mathbf{s}_t, f_{\mathbf{\Phi}}(\epsilon_t; \mathbf{s}_t)) \right].
\end{equation}
The normalization function $Z$ is removed since it is independent of the parameters $\boldsymbol{\Phi}$

The unbiased estimator of the for the gradient of the (\ref{trick}) is given by
\begin{align} \label{n_pi}
    \hat{\nabla}_{\boldsymbol{\Phi}} J_{\pi}(\boldsymbol{\Phi}) &= \nabla_{\boldsymbol{\Phi}} \log \pi_{\boldsymbol{\Phi}}(\mathbf{a}_t \rvert \mathbf{s}_t) + \left( \nabla_{\mathbf{a}_t} \log \pi_{\boldsymbol{\Phi}}(\mathbf{a}_t \rvert \mathbf{s}_t) - \nabla_{\mathbf{a}_t} Q(\mathbf{s}_t, \mathbf{a}_t)   \right) \nabla_{\boldsymbol{\Phi}}f_{\mathbf{\Phi}}(\epsilon_t; \mathbf{s}_t),
\end{align}
where $\mathbf{a}_t$ is evaluated at $f_{\mathbf{\Phi}}(\epsilon_t; \mathbf{s}_t)$.

Note that this algorithm uses two Q-functions to mitigate positive bias in the policy improvement step, which is known to degrade the value-based method's performance. Thus, two Q-functions are parameterized and trained independently and then the minimum Q-function for the value gradient in (\ref{n_v}) and the policy gradient in (\ref{n_pi}) is used.
\end{enumerate}
 \subsection{ Using SAC to solve (\ref{opt_general})}
Generally, SAC is presented as an MDP with observation and action spaces. In the two-time scale RIS design problem, the RISs, and all the UEs in the system are denoted by the environment $\mathcal{E}$, while the agent is the CPU that is able to control the RISs. The following are the key SAC elements employed to solve the optimization  problem (\ref{opt_general}).
\subsubsection{Observation space}
At each timestep $t$, the observation space consists of the phases of all the RIS elements, i.e., $\theta_{r, m}$ where $r \in \mathcal{R}, m \in \mathcal{M}$. The second part of the observation space is the phase part of the expected value of the cascaded channel, i.e., $\tan^{-1} \frac{\mathrm{Im} \{ \mathbf{\check{g}}_{j,  k}^{(1)} \}}{\mathrm{Re} \{ \mathbf{\check{g}}_{j,  k}^{(1)} \}}$ of all the users $j \in \mathcal{J}$ and all the APs $k \in \mathcal{K}$. Hence, the observation shape is $ RM + J K N$.
\subsubsection{Action space}
At each timestep, $t$ the action space is the vector containing the phase parts of the phase shifts of the RISs. Thus, the action shape is $RM$ and the action range is $[0, 2\pi)$. Since $\mathrm{tanh}(.)$ is the activation function for the final layer, which produces the values between $-1$ and $+1$, for converting the result to the desired action range, it is sufficient to set $\mathbf{a}_t = \pi(\mathbf{a}_t^{\prime}+1)$ where $\mathbf{a}_t^{\prime}$ is the output of $\mathrm{tanh}(.)$ activation layer. 
\subsubsection{Reward function} At each timestep $t$ the reward is the sum SE of all the users, i.e.,
$
    r_t = \sum_{k=1}^K \mathrm{SE}_u(t)
$
where $\mathrm{SE}_u(t)$ is the rate of $k$-th user at time step $t$.

The details of the proposed SAC algorithm for the two-time scale passive beamforming are presented in Algorithm~\ref{alg_1}.

\begin{algorithm}[h]
\caption{Soft Actor-Critic}
\label{alg_1}
\SetAlgoLined
 \textbf{Initialization:} Initialize  time, states, actions, and replay buffer $\mathcal{D}$ for storing the random states, action and reward in each time step. 
Initialize parameter vectors $\boldsymbol{\Psi}$, $\bar{\boldsymbol{\Psi}}$, $\boldsymbol{\Theta}$, $\boldsymbol{\Phi}$\;
 \For{\text{each iteration}}{
  Initialize the environment $\mathcal{E}$ and make the initial state $s_0$ \;
  \For{\text{each environment step}}{
  Select the action $\mathbf{a}_t \sim \pi_{\boldsymbol{\Phi}}(\mathbf{s}_t \rvert \mathbf{a}_t)$; \;
  
    Observe next state $s_{t+1}$ and reward $r_t$, then, store transition set
    $\{\mathbf{s}_t, \mathbf{a}_t, r_t, \mathbf{s}_{t+1} \}$ into $\mathcal{D}$; \;
  
  }
  
  \For{\text{each gradient step}}{
  $\boldsymbol{\Psi} \leftarrow \boldsymbol{\Psi} - \lambda_{V} \hat{\nabla}_{\boldsymbol{\Psi}}J_V(\boldsymbol{\Psi})$ where $\hat{\nabla}_{\boldsymbol{\Psi}}J_V(\boldsymbol{\Psi}) $ is obtained in (\ref{n_v});
  
  $\boldsymbol{\Theta}_i \leftarrow \boldsymbol{\Theta}_i - \lambda_{Q} \hat{\nabla}_{\boldsymbol{\Theta}_i} J_{Q}(\boldsymbol{\Theta}_i)$ for $i \in \{1, 2 \}$ where $\hat{\nabla}_{\boldsymbol{\Theta}_i} J_{Q}(\boldsymbol{\Theta}_i)$ is obtained in (\ref{grad_q});
  
  $\boldsymbol{\Phi} \leftarrow \boldsymbol{\Phi} - \lambda_{\pi}  \hat{\nabla}_{\boldsymbol{\Phi}} J_{\pi}(\boldsymbol{\Phi})$ where $ \hat{\nabla}_{\boldsymbol{\Phi}} J_{\pi}(\boldsymbol{\Phi})$ is obtained in (\ref{n_pi});
  
  $\bar{\boldsymbol{\Psi}} \leftarrow \bar{\tau}\bar{\boldsymbol{\Psi}} + (1-\bar{\tau}\bar{\boldsymbol{\Psi}})$;
  }
 }
\end{algorithm}


\section{Simulation Results}
In this section, the simulation results are provided to validate the effectiveness of the two-time scale approach. We consider a cell-free network like the one shown in Fig.~\ref{fig:sim_setup}. In this setup, $J = 3$ APs are considered each of which is equipped with $N = 4$ antennas communicating with $J = 3$ UEs. The LoS link between the APs and the UEs is blocked and hence the transmission is done with the assistance of $R = 3$ RISs each of which is equipped with $M = 30$ elements. All the simulation parameters are listed in Table~\ref{table:para} and Table~\ref{table:rl}\footnote{The source code of the paper will be published online when the paper gets accepted.}. For generating the channels, the large-scale path loss is calculated as $\gamma_{j, k} = 10^{-3} (d_{j, k})^{- \bar{\gamma} }$, $ \beta_{r, j} = 10^{-3} (d_{r, j})^{- \bar{\beta} }$, and $\alpha_{r, k} = 10^{-3} (d_{r, k})^{- \bar{\alpha} }$ where $d_{j, k}$ denotes the distance between the $j$-th AP and the $k$-th UE, $d_{r, j}$ is the distance between the $r$-th RIS and the $j$th AP and $d_{r, k}$ represents the distance between the $r$th RIS and the UE. Moreover, the azimuth and elevation AoA and AoD
of all the links are generated randomly from $[0, 2\pi)$ and $[0, \pi)$, respectively.

\begin{figure}[t]
    \centering
    \includegraphics[scale=0.14]{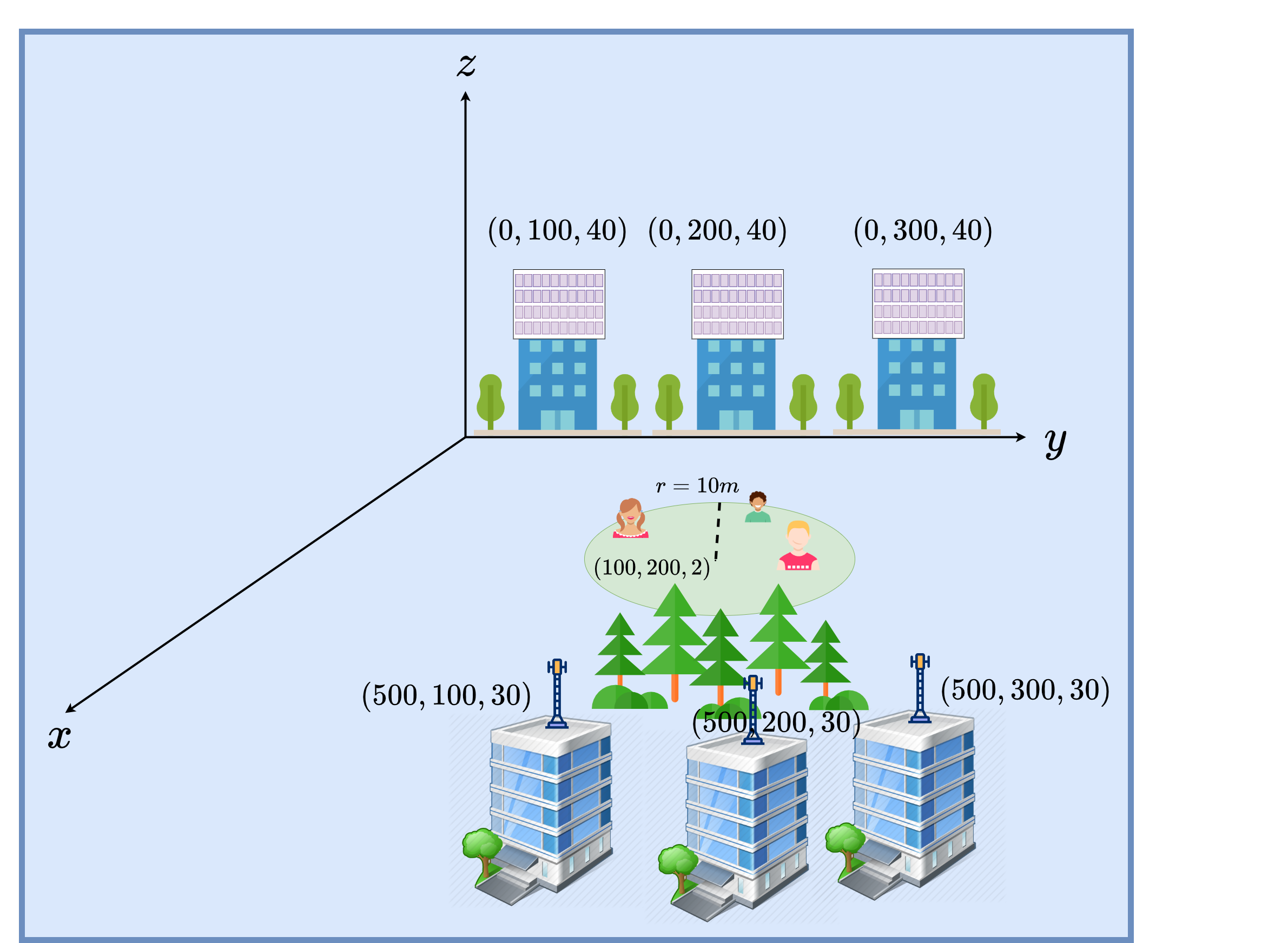}
    \caption{Simulation setup}
    \label{fig:sim_setup}
\end{figure}

\begin{table}[t]
\caption{Simulation parameters} 
\centering 
\begin{tabular}{c c} 
\hline\hline 
Parameter & Value \\ [0.5ex] 
\hline 
Number of APs & $J = 3$  \\
Number of antennas of each AP & $N = 4$ \\
Number of users & $K = 4$\\
Number of RISs & $R = 3$ \\
Number of elements of each RIS & $M = 30$ \\
Coherence time duration & $\tau_c = 200$ \\
pilot training duration & $\tau_p = 4$ \\
Rician factors  & $\epsilon_{r, k} = \kappa_{r, j} = 10 , \forall r,j,k$ \\
  Noise power & $\sigma^2 = -104 \mathrm{dBm}$ \\
  AP-UE link path loss exponent & $ \bar{\gamma} = 4$ \\
AP-RIS link path loss exponent & $ \bar{\beta} = 2.5 $ \\
RIS-UE link path loss exponent & $ \bar{\alpha} = 2$ \\
Transmitted power & $\rho = 0\mathrm{dBm}$ \\
  [1ex]  
\hline  
\end{tabular}
\label{table:para} 
\end{table}

\begin{table}[ht]
\caption{Reinforcement Learning  Hyperparameters} 
\centering 
\begin{tabular}{l r } 
\hline\hline 
Parameter & Value \\ [0.5ex] 
\hline 
Discount factor $\gamma$ & $0.995$  \\
Number of first fully connected layer & $128$ \\
 Number of second fully connected layer & $128$ \\
Learning rate & $0.0003$ \\
Replay buffer size & $4000$ \\
batch size & $32$ \\
Maximum time-steps of each episode & $50$ \\
Number of episodes & $2000$ \\
Activation function for hidden layers & ReLU \\
Activation function for output layer & $tanh(.)$ \\
target smoothing coefficient $\bar{\tau}$ & $0.005$ \\
Reward scale factor $\alpha$ & $2$ \\
gradient steps & $1$ \\
optimizer & Adam \\
 [1ex]
\hline 
\end{tabular}
\label{table:rl} 
\end{table}

\subsection{Quality of LMMSE Channel Estimation}
First of all, we investigate the performance of the LMMSE channel estimation scheme. The 
normalized mean squared error ($\mathrm{NMSE}$) for the $k$th user is defined as follows
\begin{equation}
    \mathrm{NMSE}_k = \frac{\mathrm{Tr} \{\mathrm{Cov}\{ \mathbf{q}_{j, k} - \hat{\mathbf{q}}_{j, k}, \mathbf{q}_{j, k} - \hat{\mathbf{q}}_{j, k} \}\}}{\mathrm{Tr} \{\mathrm{Cov}\{ \mathbf{q}_{j, k}, \mathbf{q}_{j, k} \}\}}
\end{equation}

Fig.~\ref{fig:nmse} shows the average $\mathrm{NMSE}$ of all the users as a function of the number of RIS elements. It is observed from  Fig.~\ref{fig:nmse} that the $\mathrm{NMSE}$ is a decreasing function with respect to the number of elements and thus, by increasing the number of RIS elements, the $\mathrm{NMSE}$ tends to zero. This shows the effectiveness of the RIS in enhancing channel estimation efficiency. Furthermore, in a pure LoS scenario, i.e., when $\kappa_{r, j} = \epsilon_{r, k} \rightarrow \infty$, the $\mathrm{NMSE}$ is independent of the number of RIS elements. This is because, in this case, all the channels are deterministic and deploying more elements on the RIS will not introduce additional errors. Moreover, for the case with $\tau_p = 4, K = 5$, i.e., when two users are sharing the same pilot sequence, the $\mathrm{NMSE}$ won't decrease with the number of the RIS elements. 

\begin{figure}[t]
    \centering
    \includegraphics[scale=0.78]{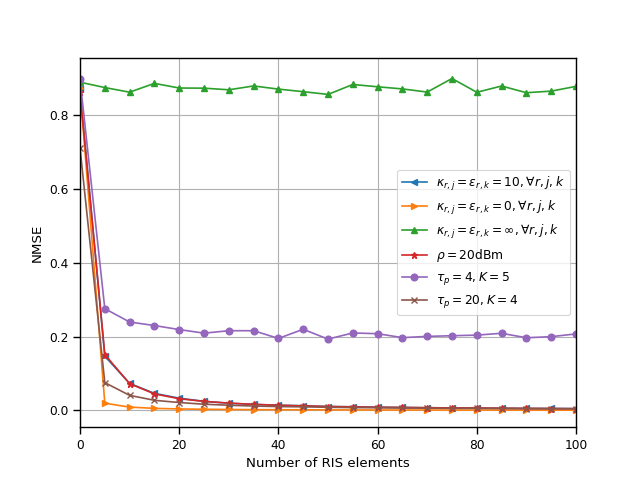}
    \caption{Average $\mathrm{NMSE}$ of all the users versus the number of RIS elements}
    \label{fig:nmse}
\end{figure}

\subsection{The effect of the number of RIS elements}

\begin{figure}[t]
    \centering
    \includegraphics[scale=0.60]{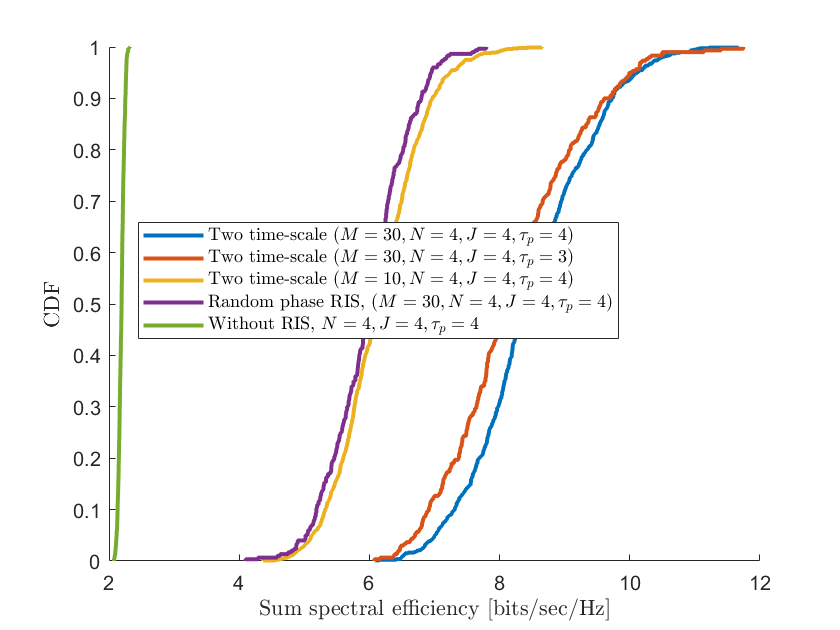}
    \caption{The cumulative distribution of the uplink sum SE for the different numbers of RIS elements.}
    \label{fig:cdf_ris}
\end{figure}

Fig.~\ref{fig:cdf_ris} compares the
cumulative distribution of the sum SE of the users for the cases with $M = 30$ and $M = 10$. For both cases, it is revealed that deploying the RIS panels causes a significant improvement in both median and 95\%-likely performance. Furthermore, the sum SE of the case without RIS is much more
concentrated around its median, compared with either case with optimized RIS phase or random RIS phase. On the other hand, for both cases with $M = 30$ and $M = 10$, UEs with better channel conditions
get better performance with two-time scale methods. Moreover, due to the pilot contamination effect, the sum SE of the users drops in comparison with the case in which there is no pilot contamination effect, yet, its performance is still significantly better than the random phase RIS case. 
\begin{figure}[t]
    \centering
\includegraphics[scale=0.60]{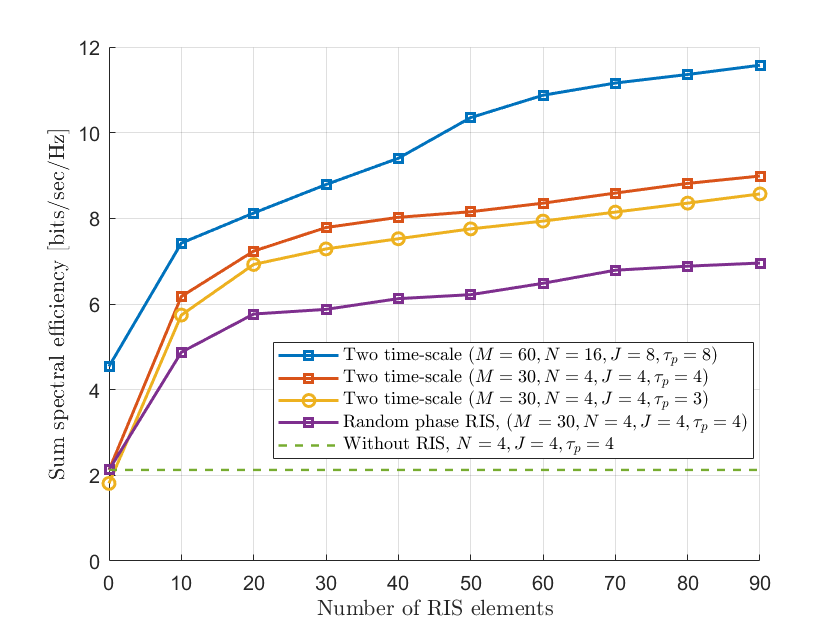}
    \caption{Sum SE of all the users versus the number of RIS elements}
    \label{fig:ris_element}
\end{figure}

In Fig.~\ref{fig:ris_element}, the performance of the two-time scale scheme versus the number of RIS elements is presented.
As revealed from Fig.~\ref{fig:ris_element}, increasing the number of elements results in an improvement of the sum SE of the users. 
Also, in the case with random phase RISs, performance is significantly better than the case without RIS, however, its performance is still very poor in comparison with the optimized RIS phases.

\subsection{The effect of the number of AP antennas}
The effect of the number of AP antennas is studied in this subsection. For the simulations, the number of RIS elements of all the panels is set to be $M = 60$ and the remaining parameters are set unchanged.

\begin{figure}[t]
    \centering
\includegraphics[scale=0.60]{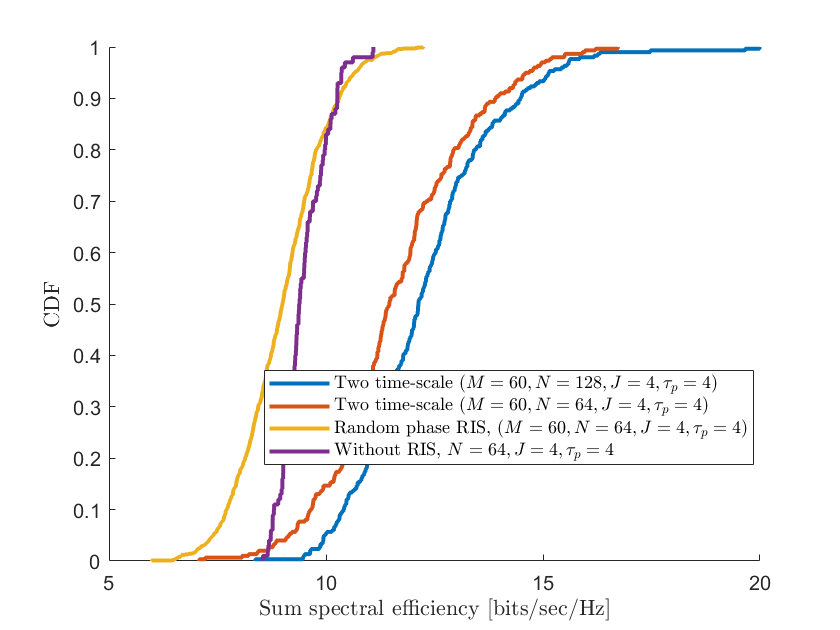}
    \caption{The cumulative distribution of the uplink sum SE for the different number of AP antennas.}
    \label{fig:cdf_ap}
\end{figure}

Fig.~\ref{fig:cdf_ap} demonstrates the cumulative distribution of the considered cell-free system with $N = 128$ and $N = 64$ number of antennas. It is seen from Fig.~\ref{fig:cdf_ap} that by increasing the number of AP antennas, the sum data rate of the users improved significantly. Also, in terms of median and in 95\%-likely, the performance of the proposed two-timescale scheme significantly outperforms the cases with random phase and without RIS. Also, similar to Fig.~\ref{fig:cdf_ris}, the optimized phase shift case has more variance than the others which shows the users with higher channel gain are highly likely to achieve considerably higher SE. Moreover, the performance of the case without RIS outperforms the random phase RIS case which shows the importance of the careful design of the RIS phase shifts. 

\begin{figure}[t]
    \centering
\includegraphics[scale=0.60]{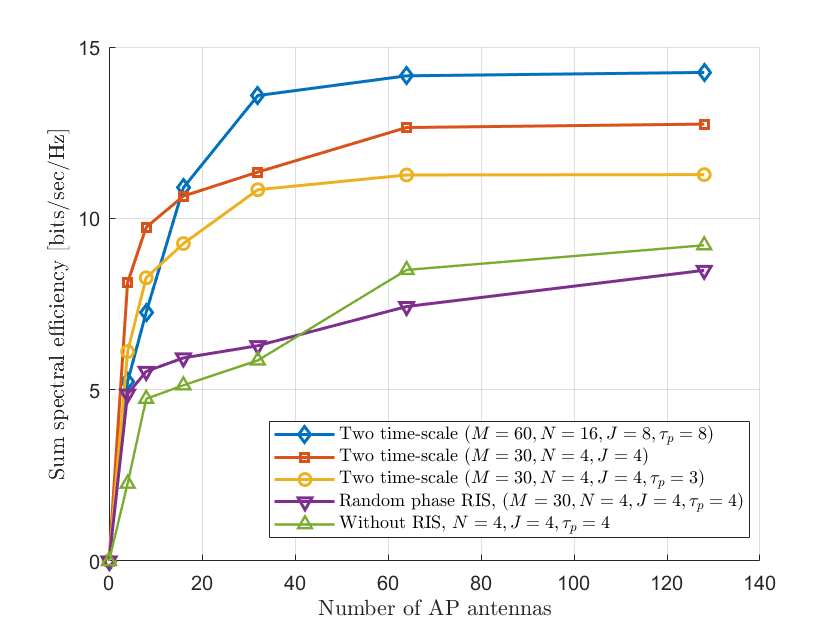}
    \caption{Sum SE of all the users versus the number of AP antennas.}
    \label{fig:ap_element_compare}
\end{figure}
Fig.~\ref{fig:ap_element_compare} shows the performance of the proposed two-timescale design as a function of the number of AP antennas. As it is observed from Fig.~\ref{fig:ap_element_compare}, increasing the number of AP elements will result in considerable sum SE improvement in the system, however, the sum SE will not increase without bound, since the interference of the other users will also increase as well. Lastly, it is noted that, when the number of AP antennas is high, the sum data rate of the case without RIS outperforms the case with random phase RIS, however, the optimized two-timescale RIS design performs better even in a high AP antenna regime.

\subsection{The effect of the transmit power}
The effect of the transmission power on the system performance is studied in this subsection.
In this scenario, the number of RIS elements is set to be $M = 30$ and the number of AP antennas is assumed to be $N = 4$. The rest of the parameters are the same as before.

\begin{figure}[t]
    \centering
\includegraphics[scale=0.60]{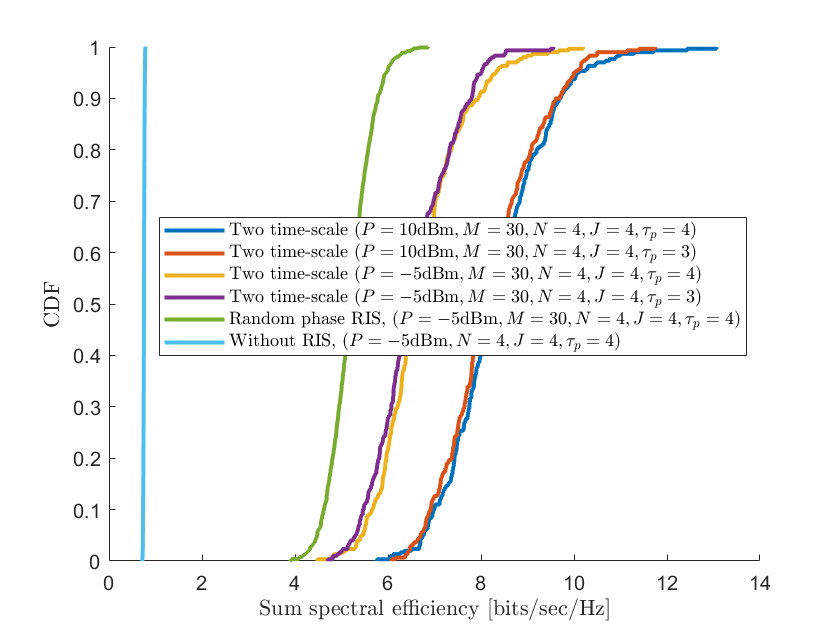}
    \caption{The cumulative distribution of the cell-free RIS-aided system for different values of transmit powers}
    \label{fig:cdf_tr}
\end{figure}

Fig.~\ref{fig:cdf_tr} shows the CDF of the system model under two different transmit powers of $\rho = -5 \mathrm{dBm}$ and $\rho = 10\mathrm{dBm}$. 
It is revealed from Fig.~\ref{fig:cdf_tr} that in the low power regime, the case without RIS panel, achieves a fairly low rate both in 95\% likely and median performance metrics. Also, as in the same previous scenarios, the case with optimized RIS phases has a significant improvement in terms of sum SE.

\begin{figure}[t]
    \centering
\includegraphics[scale=0.60]{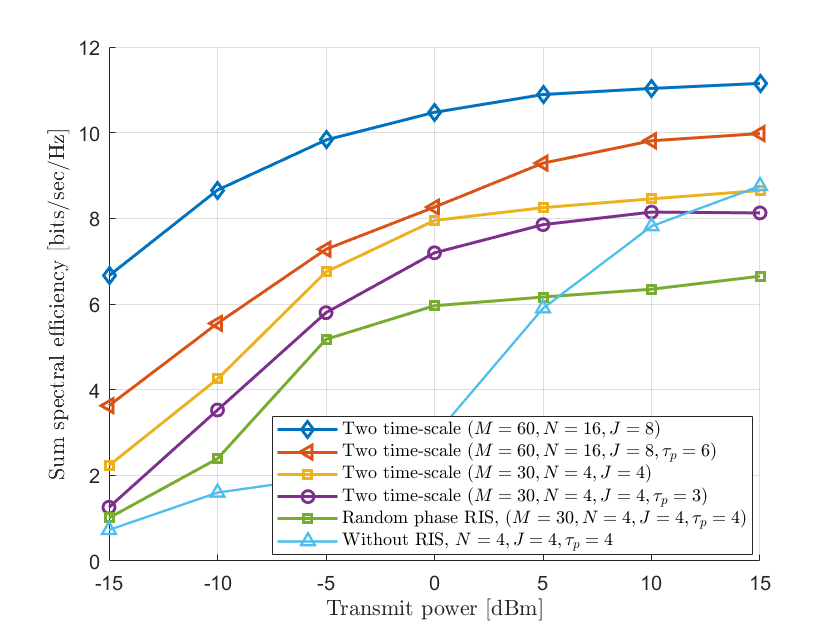}
    \caption{Average sum-rate of all the users versus the transmission power}
    \label{fig:tr_compare}
\end{figure}

In Fig.~\ref{fig:tr_compare}, the sum SE is plotted as a function of transmission power. It reveals from Fig.~\ref{fig:tr_compare} by increasing the transmit power, the sum SE of the users increases as well. Additionally, for the case without RIS, increasing the power will result in a larger SE than the optimized phase. Due to multi-user interference introduced by RIS, as power increases, the sum SE can no longer grow unbounded.

\subsection{Comparison to fully centralized and single-cell
scenarios}
A fully centralized scenario is the most advanced case where each AP sends their received data as well as the pilot signals to the CPU for detection. In this case, the CPU is responsible for channel estimation and final detection. Also, the CPU does the channel estimation based on the received pilot at the APs. In this scenario, the received signal at the CPU is given by
\begin{equation}
   \underbrace{\begin{bmatrix}
\mathbf{y}_1  \\
\vdots \\
\mathbf{y}_j \\
\vdots \\
\mathbf{y}_J
\end{bmatrix}}_{\triangleq  \mathbf{y}}
 = \sqrt{\rho_s}\sum_{i = 1}^{K} 
\underbrace{\begin{bmatrix}
\mathbf{q}_{1, i}  \\
\vdots \\
\mathbf{q}_{j, i} \\
\vdots \\
\mathbf{q}_{J, i}
\end{bmatrix}}_{\triangleq \mathbf{q}_i}
 x_i + 
\underbrace{\begin{bmatrix}
\mathbf{\check{n}}_1  \\
\vdots \\
\mathbf{\check{n}}_j \\
\vdots \\
\mathbf{\check{n}}_J
\end{bmatrix}}_{\triangleq \mathbf{n}}.
    \label{received}
\end{equation}
Using the pilot signal received by the APs, the CPU can calculate the channel estimation and feeds it back to the APs. For the comparison with the proposed method, it is assumed that the CPU has the perfect I-CSI. Next, for using the MRC method at the CPU, the CPU picks a combining vector for user $k$ such that $\mathbf{v}_k =  \mathbf{q}_k$. The instantaneous, SINR is given by
\begin{equation}
    \mathrm{SINR}_k^{\mathrm{i}} = \frac{\rho |\mathbf{v}_k^H  \mathbf{q}_k |^2}{  \rho \sum_{i=1, i \neq k}^{K}  |\mathbf{v}_k^H  \mathbf{q}_i|^2 + \sigma^2\mathbf{v}_k^H \mathbf{I}_{KJ} \mathbf{v}_k}.
    \label{sinr_i}
\end{equation}
Furthermore, the instantaneous SE for user $k$ is given by
\begin{align} \label{SE_i}
     \mathrm{SE}_k^{\mathrm{i}} = \left(1 - \frac{\tau_p}{\tau_c} \right) \log_2\left(1 + \mathrm{SINR}_k^{\mathrm{i}} \right).
\end{align}
Next, the phase shift of the RISs is the solution to the following optimization problem:
\begin{align} \label{opt_general_i}
     \max_{\boldsymbol{\Phi}_r}& \hspace{5mm} \sum_{k=1}^K\mathrm{SE}_k^{\mathrm{i}}  \\\nonumber
      \mathrm{s.t.}& \hspace{3mm} \theta_{r, m} \in [0, 2\pi), r \in \mathcal{R}, m = 1, \dots, M.
\end{align}
The problem in (\ref{opt_general_i}) is solved using the SAC method.

Next, for comparison with the single-cell scenario, in Fig.~\ref{fig:sim_setup}, the RIS at the point $(0, 200, 400)$ is chosen as the cellular RIS with $M_\mathrm{c} = 3M$ and the AP at the point $(500, 200, 30)$ is chosen as the cellular base station with $N_\mathrm{c} = 3N$ antennas. Finally, the phase shifts of the RIS are designed with the same setup as the cell-free case with one RIS and one AP, i.e., $J = R = 1$.

\begin{figure}[t]
    \centering
\includegraphics[scale=0.60]{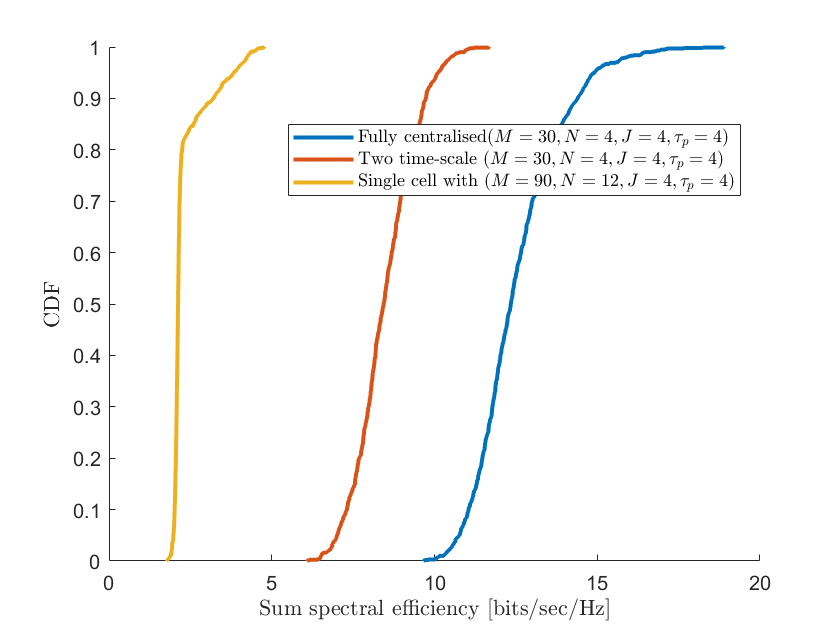}
    \caption{Sum SE of the single cell and cell-free systems.}
    \label{fig:levwel}
\end{figure}

Fig.~\ref{fig:levwel} shows the CDF of the single-cell and fully centralized scenarios. It is observed from Fig.~\ref{fig:levwel} the fully centralized scenario outperforms the other two significantly. The reason is in the I-CSI, as well as the S-CSI is  available at the CPU and the phase shifts of the RISs are designed with the global information of all the systems. Also, the two-timescale scenario outperforms the cellular case significantly.

\section{Conclusion}
This paper explored the two-timescale transmission design for cell-free massive MIMO systems with RIS support, considering channel estimation errors. An LMMSE estimator was proposed to estimate the aggregated instantaneous CSI, which was used by the APs for MRC to detect the user's signal locally. The detected signal is fed back to the CPU for final detection. Since the CPU has only the S-CSI, the final detection is based on the S-CSI. The achievable SE of the users is derived in closed form and is used by the CPU to design the phase shifts of the RIS based on the S-CSI through the SAC method. This method has reduced the signalling overhead for channel estimation significantly. 

\section*{Appendix A}
Before starting the proof of Theorem 1, the following Lemma is required. 
\begin{lem}\label{lem1}
    For user $k \in \mathcal{K}$, the mean vector and covariance matrices needed for the LMMSE estimator are given by
    \begin{align}
        \mathbb{E}\{ \mathbf{q}_{j,  k} \} &= \sum_{r = 1}^{R} \left( \eta_{r, j, k}^{(1)} \Bar{\mathbf{H}}_{r, j}  \boldsymbol{\Phi}_r  \Bar{\mathbf{g}}_{r, k}  \right) = \sum_{r=1}^{R}\mathbf{\check{g}}_{r,  j,  k}^{(1)} = \mathbf{\check{g}}_{j,  k}^{(1)}\\
    \mathbb{E}\{ \mathbf{y}_{j,  k} \} &= \sum_{r = 1}^{R} \left( \eta_{r, j, k}^{(1)} \Bar{\mathbf{H}}_{r, j}  \boldsymbol{\Phi}_r  \Bar{\mathbf{g}}_{r, k}  + \sum_{l \in \mathcal{P}_k \backslash \{k\}} \eta_{r, j, l}^{(1)} \Bar{\mathbf{H}}_{r, j}  \boldsymbol{\Phi}_r  \Bar{\mathbf{g}}_{r, l}  \right) \\ 
    &= \mathbf{\check{g}}_{j,  k}^{(1)} + \sum_{l \in \mathcal{P}_k \backslash \{k\}} \mathbf{\check{g}}_{j,  l}^{(1)}
    \end{align}
\begin{align}
\mathrm{Cov} \{\mathbf{q}_{j,  k}, \mathbf{q}_{j,  k} \} &=  \xi_{j,  k} \mathbf{a}_N \mathbf{a}_N^H + \mu_{j,  k} \mathbf{I}_N  \\
    \mathrm{Cov} \{\mathbf{y}_{j,  k}, \mathbf{y}_{j,  k} \} &= \mathrm{Cov} \{\mathbf{q}_{j,  k}, \mathbf{q}_{j,  k} \} + 2 \sum_{l \in \mathcal{P}_k \backslash \{k\}} \mathrm{Cov} \{\mathbf{q}_{j,  k}, \mathbf{q}_{j,  l} \} \\ \nonumber
    &+ \sum_{l \in \mathcal{P}_k \backslash \{k\}} \mathrm{Cov} \{ \mathbf{q}_{j,  l},  \mathbf{q}_{j,  l} \} + \frac{\sigma^2}{\rho \tau_p} \mathbf{I}_N \\
     &= \omega_{j,  k} \mathbf{a}_N \mathbf{a}_N^H + \nu_{j,  k} \mathbf{I}_N \\
    \mathrm{Cov} \{\mathbf{q}_{j,  k}, \mathbf{y}_{j,  k}   \} &= \mathrm{Cov} \{\mathbf{q}_{j,  k}, \mathbf{q}_{j,  k}    \} + \sum_{l \in \mathcal{P}_k \backslash \{k\}} \mathrm{Cov} \{\mathbf{q}_{j,  k}, \mathbf{q}_{j,  l} \}\\
    &=\xi_{j,  k} \mathbf{a}_N \mathbf{a}_N^H + \chi_{j,  k} \mathbf{I}_N \\
\mathrm{Cov} \{\mathbf{q}_{j,  k},\mathbf{q}_{j,  l} \} &=  \delta_{j,  k, l} \mathbf{I}_N \label{cov_qml_qml}
\end{align}
where
\begin{align}
    \mathbf{\check{g}}_{j,  k}^{(1)} &\triangleq  \sum_{r=1}^{R} \mathbf{\check{g}}_{r, j,  k}^{(1)} \\
    \xi_{r, j,  k} &\triangleq M(\eta_{r, j, k}^{(2)})^2, \\
    \mu_{j,  k} &\triangleq \sum_{r=1}^R ( M ( (\eta_{r, j, k}^{(3)})^2 + (\eta_{r, j, k}^{(4)})^2 ) + (\eta_{j, k}^{(5)})^2 ), \\
     \delta_{j,  k, l} &\triangleq \sum_{r = 1}^R \eta_{r, j, k}^{(3)} \eta_{r, j, l}^{(3)} \mathrm{Tr}\{ \Bar{\mathbf{g}}_{r, k} \Bar{\mathbf{g}}_{r, l}^H  \}, \\
     \omega_{r, j,  k} &\triangleq \xi_{r, j,  k} + \sum_{l \in \mathcal{P}_k \backslash \{k\}} \xi_{r, j,  l}, \\
     \nu_{j,  k} &\triangleq  \mu_{j,  k} +  \sum_{l \in \mathcal{P}_k \backslash \{k\}}\mu_{j,  l} + \sum_{l \in \mathcal{P}_k \backslash \{k\}} \delta_{j,  k, l} + \frac{\sigma^2}{\rho \tau_p}, \\
     \chi_{j,  k} &\triangleq \mu_{j,  k} + \sum_{l \in \mathcal{P}_k \backslash \{k\}}\delta_{j,  k, l}
\end{align}

\end{lem}
\begin{proof}

Recalling the definition of (\ref{channelss}) and noting that $\Tilde{\mathbf{H}}_{r, j}$, $\Tilde{\mathbf{g}}_{r, k}$, $\Tilde{\mathbf{h}}_{j,  k}$, and $\mathbf{\check{n}}_j$ are all independent of each other and having zero-mean entries, we have
\begin{align}
    \mathbb{E}\{ \mathbf{y}_{j,  k} \} =  \mathbb{E}\{ \mathbf{q}_{j,  k}  \} + \sum_{l \in \mathcal{P}_k \backslash \{k\}} \mathbb{E}\{ \mathbf{q}_{j,  l} \} \mathbb{E}\{ \boldsymbol{\varrho}_k^H \boldsymbol{\varrho}_l \} + \mathbb{E}\{  \Tilde{\mathbf{n}}_j \},
    \label{exp_y}
\end{align}
from (\ref{exp_y}), it is obvious that $\mathbb{E}\{  \mathbf{\check{N}}_j \frac{\boldsymbol{\varrho}_k}{\sqrt{\tau_p}} \} = \boldsymbol{0}^{N \times 1}$ and $\mathbb{E}\{ \boldsymbol{\varrho}_k^H \boldsymbol{\varrho}_l \} = 1$. Based on (\ref{channelss}), for calculation of $\mathbb{E}\{ \mathbf{q}_{j,  k}  \}$, we have
\begin{align}
\label{E_q}
    \mathbb{E}\{ \mathbf{q}_{j,  k}  \} &= \sum_{r = 1}^{R} \mathbb{E} \{ \mathbf{q}_{r,  j,  k}  \} = \sum_{r = 1}^{R} \mathbb{E} \{ \mathbf{\check{g}}_{r,  j,  k}^{(1)} \}  \\ \nonumber
    &= \sum_{r = 1}^{R} \left( \eta_{r, j, k}^{(1)} \Bar{\mathbf{H}}_{r, j}  \boldsymbol{\Phi}_r  \Bar{\mathbf{g}}_{r, k}  \right).
\end{align}
Based on (\ref{E_q}), $ \mathbb{E}\{ \mathbf{y}_{j,  k} \}$ can be written as
\begin{align} \label{exp_ymk}
     \mathbb{E}\{ \mathbf{y}_{j,  k} \} &= \mathbb{E}\{ \mathbf{q}_{j,  k}  \} + \sum_{l \in \mathcal{P}_k \backslash \{k\}}  \mathbb{E}\{ \mathbf{q}_{j,  l}  \} \\ \nonumber
     &=\sum_{r = 1}^{R} \left( \eta_{r, j, k}^{(1)} \Bar{\mathbf{H}}_{r, j}  \boldsymbol{\Phi}_r  \Bar{\mathbf{g}}_{r, k}  + \sum_{l \in \mathcal{P}_k \backslash \{k\}} \eta_{r, j, l}^{(1)} \Bar{\mathbf{H}}_{r, j}  \boldsymbol{\Phi}_r  \Bar{\mathbf{g}}_{r, l}  \right).
\end{align}

The covariance matrix between the unknown channel $ \mathbf{q}_{j,  k}$ and the observation vector $\mathbf{y}_{j,  k}$ can be written as
\begin{align}
    \mathrm{Cov} \{\mathbf{q}_{j,  k}, \mathbf{y}_{j,  k}   \} &= \mathbb{E}\Big\{ \Big( \mathbf{q}_{j,  k} - \mathbb{E}\{ \mathbf{q}_{j,  k}\} \Big)\Big( \mathbf{y}_{j,  k} - \mathbb{E}\{ \mathbf{y}_{j,  k} \}  \Big)^H \Big\}. 
\end{align}
Based on (\ref{y_m_k}) and (\ref{exp_ymk}), we have
\begin{align}
    \mathrm{Cov} \{\mathbf{q}_{j,  k}, \mathbf{y}_{j,  k}   \} &= \mathbb{E} \left\{
    \begin{array}{l}
          \Big( \mathbf{q}_{j,  k} - \mathbb{E}\{ \mathbf{q}_{j,  k}\} \Big) \Big( \mathbf{q}_{j,  k} + \sum_{l \in \mathcal{P}_k \backslash \{k\}} \mathbf{q}_{j,  l} \\
         + \Tilde{\mathbf{n}}_j -\mathbb{E}\{ \mathbf{q}_{j,  k}  \} - \sum_{l \in \mathcal{P}_k \backslash \{k\}}  \mathbb{E}\{ \mathbf{q}_{j,  l} \} \Big)^H
    \end{array} \right\} \\
    &=\mathbb{E} \left\{
    \begin{array}{l}
    \Big( \mathbf{q}_{j,  k} - \mathbb{E}\{ \mathbf{q}_{j,  k}\} \Big) \Big( \mathbf{q}_{j,  k} - \mathbb{E}\{ \mathbf{q}_{j,  k}\} \Big)^H \\ 
    + \Big( \mathbf{q}_{j,  k} - \mathbb{E}\{ \mathbf{q}_{j,  k}\} \Big) \left(\mathbf{\check{n}}_j  \frac{\boldsymbol{\varrho}_k}{\sqrt{\rho \tau_p}} \right)^H \\ \nonumber
    + \Big( \mathbf{q}_{j,  k} - \mathbb{E}\{ \mathbf{q}_{j,  k}\} \Big) \Big( \sum_{l \in \mathcal{P}_k \backslash \{k\}} \mathbf{q}_{j,  l}  -  \sum_{l \in \mathcal{P}_k \backslash \{k\}}  \mathbb{E}\{ \mathbf{q}_{j,  l} \}   \Big)^H
    \end{array} \right\} \\ \nonumber
    & = \mathrm{Cov} \{\mathbf{q}_{j,  k}, \mathbf{q}_{j,  k}    \} + \mathrm{Cov} \{\mathbf{q}_{j,  k}, \sum_{l \in \mathcal{P}_k \backslash \{k\}} \mathbf{q}_{j,  l} \} \\ \nonumber
    &= \mathrm{Cov} \{\mathbf{q}_{j,  k}, \mathbf{q}_{j,  k}    \} + \sum_{l \in \mathcal{P}_k \backslash \{k\}} \mathrm{Cov} \{\mathbf{q}_{j,  k}, \mathbf{q}_{j,  l} \}.
\end{align}

For the calculation of $\mathrm{Cov} \{\mathbf{q}_{j,  k}, \mathbf{q}_{j,  k} \}$, we have 
\begin{align}
    \mathrm{Cov} \{\mathbf{q}_{j,  k}, \mathbf{q}_{j,  k} \} &= \mathbb{E}\Bigg\{ \Big(  \mathbf{q}_{j,  k} - \mathbb{E}\{ \mathbf{q}_{j,  k}\} \Big) \Big( \mathbf{q}_{j,  k} - \mathbb{E}\{ \mathbf{q}_{j,  k}\} \Big)^H \Bigg\} \\ \nonumber
    &= \mathbb{E}  \left\{
    \begin{array}{l}
    \sum_{r = 1}^R \Bigg[ \Bigg( \eta_{r, j, k}^{(2)} \Bar{\mathbf{H}}_{r, j} \boldsymbol{\Phi}_r \Tilde{\mathbf{g}}_{r, k}
    +\eta_{r, j, k}^{(3)} \Tilde{\mathbf{H}}_{r, j} \boldsymbol{\Phi}_r \Bar{\mathbf{g}}_{r, k} \\ \nonumber
    + \eta_{r, j, k}^{(4)} \Tilde{\mathbf{H}}_{r, j} \boldsymbol{\Phi}_r \Tilde{\mathbf{g}}_{r, k} + \eta_{j, k}^{(5)} \Tilde{\mathbf{h}}_{j,  k} \Bigg) \Bigg( \eta_{r, j, k}^{(2)} \Tilde{\mathbf{g}}_{r, k}^H \boldsymbol{\Phi}_r^H \Bar{\mathbf{H}}_{r, j}^H \\ \nonumber
    + \eta_{r, j, k}^{(3)} \Bar{\mathbf{g}}_{r, k}^H  \boldsymbol{\Phi}_r^H \Tilde{\mathbf{H}}_{r, j}^H + \eta_{r, j, k}^{(4)} \Tilde{\mathbf{g}}_{r, k}^H \boldsymbol{\Phi}_r^H \Tilde{\mathbf{H}}_{r, j}^H + \eta_{j, k}^{(5)} \Tilde{\mathbf{h}}_{j,  k}^H \Bigg) \Bigg]
    \end{array}
    \right\} \\ \nonumber
    &=\mathbb{E} \left\{
    \begin{array}{l}
    \sum_{r = 1}^R \Bigg[ (\eta_{r, j, k}^{(2)})^2 \Bar{\mathbf{H}}_{r, j} \boldsymbol{\Phi}_r \Tilde{\mathbf{g}}_{r, k} \Tilde{\mathbf{g}}_{r, k}^H \boldsymbol{\Phi}_r^H \Bar{\mathbf{H}}_{r, j}^H \\+ (\eta_{r, j, k}^{(3)})^2 \Tilde{\mathbf{H}}_{r, j} \boldsymbol{\Phi}_r \Bar{\mathbf{g}}_{r, k} \Bar{\mathbf{g}}_{r, k}^H  \boldsymbol{\Phi}_r^H \Tilde{\mathbf{H}}_{r, j}^H \\ \nonumber
    +  (\eta_{r, j, k}^{(4)})^2 \Tilde{\mathbf{H}}_{r, j} \boldsymbol{\Phi}_r \Tilde{\mathbf{g}}_{r, k} \Tilde{\mathbf{g}}_{r, k}^H \boldsymbol{\Phi}_r^H \Tilde{\mathbf{H}}_{r, j}^H + (\eta_{j, k}^{(5)})^2 \Tilde{\mathbf{h}}_{j,  k} \Tilde{\mathbf{h}}_{j,  k}^H \Bigg]
    \end{array}
    \right\} \\ \nonumber
    & = \sum_{r = 1}^R \Big(  (M(\eta_{r, j, k}^{(2)})^2 \mathbf{a}_{N,r,j} \mathbf{a}_{N,r,j}^H +  \left( M \left( (\eta_{r, j, k}^{(3)})^2 + (\eta_{r, j, k}^{(4)})^2 \right) + (\eta_{j, k}^{(5)})^2 \right) \mathbf{I}_N \Big).
\end{align}

Similarly, for $\mathrm{Cov} \{\mathbf{y}_{j,  k}, \mathbf{y}_{j,  k} \}$ we have
\begin{align} \nonumber
    \mathrm{Cov} \{\mathbf{y}_{j,  k}, \mathbf{y}_{j,  k} \} &= \mathbb{E}\Big\{ \Big( \mathbf{y}_{j,  k} - \mathbb{E}\{ \mathbf{y}_{j,  k}\} \Big)\Big( \mathbf{y}_{j,  k} - \mathbb{E}\{ \mathbf{y}_{j,  k} \}  \Big)^H \Big\}  \\ \nonumber
    &= \mathbb{E} \left\{
    \begin{array}{l}
     \Big( \mathbf{q}_{j,  k} + \sum_{l \in \mathcal{P}_k \backslash \{k\}} \mathbf{q}_{j,  l}  + \Tilde{\mathbf{n}}_j - \mathbb{E}\{ \mathbf{q}_{j,  k}  \} \\ - \sum_{l \in \mathcal{P}_k \backslash \{k\}}  \mathbb{E}\{ \mathbf{q}_{j,  l}  \} \Big) \Big( \mathbf{q}_{j,  k} + \sum_{l \in \mathcal{P}_k \backslash \{k\}} \mathbf{q}_{j,  l} \\ + \Tilde{\mathbf{n}}_j - \mathbb{E}\{ \mathbf{q}_{j,  k}  \} - \sum_{l \in \mathcal{P}_k \backslash \{k\}}  \mathbb{E}\{ \mathbf{q}_{j,  l}  \} \}  \Big)^H 
    \end{array} \right\} \\ \label{cov_yy}
    &= \mathrm{Cov} \{\mathbf{q}_{j,  k}, \mathbf{q}_{j,  k} \} + 2\sum_{l \in \mathcal{P}_k \backslash \{k\}} \mathrm{Cov} \{\mathbf{q}_{j,  k}, \mathbf{q}_{j,  l} \} \\ \nonumber
    &+ \sum_{l \in \mathcal{P}_k \backslash \{k\}} \mathrm{Cov} \{\mathbf{q}_{j,  l},  \mathbf{q}_{j,  l} \} + \frac{\sigma^2}{\rho \tau_p} \mathbf{I}_N.
\end{align}

Next, $\mathrm{Cov} \{\mathbf{q}_{j,  k}, \mathbf{q}_{j,  l} \}$ could be calculated as follows
\begin{align}
    \mathrm{Cov} \{&\mathbf{q}_{j,  k},\mathbf{q}_{j,  l} \} =   \mathbb{E}\Bigg\{ \Big(  \mathbf{q}_{j,  k} - \mathbb{E}\{ \mathbf{q}_{j,  k}\} \Big) \Big( \mathbf{q}_{j,  l} - \mathbb{E}\{ \mathbf{q}_{j,  l}\} \Big)^H \Bigg\} \\ \nonumber
    &=   \mathbb{E}  \left\{
    \begin{array}{l}
    \sum_{r = 1}^R  \Bigg[ \Bigg( \eta_{r, j, k}^{(2)} \Bar{\mathbf{H}}_{r, j} \boldsymbol{\Phi}_r \Tilde{\mathbf{g}}_{r, k}
    +\eta_{r, j, k}^{(3)} \Tilde{\mathbf{H}}_{r, j} \boldsymbol{\Phi}_r \Bar{\mathbf{g}}_{r, k} \\ \nonumber
    + \eta_{r, j, k}^{(4)} \Tilde{\mathbf{H}}_{r, j} \boldsymbol{\Phi}_r \Tilde{\mathbf{g}}_{r, k} + \eta_{j, k}^{(5)} \Tilde{\mathbf{h}}_{j,  k} \Bigg) \Bigg( \eta_{r, j, l}^{(2)} \Tilde{\mathbf{g}}_{r, l}^H \boldsymbol{\Phi}_r^H \Bar{\mathbf{H}}_{r, j}^H \\ \nonumber
    +  \eta_{r, j, l}^{(3)} \Bar{\mathbf{g}}_{r, l}^H  \boldsymbol{\Phi}_r^H \Tilde{\mathbf{H}}_{r, j}^H + \eta_{r, j, l}^{(4)} \Tilde{\mathbf{g}}_{r, l}^H \boldsymbol{\Phi}_r^H \Tilde{\mathbf{H}}_{r, j}^H + \eta_{j, l}^{(5)} \Tilde{\mathbf{h}}_{j,  l}^H \Bigg) \Bigg]
    \end{array} \right\} \\ \nonumber
    &=   \mathbb{E} \left\{
    \begin{array}{l}
    \sum_{r = 1}^R  \Bigg[ \eta_{r, j, k}^{(3)} \eta_{r, j, l}^{(3)} \Tilde{\mathbf{H}}_{r, j} \boldsymbol{\Phi}_r \Bar{\mathbf{g}}_{r, k} \Bar{\mathbf{g}}_{r, l}^H  \boldsymbol{\Phi}_r^H \Tilde{\mathbf{H}}_{r, j}^H \Bigg] 
    \end{array}
    \right\} \\ \nonumber
    & = \sum_{r = 1}^R \eta_{r, j, k}^{(3)} \eta_{r, j, l}^{(3)} \mathrm{Tr}( \Bar{\mathbf{g}}_{r, k} \Bar{\mathbf{g}}_{r, l}^H  ) \mathbf{I}_N.
\end{align}
\end{proof}

The LMMSE estimator of the channel $\mathbf{q}_{j,  k}$ based on the observation $\mathbf{y}_{j,  k}$, can be written as \cite[Chapter 12.5]{kay1993fundamentals}
\begin{equation}
    \hat{\mathbf{q}}_{j,  k} = \mathbb{E}\{ \mathbf{q}_{j,  k}\} + \mathrm{Cov} \{ \mathbf{y}_{j,  k}, \mathbf{q}_{j,  k} \} \mathrm{Cov}^{-1} \{ \mathbf{y}_{j,  k}, \mathbf{y}_{j,  k} \} (\mathbf{y}_{j,  k} - \mathbb{E}\{ \mathbf{y}_{j,  k} \}).
\end{equation}
To this end, we begin with the calculation of $\mathrm{Cov}^{-1} \{ \mathbf{y}_{j,  k}, \mathbf{y}_{j,  k} \}$ in (\ref{cov_yy}).
First, $\mathrm{Cov}\{ \mathbf{y}_{j,  k}, \mathbf{y}_{j,  k} \}$ can be written in a compact form as follows
\begin{align}
    \mathrm{Cov}\{ \mathbf{y}_{j,  k}, \mathbf{y}_{j,  k} \} = \sum_{r=1}^R \omega_{r, j,  k} \mathbf{a}_{N, r, j} \mathbf{a}_{N, r, j}^H + \nu_{j,  k }  \mathbf{I}_N ,
\end{align}
where 
\begin{align}
    \xi_{r, j,  k} &\triangleq M(\eta_{r, j, k}^{(2)})^2, \\
    \mu_{j,  k} &\triangleq \sum_{r=1}^R ( M ( (\eta_{r, j, k}^{(3)})^2 + (\eta_{r, j, k}^{(4)})^2 ) + (\eta_{j, k}^{(5)})^2 ), \\
     \delta_{j,  k, l} &\triangleq \sum_{r = 1}^R \eta_{r, j, k}^{(3)} \eta_{r, j, l}^{(3)} \mathrm{Tr}\{ \Bar{\mathbf{g}}_{r, k} \Bar{\mathbf{g}}_{r, l}^H  \}, \\
     \omega_{r, j,  k} &\triangleq \xi_{r, j,  k} + \sum_{l \in \mathcal{P}_k \backslash \{k\}} \xi_{r, j,  l}, \\
     \nu_{j,  k} &\triangleq  \mu_{j,  k} +  \sum_{l \in \mathcal{P}_k \backslash \{k\}}\mu_{j,  l} + 2 \sum_{l \in \mathcal{P}_k \backslash \{k\}} \delta_{j,  k, l} + \frac{\sigma^2}{\rho \tau_p}, \\
     \chi_{j,  k} &\triangleq \mu_{j,  k} + \sum_{l \in \mathcal{P}_k \backslash \{k\}}\delta_{j,  k, l}.
\end{align}
Furthermore, we have
\begin{align}
    \mathrm{Cov} \{ \mathbf{y}_{j,  k}, \mathbf{q}_{j,  k} \} = \sum_{r=1}^R   \xi_{r, j,  k} \mathbf{a}_{N, r, j} \mathbf{a}_{N, r, j}^H + \chi_{j,  k} \mathbf{I}_N.
\end{align}
Next, we arrive at
\begin{align} \nonumber
    \mathrm{Cov} \{ \mathbf{y}_{j,  k}, & \mathbf{q}_{j,  k} \}  \mathrm{Cov}^{-1} \{ \mathbf{y}_{j,  k}, \mathbf{y}_{j,  k} \} \\  
    &= \Big(\sum_{r=1}^R   \xi_{r, j,  k} \mathbf{a}_{N, r, j} \mathbf{a}_{N, r, j}^H + \chi_{j,  k} \mathbf{I}_N \Big)\Big( \sum_{r=1}^R \omega_{r, j,  k} \mathbf{a}_{N, r, j} \mathbf{a}_{N, r, j}^H + \nu_{j,  k }  \mathbf{I}_N  \Big)^{-1} \triangleq \mathbf{A}_{j,  k} = \mathbf{A}_{j,  k}^H .
\end{align}
Then, the estimated channel can be calculated as
\begin{align}
    \hat{\mathbf{q}}_{j,  k} &= \mathbf{\check{g}}_{j,  k}^{(1)} +  \mathbf{A}_{j,  k} \left( \mathbf{y}_{j,  k} -  \mathbf{\check{g}}_{j,  k}^{(1)} - \sum_{l \in \mathcal{P}_k \backslash \{k\}} \mathbf{\check{g}}_{j,  l}^{(1)} \right) \\ 
    &= \mathbf{A}_{r, j,  k} \mathbf{y}_{j,  k} + (\mathbf{I}_N - \mathbf{A}_{j,  k})\mathbf{\check{g}}_{j,  k}^{(1)} - \sum_{l \in \mathcal{P}_k \backslash \{k\}} \mathbf{A}_{r, j,  k} \mathbf{\check{g}}_{j,  l}^{(1)}.
    \label{thisis}
\end{align}
Additionally, (\ref{thisis}) can be further expanded as follows
\begin{align} \label{est_q}
    \hat{\mathbf{q}}_{j,  k} &= \mathbf{A}_{j,  k} \left( \sum_{u = 1}^{4} \mathbf{\check{g}}_{j,  k}^{(u)} +  \mathbf{h}_{j,  k} + \sum_{l \in \mathcal{P}_k \backslash \{k\}} \left( \sum_{u = 1}^{4} \mathbf{\check{g}}_{j,  l}^{(u)} + \mathbf{h}_{j,  l} \right) + \Tilde{\mathbf{n}}_j \right) \\ \nonumber
    &+ (\mathbf{I}_N - \mathbf{A}_{j,  k})\mathbf{\check{g}}_{j,  k}^{(1)} - \sum_{l \in \mathcal{P}_k \backslash \{k\}} \mathbf{A}_{j,  k} \mathbf{\check{g}}_{j,  l}^{(1)} \\ \nonumber
    &= \mathbf{\check{g}}_{j,  k}^{(1)} + \sum_{u = 2}^{4} \mathbf{A}_{j,  k}\mathbf{\check{g}}_{j,  k}^{(u)} + \mathbf{A}_{j,  k} \mathbf{h}_{j,  k} + \sum_{l \in \mathcal{P}_k \backslash \{k\}} \mathbf{A}_{j,  k} \left( \sum_{u = 2}^{4}  \mathbf{\check{g}}_{j,  l}^{(u)} +  \mathbf{h}_{j,  l} \right) \\ \nonumber
    &+  \mathbf{A}_{j,  k} \Tilde{\mathbf{n}}_j.
\end{align}
Furthermore, the channel estimation expression in (\ref{est_q}) could be written in a more simpler form as follows
\begin{align} \label{est_r}
     \hat{\mathbf{q}}_{j,  k} &= \sum_{r=1}^R \Bigg( \eta_{r, j, k}^{(1)} \Bar{\mathbf{H}}_{r, j}  \boldsymbol{\Phi}_r  \Bar{\mathbf{g}}_{r, k} +\eta_{r, j, k}^{(2)}  \mathbf{A}_{j,  k}  \Bar{\mathbf{H}}_{r, j} \boldsymbol{\Phi}_r \Tilde{\mathbf{g}}_{r, k} \\\nonumber
     &+ \eta_{r, j, k}^{(3)} \mathbf{A}_{j,  k} \Tilde{\mathbf{H}}_{r, j} \boldsymbol{\Phi}_r \Bar{\mathbf{g}}_{r, k} + \eta_{r, j, k}^{(4)} \mathbf{A}_{j,  k} \Tilde{\mathbf{H}}_{r, j} \boldsymbol{\Phi}_r \Tilde{\mathbf{g}}_{r, k} + \eta_{j, k}^{(5)} \mathbf{A}_{j,  k} \Tilde{\mathbf{h}}_{j, k} \\\nonumber
     &+ \sum_{l \in \mathcal{P}_k \backslash \{k\}} \Bigg(  \eta_{r, j, l}^{(2)} \mathbf{A}_{j,  k} \Bar{\mathbf{H}}_{r, j} \boldsymbol{\Phi}_r \Tilde{\mathbf{g}}_{r, l} +  \eta_{r, j, l}^{(3)} \mathbf{A}_{r, j,  k} \Tilde{\mathbf{H}}_{r, j} \boldsymbol{\Phi}_r \Bar{\mathbf{g}}_{r, l} \\\nonumber
     &+  \eta_{r, j, l}^{(4)} \mathbf{A}_{r, j,  k} \Tilde{\mathbf{H}}_{r, j} \boldsymbol{\Phi}_r \Tilde{\mathbf{g}}_{r, l} +  \eta_{j, l}^{(5)} \mathbf{A}_{r, j,  k} \Tilde{\mathbf{h}}_{j, l} \Bigg) + \mathbf{A}_{j,  k} \Tilde{\mathbf{n}}_j \Bigg).
\end{align}

\section*{Appendix B}
\subsection{Derivation of $ \mathbb{E} \{\hat{\mathbf{q}}_{j, k}^H \mathbf{q}_{j, k}\}$}
We first begin with finding the $\mathbb{E}\{ \mathbf{w}_{k, k} \}$ as follows
\begin{align}\nonumber
    \mathbb{E} \{ \mathbf{w}_{k, k} \} &=  \mathbb{E}\left\{ \left[\hat{\mathbf{q}}_{1, k}^H \mathbf{q}_{1, k}, \dots, \hat{\mathbf{q}}_{J, k}^H \mathbf{q}_{J, k}  \right]^T \right\} \\ 
    &=  \Big[\mathbb{E} \{\hat{\mathbf{q}}_{1, k}^H \mathbf{q}_{1, k}\}, \dots, \mathbb{E}\{\hat{\mathbf{q}}_{J, k}^H \mathbf{q}_{J, k}\}  \Big]^T. \label{Ewkk} 
\end{align}
Next, we move on to calculate the $j$th element of (\ref{Ewkk}) where the other elements are similar
\begin{align} \nonumber 
    \mathbb{E} \{\hat{\mathbf{q}}_{j, k}^H \mathbf{q}_{j, k}\} &= \sum_{r=1}^R  \Bigg(\mathbb{E} \Bigg\{ \Bigg(  \eta_{r, j, k}^{(1)} \Bar{\mathbf{H}}_{r, j}  \boldsymbol{\Phi}_r  \Bar{\mathbf{g}}_{r, k} +\eta_{r, j, k}^{(2)}  \mathbf{A}_{j,  k}  \Bar{\mathbf{H}}_{r, j} \boldsymbol{\Phi}_r \Tilde{\mathbf{g}}_{r, k} \\\nonumber
     &+ \eta_{r, j, k}^{(3)} \mathbf{A}_{j,  k} \Tilde{\mathbf{H}}_{r, j} \boldsymbol{\Phi}_r \Bar{\mathbf{g}}_{r, k} + \eta_{r, j, k}^{(4)} \mathbf{A}_{j,  k} \Tilde{\mathbf{H}}_{r, j} \boldsymbol{\Phi}_r \Tilde{\mathbf{g}}_{r, k} + \eta_{j, k}^{(5)} \mathbf{A}_{j,  k} \Tilde{\mathbf{h}}_{j, k} \\\nonumber
     &+ \sum_{l \in \mathcal{P}_k \backslash \{k\}} \Bigg(  \eta_{r, j, l}^{(2)} \mathbf{A}_{j,  k} \Bar{\mathbf{H}}_{r, j} \boldsymbol{\Phi}_r \Tilde{\mathbf{g}}_{r, l} +  \eta_{r, j, l}^{(3)} \mathbf{A}_{r, j,  k} \Tilde{\mathbf{H}}_{r, j} \boldsymbol{\Phi}_r \Bar{\mathbf{g}}_{r, l} \\\nonumber
     &+  \eta_{r, j, l}^{(4)} \mathbf{A}_{r, j,  k} \Tilde{\mathbf{H}}_{r, j} \boldsymbol{\Phi}_r \Tilde{\mathbf{g}}_{r, l} +  \eta_{j, l}^{(5)} \mathbf{A}_{r, j,  k} \Tilde{\mathbf{h}}_{j, l} \Bigg) + \mathbf{A}_{j,  k} \Tilde{\mathbf{n}}_j \Bigg)^H
     \sum_{r^\prime=1}^R\Bigg( \eta_{r^\prime, j, k}^{(1)} \Bar{\mathbf{H}}_{r^\prime, j}  \boldsymbol{\Phi}_{r^\prime}  \Bar{\mathbf{g}}_{r^\prime, k} \\ \nonumber
     &+  \eta_{r^\prime, j, k}^{(2)} \Bar{\mathbf{H}}_{r^\prime, j} \boldsymbol{\Phi}_r^\prime \Tilde{\mathbf{g}}_{r^\prime, k} + \eta_{r^\prime, j, k}^{(3)} \Tilde{\mathbf{H}}_{r^\prime, j} \boldsymbol{\Phi}_{r^\prime} \Bar{\mathbf{g}}_{r^\prime, k} + \eta_{r^\prime, j, k}^{(4)} \Tilde{\mathbf{H}}_{r^\prime, j} \boldsymbol{\Phi}_{r^\prime} \Tilde{\mathbf{g}}_{r^\prime, k} + \eta_{j, k}^{(5)} \Tilde{\mathbf{h}}_{j, k} 
     \Bigg) \Bigg\} \Bigg) \\ \nonumber
     &= \sum_{r=1}^R  \Bigg( \sum_{r^\prime=1}^R \Big(\eta_{r, j, k}^{(1)} \eta_{r^\prime, j, k}^{(1)} \Bar{\mathbf{g}}_{r, k}^H \boldsymbol{\Phi}_r^H \Bar{\mathbf{H}}_{r, j}^H \Bar{\mathbf{H}}_{r, j}  \boldsymbol{\Phi}_r  \Bar{\mathbf{g}}_{r, k} \Big) + (\eta_{r, j, k}^{(2)})^2  \mathbb{E}\{ \Tilde{\mathbf{g}}_{r, k}^H \boldsymbol{\Phi}_r^H \Bar{\mathbf{H}}_{r, j}^H \mathbf{A}_{j,  k}^H \Bar{\mathbf{H}}_{r, j} \boldsymbol{\Phi}_r \Tilde{\mathbf{g}}_{r, k}   \} \\\nonumber
     &+ (\eta_{r, j, k}^{(3)})^2 \Bar{\mathbf{g}}_{r, k}^H \boldsymbol{\Phi}_r^H \mathbb{E}\{\Tilde{\mathbf{H}}_{r, j}^H \mathbf{A}_{j,  k}^H \Tilde{\mathbf{H}}_{r, j}\} \boldsymbol{\Phi}_r \Bar{\mathbf{g}}_{r, k}     + (\eta_{r, j, k}^{(4)})^2 \mathbb{E}\{\Tilde{\mathbf{g}}_{r, k}^H \boldsymbol{\Phi}_r^H \Tilde{\mathbf{H}}_{r, j}^H \mathbf{A}_{j,  k}^H \Tilde{\mathbf{H}}_{r, j} \boldsymbol{\Phi}_r \Tilde{\mathbf{g}}_{r, k}\} \\ \nonumber
     &+ (\eta_{j, k}^{(5)})^2 \mathbb{E}\{\Tilde{\mathbf{h}}_{j, k}^H \mathbf{A}_{j,  k}^H  \Tilde{\mathbf{h}}_{j, k} \} + \sum_{l \in \mathcal{P}_k \backslash \{k\}} \Big( \eta_{r, j, k}^{(3)} \eta_{r, j, l}^{(3)} \mathbb{E}\{\Bar{\mathbf{g}}_{r, l}^H \boldsymbol{\Phi}_r^H \Tilde{\mathbf{H}}_{r, j}^H \mathbf{A}_{j,  k}^H \Tilde{\mathbf{H}}_{r, j} \boldsymbol{\Phi}_r \Bar{\mathbf{g}}_{r, k} \} \Big) \Bigg) \\\nonumber
     &= \sum_{r=1}^R  \Bigg(   \sum_{r^\prime=1}^R \Big(\eta_{r, j, k}^{(1)} \eta_{r^\prime, j, k}^{(1)} \Bar{\mathbf{g}}_{r, k}^H \boldsymbol{\Phi}_r^H \Bar{\mathbf{H}}_{r, j}^H \Bar{\mathbf{H}}_{r^\prime, j}  \boldsymbol{\Phi}_{r^\prime}  \Bar{\mathbf{g}}_{r^\prime, k} \Big) + (\eta_{r, j, k}^{(2)})^2 M g_{r, j}\\ \nonumber
     &+  (\eta_{r, j, k}^{(3)})^2 M \mathrm{Tr}\{ \mathbf{A}_{j,  k}^H \}+ (\eta_{r, j, k}^{(4)})^2 M \mathrm{Tr}\{ \mathbf{A}_{j,  k}^H \} + (\eta_{j, k}^{(5)})^2 \mathrm{Tr}\{ \mathbf{A}_{j,  k}^H \} \\ \label{exp_qhq} 
     &+ \sum_{l \in \mathcal{P}_k \backslash \{k\}} \Big( \eta_{r, j, k}^{(3)} \eta_{r, j, l}^{(3)}  \mathrm{Tr}\{ \mathbf{A}_{j,  k}^H \}\Bar{\mathbf{g}}_{r, l}^H  \Bar{\mathbf{g}}_{r, k} \Big) \Bigg),
\end{align}
 where
\begin{align}
    g_{r, j} &\triangleq  \mathbf{a}_{N, r, j}^H \mathbf{A}_{j,  k} \mathbf{a}_{N, r, j} = g_{r, j}^H.
\end{align}

\subsection{Derivation of $  \mathbb{E}\{\hat{\mathbf{q}}_{j, k}^H \mathbf{q}_{j, i}  (\hat{\mathbf{q}}_{h, k}^H \mathbf{q}_{h, i})^H\}    $}

Next, we will focus on the derivation of $  \mathbb{E}\{\hat{\mathbf{q}}_{j, k}^H \mathbf{q}_{j, i}  (\hat{\mathbf{q}}_{h, k}^H \mathbf{q}_{h, i})^H\}    $ as follows
\begin{align} \nonumber
 \mathbb{E}\{\hat{\mathbf{q}}_{j, k}^H & \mathbf{q}_{j, i}  (\hat{\mathbf{q}}_{h, k}^H \mathbf{q}_{h, i})^H\} = \mathbb{E}\Bigg\{\sum_{r_1=1}^R 
 \Bigg(
 \eta_{r_1, j, k}^{(1)} \Bar{\mathbf{H}}_{r_1, j}  \boldsymbol{\Phi}_{r_1}  \Bar{\mathbf{g}}_{{r_1}, k} +\eta_{{r_1}, j, k}^{(2)}  \mathbf{A}_{j,  k}  \Bar{\mathbf{H}}_{r, j} \boldsymbol{\Phi}_r \Tilde{\mathbf{g}}_{r, k} \\\nonumber
     &+ \eta_{{r_1}, j, k}^{(3)} \mathbf{A}_{j,  k} \Tilde{\mathbf{H}}_{{r_1}, j} \boldsymbol{\Phi}_{r_1} \Bar{\mathbf{g}}_{{r_1}, k} + \eta_{{r_1}, j, k}^{(4)} \mathbf{A}_{j,  k} \Tilde{\mathbf{H}}_{{r_1}, j} \boldsymbol{\Phi}_{r_1} \Tilde{\mathbf{g}}_{{r_1}, k} + \eta_{j, k}^{(5)} \mathbf{A}_{j,  k} \Tilde{\mathbf{h}}_{j, k} \\\nonumber
     &+ \sum_{l \in \mathcal{P}_k \backslash \{k\}} \Bigg(  \eta_{{r_1}, j, l}^{(2)} \mathbf{A}_{j,  k} \Bar{\mathbf{H}}_{{r_1}, j} \boldsymbol{\Phi}_{r_1} \Tilde{\mathbf{g}}_{{r_1}, l} +  \eta_{{r_1}, j, l}^{(3)} \mathbf{A}_{{r_1}, j,  k} \Tilde{\mathbf{H}}_{{r_1}, j} \boldsymbol{\Phi}_{r_1} \Bar{\mathbf{g}}_{{r_1}, l} \\\nonumber
     &+  \eta_{{r_1}, j, l}^{(4)} \mathbf{A}_{{r_1}, j,  k} \Tilde{\mathbf{H}}_{{r_1}, j} \boldsymbol{\Phi}_{r_1} \Tilde{\mathbf{g}}_{{r_1}, l} +  \eta_{j, l}^{(5)} \mathbf{A}_{j,  k} \Tilde{\mathbf{h}}_{j, l} \Bigg) + \mathbf{A}_{j,  k} \Tilde{\mathbf{n}}_j 
     \Bigg)^H \\ \nonumber
     & \times
     \sum_{r_2=1}^R \Bigg( \eta_{{r_2}, j, i}^{(1)} \Bar{\mathbf{H}}_{{r_2}, j}  \boldsymbol{\Phi}_{r_2}  \Bar{\mathbf{g}}_{{r_2}, i} +  \eta_{{r_2}, j, i}^{(2)} \Bar{\mathbf{H}}_{{r_2}, j} \boldsymbol{\Phi}_{r_2} \Tilde{\mathbf{g}}_{{r_2}, i} + \eta_{{r_2}, j, i}^{(3)} \Tilde{\mathbf{H}}_{{r_2}, j} \boldsymbol{\Phi}_{r_2} \Bar{\mathbf{g}}_{{r_2}, i} \\ \nonumber
     &+ \eta_{{r_2}, j, i}^{(4)} \Tilde{\mathbf{H}}_{{r_2}, j} \boldsymbol{\Phi}_{r_2} \Tilde{\mathbf{g}}_{{r_2}, i} + \eta_{j, i}^{(5)} \Tilde{\mathbf{h}}_{j, i} 
     \Bigg) \\ \nonumber
     &\times 
     \sum_{r_3=1}^R\Bigg( \eta_{{r_3}, h, i}^{(1)} \Bar{\mathbf{H}}_{{r_3}, h}  \boldsymbol{\Phi}_{r_3}  \Bar{\mathbf{g}}_{{r_3}, i} +  \eta_{{r_3}, h, i}^{(2)} \Bar{\mathbf{H}}_{{r_3}, h} \boldsymbol{\Phi}_{r_3} \Tilde{\mathbf{g}}_{{r_3}, i} + \eta_{{r_3}, h, i}^{(3)} \Tilde{\mathbf{H}}_{{r_3}, h} \boldsymbol{\Phi}_{r_3} \Bar{\mathbf{g}}_{{r_3}, i} \\ \nonumber
     &+ \eta_{{r_3}, h, i}^{(4)} \Tilde{\mathbf{H}}_{{r_3}, h} \boldsymbol{\Phi}_{r_3} \Tilde{\mathbf{g}}_{{r_3}, i} + \eta_{j, i}^{(5)} \Tilde{\mathbf{h}}_{h, i}  
     \Bigg)^H \\ \nonumber
     & \times 
      \sum_{r_4=1}^R\Bigg(
\eta_{{r_4}, h, k}^{(1)} \Bar{\mathbf{H}}_{{r_4}, h}  \boldsymbol{\Phi}_{r_4}  \Bar{\mathbf{g}}_{{r_4}, k} +\eta_{{r_4}, h, k}^{(2)}  \mathbf{A}_{h,  k}  \Bar{\mathbf{H}}_{{r_4}, h} \boldsymbol{\Phi}_{r_4} \Tilde{\mathbf{g}}_{{r_4}, k} \\\nonumber
     &+ \eta_{{r_4}, h, k}^{(3)} \mathbf{A}_{h,  k} \Tilde{\mathbf{H}}_{{r_4}, h} \boldsymbol{\Phi}_{r_4} \Bar{\mathbf{g}}_{{r_4}, k} + \eta_{{r_4}, h, k}^{(4)} \mathbf{A}_{h,  k} \Tilde{\mathbf{H}}_{{r_4}, h} \boldsymbol{\Phi}_{r_4} \Tilde{\mathbf{g}}_{{r_4}, k} + \eta_{h, k}^{(5)} \mathbf{A}_{h,  k} \Tilde{\mathbf{h}}_{h, k} \\\nonumber
     &+ \sum_{l \in \mathcal{P}_k \backslash \{k\}} \Bigg(  \eta_{{r_4}, h, l}^{(2)} \mathbf{A}_{h,  k} \Bar{\mathbf{H}}_{r, h} \boldsymbol{\Phi}_{r_4} \Tilde{\mathbf{g}}_{{r_4}, l} +  \eta_{{r_4}, h, l}^{(3)} \mathbf{A}_{h,  k} \Tilde{\mathbf{H}}_{{r_4}, h} \boldsymbol{\Phi}_{r_4} \Bar{\mathbf{g}}_{{r_4}, l} \\ \label{offdiag}
     &+  \eta_{{r_4}, h, l}^{(4)} \mathbf{A}_{h,  k} \Tilde{\mathbf{H}}_{{r_4}, h} \boldsymbol{\Phi}_{r_4} \Tilde{\mathbf{g}}_{{r_4}, l} +  \eta_{h, l}^{(5)} \mathbf{A}_{h,  k} \Tilde{\mathbf{h}}_{h, l} \Bigg) + \mathbf{A}_{h,  k} \Tilde{\mathbf{n}}_j\Bigg)
     \Bigg\}. 
\end{align}
We can write (\ref{offdiag}) in a more compact form as follows
\begin{align} \nonumber
    \mathbb{E}\{\hat{\mathbf{q}}_{j, k}^H &  \mathbf{q}_{j, i}  (\hat{\mathbf{q}}_{h, k}^H \mathbf{q}_{h, i})^H\} \\ \nonumber
    &=  \Bigg( \sum_{r_1=1}^R \Big( \sum_{u_1 = 2}^{4} \mathbf{A}_{j,  k} \mathbf{\check{g}}_{r_1,  j,  k}^{(u_1)} + \mathbf{\check{g}}_{r_1,  j,  k}^{(1)} +  \sum_{l \in \mathcal{P}_k \backslash \{k\}} \sum_{u_1 = 2}^{4} \mathbf{A}_{j,  k} \mathbf{\check{g}}_{r_1,  j,  l}^{(u_1)} \Big) \\\nonumber
    &+ \eta_{j, k}^{(5)} \mathbf{A}_{j,  k} \Tilde{\mathbf{h}}_{j, k} + \eta_{j, l}^{(5)} \mathbf{A}_{j,  k} \Tilde{\mathbf{h}}_{j, l} + \mathbf{A}_{j,  k} \Tilde{\mathbf{n}}_j \Bigg)^H \\ \nonumber
    & \times \Bigg( \sum_{r_2=1}^R\sum_{u_2 = 1}^{4} \mathbf{\check{g}}_{r_2,  j,  i}^{(u_2)} + \mathbf{h}_{j, i} \Bigg) \\ \nonumber
    & \times \Bigg( \sum_{r_3=1}^R\sum_{u_3 = 1}^{4} \mathbf{\check{g}}_{r_3,  h,  i}^{(u_3)} + \mathbf{h}_{h, i} \Bigg)^H \\\nonumber
    & \times \Bigg( \sum_{r_4=1}^R \Big( \sum_{u_4 = 2}^{4} \mathbf{A}_{h,  k} \mathbf{\check{g}}_{r_4,  h,  k}^{(u_4)} + \mathbf{\check{g}}_{r_4,  h,  k}^{(1)} +  \sum_{l \in \mathcal{P}_k \backslash \{k\}} \sum_{u_4 = 2}^{4} \mathbf{A}_{h,  k} \mathbf{\check{g}}_{r_4,  h,  l}^{(u_4)} \Big) \\ \label{new_off}
    &+ \eta_{h, l}^{(5)} \mathbf{A}_{h,  k} \Tilde{\mathbf{h}}_{h, l} + \eta_{h, k}^{(5)} \mathbf{A}_{h,  k} \Tilde{\mathbf{h}}_{h, k} + \mathbf{A}_{h,  k} \Tilde{\mathbf{n}}_h \Bigg).
\end{align}

Next, in order to find the non-zero terms (\ref{new_off}), we will write the channel model in (\ref{channelss}) in the following way:
\begin{align} \nonumber
    \mathbf{q}_{j, k} &= \eta_{r, j, k}^{(1)} \Bar{\mathbf{H}}_{r, j}  \boldsymbol{\Phi}_r  \Bar{\mathbf{g}}_{r, k} +  \eta_{r, j, k}^{(2)} \Bar{\mathbf{H}}_{r, j} \boldsymbol{\Phi}_r \Tilde{\mathbf{g}}_{r, k} \\ \nonumber
    &+ \eta_{r, j, k}^{(3)} \Tilde{\mathbf{H}}_{r, j} \boldsymbol{\Phi}_r \Bar{\mathbf{g}}_{r, k} + \eta_{r, j, k}^{(4)} \Tilde{\mathbf{H}}_{r, j} \boldsymbol{\Phi}_r \Tilde{\mathbf{g}}_{r, k} + \eta_{j, k}^{(5)} \Tilde{\mathbf{h}}_{j, k} \\ \label{channelss}
    &=  \mathcal{Q}_{r, j, k}^{(1)}(.)  + \mathcal{Q}_{r, j, k}^{(2)}(\Tilde{\mathbf{g}}_{r, k}) + \mathcal{Q}_{r, j, k}^{(3)}(\Tilde{\mathbf{H}}_{r, j}) + \mathcal{Q}_{r, j, k}^{(4)}(\Tilde{\mathbf{g}}_{r, k}, \Tilde{\mathbf{H}}_{r, j}) + \mathcal{Q}_{j, k}^{(5)}(\Tilde{\mathbf{h}}_{j, k}),
\end{align}
where $ \mathcal{Q}_{r, j, k}^{(u)}(\mathbf{x}) $ is random variable indicator and shows that the $u$th term of the channel is a function of a the random variable $\mathbf{x}$. Furthermore, $\mathcal{Q}_{r, j, k}^{(1)}(.)$ indicates that the first term, i.e., $\eta_{r, j, k}^{(1)} \Bar{\mathbf{H}}_{r, j}  \boldsymbol{\Phi}_r  \Bar{\mathbf{g}}_{r, k}$ does not contain a random variable. Hence, we can rewrite (\ref{new_off}) as follows

\begin{align} \nonumber
    \mathbb{E}\{& \hat{\mathbf{q}}_{j, k}^H \mathbf{q}_{j, i}  (\hat{\mathbf{q}}_{h, k}^H \mathbf{q}_{h, i})^H\} \\ \nonumber
    &= \mathbb{E} \Bigg\{  \Bigg( \sum_{r_1=1}^R \Big( \mathcal{Q}_{r_1, j, k}^{(1)}(.)  + \hat{\mathcal{Q}}_{r_1, j, k}^{(2)}(\Tilde{\mathbf{g}}_{r_1, k}) + \hat{\mathcal{Q}}_{r_1, j, k}^{(3)}(\Tilde{\mathbf{H}}_{r_1, j}) + \hat{\mathcal{Q}}_{r_1, j, k}^{(4)}(\Tilde{\mathbf{g}}_{r_1, k}, \Tilde{\mathbf{H}}_{r_1, j}) + \hat{\mathcal{Q}}_{j, k}^{(5)}(\Tilde{\mathbf{h}}_{j, k}) + \\ \nonumber
    & + \sum_{l \in \mathcal{P}_k \backslash \{k\}} \big( \hat{\mathcal{Q}}_{r_1, j, l}^{(2)}(\Tilde{\mathbf{g}}_{r_1, l}) + \hat{\mathcal{Q}}_{r_1, j, l}^{(3)}(\Tilde{\mathbf{H}}_{r_1, j}) + \hat{\mathcal{Q}}_{r_1, j, l}^{(4)}(\Tilde{\mathbf{g}}_{r_1, l}, \Tilde{\mathbf{H}}_{r_1, j}) + \hat{\mathcal{Q}}_{j, l}^{(5)}(\Tilde{\mathbf{h}}_{j, l}) \big) + \mathbf{A}_{j,  k} \Tilde{\mathbf{n}}_j  \Big) \Bigg)^H \\ \nonumber
    & \times \Bigg( \sum_{r_2=1}^R \mathcal{Q}_{r_2, j, i}^{(1)}(.)  + \mathcal{Q}_{r_2, j, i}^{(2)}(\Tilde{\mathbf{g}}_{r_2, i}) + \mathcal{Q}_{r_2, j, i}^{(3)}(\Tilde{\mathbf{H}}_{r_2, j}) + \mathcal{Q}_{r_2, j, i}^{(4)}(\Tilde{\mathbf{g}}_{r_2, i}, \Tilde{\mathbf{H}}_{r_2, j}) + \mathcal{Q}_{j, i}^{(5)}(\Tilde{\mathbf{h}}_{j, i}) \Bigg) \\ \nonumber
    & \times \Bigg( \sum_{r_3=1}^R \mathcal{Q}_{r_3, h, i}^{(1)}(.)  + \mathcal{Q}_{r_3, h, i}^{(2)}(\Tilde{\mathbf{g}}_{r_3, i}) + \mathcal{Q}_{r_3, h, i}^{(3)}(\Tilde{\mathbf{H}}_{r_3, h}) + \mathcal{Q}_{r_3, h, i}^{(4)}(\Tilde{\mathbf{g}}_{r_3, i}, \Tilde{\mathbf{H}}_{r_3, h}) + \mathcal{Q}_{h, i}^{(5)}(\Tilde{\mathbf{h}}_{h, i}) \Bigg)^H \\  \nonumber
    & \times \Bigg( \sum_{r_4=1}^R \Big( \mathcal{Q}_{r_4, h, k}^{(1)}(.)  + \hat{\mathcal{Q}}_{r_4, h, k}^{(2)}(\Tilde{\mathbf{g}}_{r_4, k}) + \hat{\mathcal{Q}}_{r_4, h, k}^{(3)}(\Tilde{\mathbf{H}}_{r_4, h}) + \hat{\mathcal{Q}}_{r_4, h, k}^{(4)}(\Tilde{\mathbf{g}}_{r_4, k}, \Tilde{\mathbf{H}}_{r_4, h}) + \hat{\mathcal{Q}}_{h, k}^{(5)}(\Tilde{\mathbf{h}}_{h, k}) \\ \label{new_off_Q}
    & + \sum_{l \in \mathcal{P}_k \backslash \{k\}} \big( \hat{\mathcal{Q}}_{r_4, h, l}^{(2)}(\Tilde{\mathbf{g}}_{r_4, l}) + \hat{\mathcal{Q}}_{r_4, h, l}^{(3)}(\Tilde{\mathbf{H}}_{r_4, h}) + \hat{\mathcal{Q}}_{r_4, h, l}^{(4)}(\Tilde{\mathbf{g}}_{r_4, l}, \Tilde{\mathbf{H}}_{r_4, h}) + \hat{\mathcal{Q}}_{h, l}^{(5)}(\Tilde{\mathbf{h}}_{h, k})\big)  + \mathbf{A}_{h,  k} \Tilde{\mathbf{n}}_h \Big) \Bigg) \Bigg\},
\end{align}
where $\hat{\mathcal{Q}}_{r, j, k}^{(u)}(\mathbf{x}) = \mathbf{A}_{j,  k}\mathcal{Q}_{r, j, j}^{(u)}(\mathbf{x})$.
Thanks to the in-dependency of the channels, the majority of the expectations in (\ref{new_off_Q}) are zero. Since there are four expressions multiplied by each other, we write the non-zero entries one-by-one as follows. 

First of all, the expectations of the non-random terms are always non-zero and could be written as follows
\begin{align} \nonumber
    \mathbb{E} \{ & \mathcal{Q}_{r_1, j, k}^{(1)}(.)  \mathcal{Q}_{r_2, j, i}^{(1)}(.) \mathcal{Q}_{r_3, h, i}^{(1)}(.) \mathcal{Q}_{r_4, h, k}^{(1)}(.)\} = \\ \label{cross_1_f}
    & \eta_{r_1, j, k}^{(1)} \eta_{r_2, j, i}^{(1)}  \eta_{r_3, h, i}^{(1)}   \eta_{r_4, h, k}^{(1)}
  \Bar{\mathbf{g}}_{r_1, k}^H \boldsymbol{\Phi}_{r_1}^H \Bar{\mathbf{H}}_{r_1, j}^H
  \Bar{\mathbf{H}}_{r_2, j}  \boldsymbol{\Phi}_{r_2}  \Bar{\mathbf{g}}_{r_2, i}
  \Bar{\mathbf{g}}_{r_3, i}^H \boldsymbol{\Phi}_{r_3}^H \Bar{\mathbf{H}}_{r_3, h}^H   
   \Bar{\mathbf{H}}_{r_4, h}  \boldsymbol{\Phi}_{r_4}  \Bar{\mathbf{g}}_{r_4, k}.
\end{align}

The next non-zero term happens when the case $\hat{\mathcal{Q}}_{r_1, j, k}^{(2)}(\Tilde{\mathbf{g}}_{r_1, k})  \mathcal{Q}_{r_2, j, i}^{(1)}(.) \mathcal{Q}_{r_3, h, i}^{(1)}(.) \hat{\mathcal{Q}}_{r_4, h, k}^{(2)}(\Tilde{\mathbf{g}}_{r_4, k})$ occurs with $r_1 = r_4$ and it could be written as
\begin{align} \nonumber
    \mathbb{E} \{ & \hat{\mathcal{Q}}_{r_1, j, k}^{(2)}(\Tilde{\mathbf{g}}_{r_1, k})  \mathcal{Q}_{r_2, j, i}^{(1)}(.) \mathcal{Q}_{r_3, h, i}^{(1)}(.) \hat{\mathcal{Q}}_{r_4, h, k}^{(2)}(\Tilde{\mathbf{g}}_{r_4, k})\} = \\ \nonumber
    &\mathbb{E}\{ \Tilde{\mathbf{g}}_{r_1, k}^H  \boldsymbol{\Phi}_{r_1}^H \Bar{\mathbf{H}}_{r_1, j}^H \mathbf{A}_{j,  k}^H
    \Bar{\mathbf{H}}_{r_2, j}  \boldsymbol{\Phi}_{r_2}  \Bar{\mathbf{g}}_{r_2, i}
    \Bar{\mathbf{g}}_{r_3, i}^H \boldsymbol{\Phi}_{r_3}^H \Bar{\mathbf{H}}_{r_3, h}^H \mathbf{A}_{h,  k}
 \Bar{\mathbf{H}}_{r_1, h} \boldsymbol{\Phi}_{r_1} \Tilde{\mathbf{g}}_{r_1, k} \}  \\ \label{cross_2}
 &= \eta_{r_1, j, k}^{(2)} \eta_{r_2, j, i}^{(1)} \eta_{r_3, h, i}^{(1)} \eta_{r_1, h, k}^{(2)} f_{r_2, j, i}(\boldsymbol{\Phi}_{r_2}) f_{r_3, h, i}^H(\boldsymbol{\Phi}_{r_3})  \mathrm{Tr}\{\Bar{\mathbf{H}}_{r_1, j}^H \mathbf{A}_{j,  k}^H \mathbf{a}_{N, r_2, j}  \mathbf{a}_{N, r_3, h}^H \mathbf{A}_{h,  k} \Bar{\mathbf{H}}_{r_1, h} \}.
\end{align}

The next non-zero term happens in the case $\hat{\mathcal{Q}}_{r_1, j, k}^{(3)}(\Tilde{\mathbf{H}}_{r_1, j}) \mathcal{Q}_{r_2, j, i}^{(3)}(\Tilde{\mathbf{H}}_{r_2, j}) \mathcal{Q}_{r_3, h, i}^{(1)}(.) \mathcal{Q}_{r_4, h, k}^{(1)}(.) $ with $r_1 = r_2$, hence, we arrive at
\begin{align} \nonumber
\mathbb{E}\{ & \hat{\mathcal{Q}}_{r_1, j, k}^{(3)}(\Tilde{\mathbf{H}}_{r_1, j}) \mathcal{Q}_{r_2, j, i}^{(3)}(\Tilde{\mathbf{H}}_{r_2, j}) \mathcal{Q}_{r_3, h, i}^{(1)}(.) \mathcal{Q}_{r_4, h, k}^{(1)}(.) \}  \\ \nonumber
 & =  \eta_{r_1, j, k}^{(3)} \eta_{r_2, j, i}^{(3)}  \eta_{r_3, h, i}^{(1)}    \eta_{r_4, h, k}^{(1)} 
\Bar{\mathbf{g}}_{r_1, k}^H \boldsymbol{\Phi}_{r_1}^H \mathbb{E}\{ \Tilde{\mathbf{H}}_{r_1, j}^H \mathbf{A}_{j,  k}^H
 \Tilde{\mathbf{H}}_{r_2, j} \} \boldsymbol{\Phi}_{r_2} \Bar{\mathbf{g}}_{r_2, i} 
 \Bar{\mathbf{g}}_{r_3, i}^H \boldsymbol{\Phi}_{r_3}^H \Bar{\mathbf{H}}_{r_3, h}^H
 \Bar{\mathbf{H}}_{r_4, h}  \boldsymbol{\Phi}_{r_4}  \Bar{\mathbf{g}}_{r_4, k} \\ \nonumber
 &= \eta_{r_1, j, k}^{(3)} \eta_{r_1, j, i}^{(3)}  \eta_{r_3, h, i}^{(1)}    \eta_{r_4, h, k}^{(1)}  \mathrm{Tr}\{ \mathbf{A}_{j,  k} \}
 \Bar{\mathbf{g}}_{r_1, k}^H \boldsymbol{\Phi}_{r_1}^H  \boldsymbol{\Phi}_{r_1} \Bar{\mathbf{g}}_{r_1, i} 
 \Bar{\mathbf{g}}_{r_3, i}^H \boldsymbol{\Phi}_{r_3}^H \Bar{\mathbf{H}}_{r_3, h}^H
 \Bar{\mathbf{H}}_{r_4, h}  \boldsymbol{\Phi}_{r_4}  \Bar{\mathbf{g}}_{r_4, k} \\\label{cross_3}
 &=\eta_{r_1, j, k}^{(3)} \eta_{r_1, j, i}^{(3)}  \eta_{r_3, h, i}^{(1)}    \eta_{r_4, h, k}^{(1)}  \mathrm{Tr}\{ \mathbf{A}_{j,  k} \} f_{r_4, h, k}(\boldsymbol{\Phi}_{r_4}) f_{r_3, h, i}^H(\boldsymbol{\Phi}_{r_3})
 \Bar{\mathbf{g}}_{r_1, k}^H  \Bar{\mathbf{g}}_{r_1, i}  \mathbf{a}_{N, r_3, h}^H  \mathbf{a}_{N, r_4, h}.
\end{align}

The next non-zero terms appear when $ \mathcal{Q}_{r_1, j, k}^{(1)}(.) \mathcal{Q}_{r_2, j, i}^{(2)}(\Tilde{\mathbf{g}}_{r_2, i})  \mathcal{Q}_{r_3, h, i}^{(2)}(\Tilde{\mathbf{g}}_{r_3, i}) \mathcal{Q}_{r_4, h, k}^{(1)}(.) $ with $r_2 = r_3$ and could be written as follows:
\begin{align} \nonumber
    \mathbb{E}\{ & \mathcal{Q}_{r_1, j, k}^{(1)}(.) \mathcal{Q}_{r_2, j, i}^{(2)}(\Tilde{\mathbf{g}}_{r_2, i})  \mathcal{Q}_{r_3, h, i}^{(2)}(\Tilde{\mathbf{g}}_{r_3, i}) \mathcal{Q}_{r_4, h, k}^{(1)}(.)  \} \\ \nonumber
    &= \eta_{r_1, j, k}^{(1)}  \eta_{r_2, j, i}^{(2)}  \eta_{r_3, h, i}^{(2)} \eta_{r_4, h, k}^{(1)} 
\Bar{\mathbf{g}}_{r_1, k}^H  \boldsymbol{\Phi}_{r_1} ^H  \Bar{\mathbf{H}}_{r_1, j} ^H
\Bar{\mathbf{H}}_{r_2, j} \boldsymbol{\Phi}_{r_2} \mathbb{E}\{ \Tilde{\mathbf{g}}_{r_2, i} \Tilde{\mathbf{g}}_{r_3, i}^H \} \boldsymbol{\Phi}_{r_3}^H \Bar{\mathbf{H}}_{r_3, h} ^H
 \Bar{\mathbf{H}}_{r_4, h}  \boldsymbol{\Phi}_{r_4}  \Bar{\mathbf{g}}_{r_4, k}  \\\label{cross_4}
 &= \eta_{r_1, j, k}^{(1)}  \eta_{r_2, j, i}^{(2)}  \eta_{r_2, h, i}^{(2)} \eta_{r_4, h, k}^{(1)} f_{r_1, j, k}^H(\boldsymbol{\Phi}_{r_1}) f_{r_4, h, k}(\boldsymbol{\Phi}_{r_4})
 \mathbf{a}_{N, r_1, j}^H  \Bar{\mathbf{H}}_{r_2, j} \Bar{\mathbf{H}}_{r_2, h} ^H
 \mathbf{a}_{N, r_4, h}.
\end{align}

The next term that needs to be calculated is the term when $\mathcal{Q}_{r_1, j, k}^{(1)}(.) \mathcal{Q}_{r_2, j, i}^{(1)}(.) \mathcal{Q}_{r_3, h, i}^{(3)}(\Tilde{\mathbf{H}}_{r_3, h})\hat{\mathcal{Q}}_{r_4, h, k}^{(3)}(\Tilde{\mathbf{H}}_{r_4, h})$ happens with $r_3 = r_4$ and given as follows:
\begin{align} \nonumber
\mathbb{E}\{ & \mathcal{Q}_{r_1, j, k}^{(1)}(.) \mathcal{Q}_{r_2, j, i}^{(1)}(.) \mathcal{Q}_{r_3, h, i}^{(3)}(\Tilde{\mathbf{H}}_{r_3, h})\hat{\mathcal{Q}}_{r_4, h, k}^{(3)}(\Tilde{\mathbf{H}}_{r_4, h}) \} \\ \nonumber
&= \eta_{r_1, j, k}^{(1)} \eta_{r_2, j, i}^{(1)}  \eta_{r_3, h, i}^{(3)} \eta_{r_4, h, k}^{(3)}
\Bar{\mathbf{g}}_{r_1, k}^H \boldsymbol{\Phi}_{r_1}^H \Bar{\mathbf{H}}_{r_1, j} ^H
   \Bar{\mathbf{H}}_{r_2, j}  \boldsymbol{\Phi}_{r_2}  \Bar{\mathbf{g}}_{r_2, i}
   \Bar{\mathbf{g}}_{r_3, i}^H \boldsymbol{\Phi}_{r_3}^H \mathbb{E}\{ \Tilde{\mathbf{H}}_{r_3, h}^H 
     \mathbf{A}_{h,  k} \Tilde{\mathbf{H}}_{r_4, h} \} \boldsymbol{\Phi}_{r_4} \Bar{\mathbf{g}}_{r_4, k} \\ \label{cross_5}
     &= \eta_{r_1, j, k}^{(1)} \eta_{r_2, j, i}^{(1)}  \eta_{r_3, h, i}^{(3)} \eta_{r_3, h, k}^{(3)} \mathrm{Tr}\{ \mathbf{A}_{h,  k} \} f_{r_1, j, k}^H(\boldsymbol{\Phi}_{r_1}) f_{r_2, j, i}(\boldsymbol{\Phi}_{r_2})
\mathbf{a}_{N, r_1, j}^H  \mathbf{a}_{N, r_2, j} 
   \Bar{\mathbf{g}}_{r_3, i}^H  \Bar{\mathbf{g}}_{r_3, k}.
\end{align}

Next, the other non-zero term is when $ \hat{\mathcal{Q}}_{r_1, j, k}^{(3)}(\Tilde{\mathbf{H}}_{r_1, j}) \mathcal{Q}_{r_2, j, i}^{(3)}(\Tilde{\mathbf{H}}_{r_2, j})  \mathcal{Q}_{r_3, h, i}^{(3)}(\Tilde{\mathbf{H}}_{r_3, h}) \hat{\mathcal{Q}}_{r_4, h, k}^{(3)}(\Tilde{\mathbf{H}}_{r_4, h})$ with $r_1 = r_2$ and $r_3 = r_4$ and is given by
\begin{align} \nonumber
    \mathbb{E}\{ & \hat{\mathcal{Q}}_{r_1, j, k}^{(3)}(\Tilde{\mathbf{H}}_{r_1, j}) \mathcal{Q}_{r_2, j, i}^{(3)}(\Tilde{\mathbf{H}}_{r_2, j})  \mathcal{Q}_{r_3, h, i}^{(3)}(\Tilde{\mathbf{H}}_{r_3, h}) \hat{\mathcal{Q}}_{r_4, h, k}^{(3)}(\Tilde{\mathbf{H}}_{r_4, h}) \} \\ \nonumber
    &= \eta_{r_1, j, k}^{(3)} \eta_{r_2, j, i}^{(3)} \eta_{r_3, h, i}^{(3)} \eta_{r_4, h, k}^{(3)}
\Bar{\mathbf{g}}_{r_1, k}^H \boldsymbol{\Phi}_{r_1}^H \mathbb{E}\{  \Tilde{\mathbf{H}}_{r_1, j}^H \mathbf{A}_{j,  k}^H
 \Tilde{\mathbf{H}}_{r_2, j} \}  \boldsymbol{\Phi}_{r_2} \Bar{\mathbf{g}}_{r_2, i} \Bar{\mathbf{g}}_{r_3, i}^H \boldsymbol{\Phi}_{r_3}^H \mathbb{E}\{ \Tilde{\mathbf{H}}_{r_3, h}^H 
     \mathbf{A}_{h,  k} \Tilde{\mathbf{H}}_{r_4, h} \} \boldsymbol{\Phi}_{r_4} \Bar{\mathbf{g}}_{r_4, k} \\\label{cross_6}
     &= \eta_{r_1, j, k}^{(3)} \eta_{r_1, j, i}^{(3)} \eta_{r_3, h, i}^{(3)} \eta_{r_3, h, k}^{(3)} \mathrm{Tr}\{ \mathbf{A}_{j,  k} \} \mathrm{Tr}\{ \mathbf{A}_{h,  k} \}
\Bar{\mathbf{g}}_{r_1, k}^H  \Bar{\mathbf{g}}_{r_1, i} \Bar{\mathbf{g}}_{r_3, i}^H \Bar{\mathbf{g}}_{r_3, k}.
\end{align}

The next non-zero term is when $ \hat{\mathcal{Q}}_{r_1, j, k}^{(2)}(\Tilde{\mathbf{g}}_{r_1, k}) \mathcal{Q}_{r_2, j, i}^{(2)}(\Tilde{\mathbf{g}}_{r_2, i}) \mathcal{Q}_{r_3, h, i}^{(2)}(\Tilde{\mathbf{g}}_{r_3, i}) \hat{\mathcal{Q}}_{r_4, h, k}^{(2)}(\Tilde{\mathbf{g}}_{r_4, k})$ with $r_1 = r_4$ and $r_2 = r_3$ and is given by
\begin{align} \nonumber
     \mathbb{E}\{ & \hat{\mathcal{Q}}_{r_1, j, k}^{(2)}(\Tilde{\mathbf{g}}_{r_1, k}) \mathcal{Q}_{r_2, j, i}^{(2)}(\Tilde{\mathbf{g}}_{r_2, i}) \mathcal{Q}_{r_3, h, i}^{(2)}(\Tilde{\mathbf{g}}_{r_3, i}) \hat{\mathcal{Q}}_{r_4, h, k}^{(2)}(\Tilde{\mathbf{g}}_{r_4, k}) \} \\\nonumber
     &=  \eta_{{r_1}, j, k}^{(2)} \eta_{{r_2}, j, i}^{(2)} \eta_{{r_3}, h, i}^{(2)} \eta_{{r_4}, h, k}^{(2)}
\mathbb{E}\{\Tilde{\mathbf{g}}_{r_1, k}^H \boldsymbol{\Phi}_r^H \Bar{\mathbf{H}}_{r_1, j}^H  \mathbf{A}_{j,  k}^H \Bar{\mathbf{H}}_{{r_2}, j} \boldsymbol{\Phi}_{r_2} \mathbb{E}\{ \Tilde{\mathbf{g}}_{{r_2}, i}
\Tilde{\mathbf{g}}_{{r_3}, i}^H \} \boldsymbol{\Phi}_{r_3}^H  \Bar{\mathbf{H}}_{{r_3}, h}^H 
  \mathbf{A}_{h,  k}  \Bar{\mathbf{H}}_{{r_4}, h} \boldsymbol{\Phi}_{r_4} \Tilde{\mathbf{g}}_{{r_4}, k} \} \\\nonumber
  &= \eta_{{r_1}, j, k}^{(2)} \eta_{{r_2}, j, i}^{(2)} \eta_{{r_2}, h, i}^{(2)} \eta_{{r_1}, h, k}^{(2)}
\mathbb{E}\{\Tilde{\mathbf{g}}_{r_1, k}^H \boldsymbol{\Phi}_r^H \Bar{\mathbf{H}}_{r_1, j}^H  \mathbf{A}_{j,  k}^H \Bar{\mathbf{H}}_{{r_2}, j}   \Bar{\mathbf{H}}_{{r_2}, h}^H 
  \mathbf{A}_{h,  k}  \Bar{\mathbf{H}}_{{r_1}, h} \boldsymbol{\Phi}_{r_1} \Tilde{\mathbf{g}}_{{r_1}, k} \} \\\label{cross_7}
  &= \eta_{{r_1}, j, k}^{(2)} \eta_{{r_2}, j, i}^{(2)} \eta_{{r_2}, h, i}^{(2)} \eta_{{r_1}, h, k}^{(2)} \mathrm{Tr} \{  \Bar{\mathbf{H}}_{r_1, j}^H  \mathbf{A}_{j,  k}^H \Bar{\mathbf{H}}_{{r_2}, j}   \Bar{\mathbf{H}}_{{r_2}, h}^H 
  \mathbf{A}_{h,  k}  \Bar{\mathbf{H}}_{{r_1}, h}  \}.
\end{align}

The next non-zero term is $\hat{\mathcal{Q}}_{r_1, j, k}^{(4)}(\Tilde{\mathbf{g}}_{r_1, k}, \Tilde{\mathbf{H}}_{r_1, j})  \mathcal{Q}_{r_2, j, i}^{(4)}(\Tilde{\mathbf{g}}_{r_2, i}, \Tilde{\mathbf{H}}_{r_2, j}) \mathcal{Q}_{r_3, h, i}^{(4)}(\Tilde{\mathbf{g}}_{r_3, i}, \Tilde{\mathbf{H}}_{r_3, h})  \hat{\mathcal{Q}}_{r_4, h, k}^{(4)}(\Tilde{\mathbf{g}}_{r_4, k}, \Tilde{\mathbf{H}}_{r_4, h})$ with $r_1 = r_2 = r_3 = r_4$ and could be written as follows:
\begin{align} \nonumber
\mathbb{E}\{& \hat{\mathcal{Q}}_{r_1, j, k}^{(4)}(\Tilde{\mathbf{g}}_{r_1, k}, \Tilde{\mathbf{H}}_{r_1, j})  \mathcal{Q}_{r_2, j, i}^{(4)}(\Tilde{\mathbf{g}}_{r_2, i}, \Tilde{\mathbf{H}}_{r_2, j}) \mathcal{Q}_{r_3, h, i}^{(4)}(\Tilde{\mathbf{g}}_{r_3, i}, \Tilde{\mathbf{H}}_{r_3, h})  \hat{\mathcal{Q}}_{r_4, h, k}^{(4)}(\Tilde{\mathbf{g}}_{r_4, k}, \Tilde{\mathbf{H}}_{r_4, h})  \} \\ \label{cross_8}
 &= \eta_{{r_1}, j, k}^{(4)}  \eta_{{r_1}, j, i}^{(4)}  \eta_{{r_1}, h, i}^{(4)}  \eta_{{r_1}, h, k}^{(4)}M \mathrm{Tr}\{ \mathbf{A}_{j,  k}^H \} \mathrm{Tr}\{ \mathbf{A}_{h,  k} \}.
\end{align}

The next term to be calculated is $\hat{\mathcal{Q}}_{r_1, j, k}^{(3)}(\Tilde{\mathbf{H}}_{r_1, j}) \mathcal{Q}_{r_2, j, i}^{(4)}(\Tilde{\mathbf{g}}_{r_2, i}, \Tilde{\mathbf{H}}_{r_2, j}) \mathcal{Q}_{r_3, h, i}^{(4)}(\Tilde{\mathbf{g}}_{r_3, i}, \Tilde{\mathbf{H}}_{r_3, h}) \hat{\mathcal{Q}}_{r_4, h, k}^{(3)}(\Tilde{\mathbf{H}}_{r_4, h})$ and is given as follows:
\begin{align} \nonumber
    \mathbb{E}\{& \hat{\mathcal{Q}}_{r_1, j, k}^{(3)}(\Tilde{\mathbf{H}}_{r_1, j}) \mathcal{Q}_{r_2, j, i}^{(4)}(\Tilde{\mathbf{g}}_{r_2, i}, \Tilde{\mathbf{H}}_{r_2, j}) \mathcal{Q}_{r_3, h, i}^{(4)}(\Tilde{\mathbf{g}}_{r_3, i}, \Tilde{\mathbf{H}}_{r_3, h}) \hat{\mathcal{Q}}_{r_4, h, k}^{(3)}(\Tilde{\mathbf{H}}_{r_4, h}) \} \\\label{cross_9}
 &= \eta_{{r_1}, j, k}^{(3)}  \eta_{{r_1}, j, i}^{(4)} \eta_{{r_1}, h, i}^{(4)} \eta_{{r_1}, h, k}^{(3)} M  \mathrm{Tr}\{ \mathbf{A}_{j,  k}^H\} \mathrm{Tr}\{ \mathbf{A}_{h,  k}\}.
\end{align}
The next part is derivation of the $\hat{\mathcal{Q}}_{r_1, j, k}^{(4)}(\Tilde{\mathbf{g}}_{r_1, k}, \Tilde{\mathbf{H}}_{r_1, j}) \mathcal{Q}_{r_2, j, i}^{(3)}(\Tilde{\mathbf{H}}_{r_2, j}) \mathcal{Q}_{r_3, h, i}^{(3)}(\Tilde{\mathbf{H}}_{r_3, h}) \hat{\mathcal{Q}}_{r_4, h, k}^{(4)}(\Tilde{\mathbf{g}}_{r_4, k}, \Tilde{\mathbf{H}}_{r_4, h})$ and is given as follows:
\begin{align} \nonumber
    \mathbb{E}\{& \hat{\mathcal{Q}}_{r_1, j, k}^{(4)}(\Tilde{\mathbf{g}}_{r_1, k}, \Tilde{\mathbf{H}}_{r_1, j}) \mathcal{Q}_{r_2, j, i}^{(3)}(\Tilde{\mathbf{H}}_{r_2, j}) \mathcal{Q}_{r_3, h, i}^{(3)}(\Tilde{\mathbf{H}}_{r_3, h}) \hat{\mathcal{Q}}_{r_4, h, k}^{(4)}(\Tilde{\mathbf{g}}_{r_4, k}, \Tilde{\mathbf{H}}_{r_4, h})\} \\ \label{cross_10}
 &= \eta_{{r_1}, j, k}^{(4)} \eta_{{r_1}, j, i}^{(3)} \eta_{{r_1}, h, i}^{(3)}  \eta_{{r_1}, h, k}^{(4)} M
 \mathrm{Tr}\{ \mathbf{A}_{j,  k}^H\} \mathrm{Tr}\{ \mathbf{A}_{h,  k}\}.
\end{align}

\subsection*{The Pilot Contamination part of $ \mathbb{E}\{ \hat{\mathbf{q}}_{j, k}^H \mathbf{q}_{j, i}  (\hat{\mathbf{q}}_{h, k}^H \mathbf{q}_{h, i})^H\}$ }

So far, all the non-zero terms of the signal part have been calculated. Next, we move on to the pilot contamination part, the first non-zero term is when $\hat{\mathcal{Q}}_{r_1, j, l}^{(2)}(\Tilde{\mathbf{g}}_{r_1, l})  \mathcal{Q}_{r_2, j, i}^{(1)}(.)  \mathcal{Q}_{r_4, h, k}^{(1)}(.) \hat{\mathcal{Q}}_{r_4, h, l}^{(2)}(\Tilde{\mathbf{g}}_{r_4, l})$ with $r_1 = r_4$. 
\begin{align} \nonumber
    \mathbb{E}\{ & \hat{\mathcal{Q}}_{r_1, j, l}^{(2)}(\Tilde{\mathbf{g}}_{r_1, l})  \mathcal{Q}_{r_2, j, i}^{(1)}(.)  \mathcal{Q}_{r_4, h, k}^{(1)}(.) \hat{\mathcal{Q}}_{r_4, h, l}^{(2)}(\Tilde{\mathbf{g}}_{r_4, l}) \} \\\nonumber
    &= \eta_{r_1, j, l}^{(2)} \eta_{r_2, j, i}^{(1)}  \eta_{r_3, h, i}^{(1)}  \eta_{r_4, h, l}^{(2)} 
 \mathbb{E}\{ \Tilde{\mathbf{g}}_{r_1, l}^H \boldsymbol{\Phi}_{r_1}^H \Bar{\mathbf{H}}_{r_1, j}^H \mathbf{A}_{j,  k}^H
  \Bar{\mathbf{H}}_{r_2, j}  \boldsymbol{\Phi}_{r_2}  \Bar{\mathbf{g}}_{r_2, i}
 \Bar{\mathbf{g}}_{r_3, i}^H \boldsymbol{\Phi}_{r_3}^H \Bar{\mathbf{H}}_{r_3, h}^H
\mathbf{A}_{h,  k} \Bar{\mathbf{H}}_{r_4, h} \boldsymbol{\Phi}_{r_4} \Tilde{\mathbf{g}}_{r_4, l} \} \\ \label{pc_cross_1}
&= \eta_{r_1, j, l}^{(2)} \eta_{r_2, j, i}^{(1)}  \eta_{r_3, h, i}^{(1)}  \eta_{r_1, h, l}^{(2)} f_{r_2, j, i}(\boldsymbol{\Phi}_{r_2}) f_{r_3, h, i}^H(\boldsymbol{\Phi}_{r_3})  \mathrm{Tr}\{\Bar{\mathbf{H}}_{r_1, j}^H \mathbf{A}_{j,  k}^H \mathbf{a}_{N, r_2, j}  \mathbf{a}_{N, r_3, h}^H \mathbf{A}_{h,  k} \Bar{\mathbf{H}}_{r_1, h} \}.
\end{align}
The next one is when $\hat{\mathcal{Q}}_{r_1, j, l}^{(3)}(\Tilde{\mathbf{H}}_{r_1, j})  \mathcal{Q}_{r_2, j, i}^{(3)}(\Tilde{\mathbf{H}}_{r_2, j}) \mathcal{Q}_{r_3, h, i}^{(1)}(.)  \mathcal{Q}_{r_4, h, k}^{(1)}(.)$ happens with $r_1 = r_2$ and is as follows
\begin{align} \nonumber
    \mathbb{E}\{& \hat{\mathcal{Q}}_{r_1, j, l}^{(3)}(\Tilde{\mathbf{H}}_{r_1, j})  \mathcal{Q}_{r_2, j, i}^{(3)}(\Tilde{\mathbf{H}}_{r_2, j}) \mathcal{Q}_{r_4, h, i}^{(1)}(.)  \mathcal{Q}_{r_4, h, k}^{(1)}(.) \} \\ \nonumber
    &= \eta_{r_1, j, l}^{(3)} \eta_{r_2, j, i}^{(3)}  \eta_{r_3, h, i}^{(1)} \eta_{r_4, h, k}^{(1)}
\Bar{\mathbf{g}}_{r_1, l}^H \boldsymbol{\Phi}_{r_1} ^H  \mathbb{E}\{ \Tilde{\mathbf{H}}_{r_1, j} ^H \mathbf{A}_{j,  k}^H
 \Tilde{\mathbf{H}}_{r_2, j} \} \boldsymbol{\Phi}_{r_2} \Bar{\mathbf{g}}_{r_2, i}
 \Bar{\mathbf{g}}_{r_3, i}^H \boldsymbol{\Phi}_{r_3}  ^H \Bar{\mathbf{H}}_{r_3, h} ^H
    \Bar{\mathbf{H}}_{r_4, h}  \boldsymbol{\Phi}_{r_4}  \Bar{\mathbf{g}}_{r_4, k} \\ \label{pc_cross_2}
    &= \eta_{r_1, j, l}^{(3)} \eta_{r_1, j, i}^{(3)}  \eta_{r_3, h, i}^{(1)} \eta_{r_4, h, k}^{(1)} \mathrm{Tr}\{ \mathbf{A}_{j,  k} \} f_{r_4, h, k}(\boldsymbol{\Phi}_{r_4}) f_{r_3, h, i}^H(\boldsymbol{\Phi}_{r_3})
\Bar{\mathbf{g}}_{r_1, l}^H  \Bar{\mathbf{g}}_{r_1, i}
   \mathbf{a}_{N, r_3, h}^H  
    \mathbf{a}_{N, r_1, h}.
\end{align}
The next non-zero term happens when $\mathcal{Q}_{r_1, j, k}^{(1)}(.) \mathcal{Q}_{r_2, j, i}^{(1)}(.)  \mathcal{Q}_{r_3, h, i}^{(3)}(\Tilde{\mathbf{H}}_{r_3, h})\hat{\mathcal{Q}}_{r_4, h, l}^{(3)}(\Tilde{\mathbf{H}}_{r_4, h})$ appears with $r_3 = r_4$ and can be written as follows:
\begin{align} \nonumber
    \mathbb{E}\{& \mathcal{Q}_{r_1, j, k}^{(1)}(.) \mathcal{Q}_{r_2, j, i}^{(1)}(.)  \mathcal{Q}_{r_3, h, i}^{(3)}(\Tilde{\mathbf{H}}_{r_3, h})\hat{\mathcal{Q}}_{r_4, j, l}^{(3)}(\Tilde{\mathbf{H}}_{r_4, j})  \} \\ \nonumber
    &= \eta_{r_1, j, k}^{(1)} \eta_{r_2, j, i}^{(1)}  \eta_{r_3, h, i}^{(3)} \eta_{r_4, h, l}^{(3)}
\Bar{\mathbf{g}}_{r_1, k}^H \boldsymbol{\Phi}_{r_1}^H \Bar{\mathbf{H}}_{r_1, j} ^H
   \Bar{\mathbf{H}}_{r_2, j}  \boldsymbol{\Phi}_{r_2}  \Bar{\mathbf{g}}_{r_2, i}
   \Bar{\mathbf{g}}_{r_3, i}^H \boldsymbol{\Phi}_{r_3}^H \mathbb{E}\{ \Tilde{\mathbf{H}}_{r_3, h}^H 
    \mathbf{A}_{h,  k} \Tilde{\mathbf{H}}_{r_4, h} \} \boldsymbol{\Phi}_{r_4} \Bar{\mathbf{g}}_{r_4, l} \\\label{pc_cross_3}
    &= \eta_{r_1, j, k}^{(1)} \eta_{r_2, j, i}^{(1)}  \eta_{r_3, h, i}^{(3)} \eta_{r_3, h, l}^{(3)} \mathrm{Tr}\{ \mathbf{A}_{h,  k} \} f_{r_1, j, k}^H(\boldsymbol{\Phi}_{r_1}) f_{r_2, j, i}(\boldsymbol{\Phi}_{r_2})
 \mathbf{a}_{N, r_1, j}^H 
   \mathbf{a}_{N, r_2, j} 
   \Bar{\mathbf{g}}_{r_3, i}^H  \Bar{\mathbf{g}}_{r_3, l}.
\end{align}
The next non-zero term is $\hat{\mathcal{Q}}_{r_1, j, l}^{(3)}(\Tilde{\mathbf{H}}_{r_1, j}) \mathcal{Q}_{r_2, j, i}^{(3)}(\Tilde{\mathbf{H}}_{r_2, j}) \mathcal{Q}_{r_3, h, i}^{(3)}(\Tilde{\mathbf{H}}_{r_3, h}) \hat{\mathcal{Q}}_{r_4, h, l}^{(3)}(\Tilde{\mathbf{H}}_{r_4, h})$ with $r_1 = r_2$ and $r_3 = r_4$
\begin{align} \nonumber
    \mathbb{E}\{ & \hat{\mathcal{Q}}_{r_1, j, l}^{(3)}(\Tilde{\mathbf{H}}_{r_1, j}) \mathcal{Q}_{r_2, j, i}^{(3)}(\Tilde{\mathbf{H}}_{r_2, j}) \mathcal{Q}_{r_3, h, i}^{(3)}(\Tilde{\mathbf{H}}_{r_3, h}) \hat{\mathcal{Q}}_{r_4, h, l}^{(3)}(\Tilde{\mathbf{H}}_{r_4, h}) \} \\\nonumber
    &= \eta_{r_1, j, l}^{(3)} \eta_{r_2, j, i}^{(3)}  \eta_{r_3, h, i}^{(3)} \eta_{r_4, h, l}^{(3)}
\Bar{\mathbf{g}}_{r_1, l}^H \boldsymbol{\Phi}_{r_1} ^H  \mathbb{E}\{ \Tilde{\mathbf{H}}_{r_1, j} ^H \mathbf{A}_{j,  k}^H
 \Tilde{\mathbf{H}}_{r_2, j} \} \boldsymbol{\Phi}_{r_2} \Bar{\mathbf{g}}_{r_2, i}
   \Bar{\mathbf{g}}_{r_3, i}^H \boldsymbol{\Phi}_{r_3}^H \mathbb{E}\{ \Tilde{\mathbf{H}}_{r_3, h}^H 
    \mathbf{A}_{h,  k} \Tilde{\mathbf{H}}_{r_4, h} \} \boldsymbol{\Phi}_{r_4} \Bar{\mathbf{g}}_{r_4, l} \\\label{pc_cross_4}
    &= \eta_{r_1, j, l}^{(3)} \eta_{r_1, j, i}^{(3)}  \eta_{r_3, h, i}^{(3)} \eta_{r_3, h, l}^{(3)} \mathrm{Tr}\{ \mathbf{A}_{j,  k} \} \mathrm{Tr}\{ \mathbf{A}_{h,  k} \}
\Bar{\mathbf{g}}_{r_1, k}^H  \Bar{\mathbf{g}}_{r_1, i} \Bar{\mathbf{g}}_{r_3, i}^H \Bar{\mathbf{g}}_{r_3, k}.
\end{align}

The next term to be calculated is as follows:
\begin{align} \nonumber
    \mathbb{E}\{ & \hat{\mathcal{Q}}_{r_1, j, l}^{(4)}(\Tilde{\mathbf{g}}_{r_1, l}, \Tilde{\mathbf{H}}_{r_1, j})  \mathcal{Q}_{r_2, j, i}^{(4)}(\Tilde{\mathbf{g}}_{r_2, i}, \Tilde{\mathbf{H}}_{r_2, j})  \mathcal{Q}_{r_3, h, i}^{(4)}(\Tilde{\mathbf{g}}_{r_3, i}, \Tilde{\mathbf{H}}_{r_3, h})  \hat{\mathcal{Q}}_{r_4, h, l}^{(4)}(\Tilde{\mathbf{g}}_{r_4, l}, \Tilde{\mathbf{H}}_{r_4, h}) \} \\\nonumber
    &= \eta_{{r_1}, j, l}^{(4)}  \eta_{{r_2}, j, i}^{(4)}  \eta_{{r_3}, h, i}^{(4)}  \eta_{{r_4}, h, l}^{(4)}
\mathbb{E}\{\Tilde{\mathbf{g}}_{{r_1}, l}^H \boldsymbol{\Phi}_{r_1}^H \Tilde{\mathbf{H}}_{{r_1}, j}^H \mathbf{A}_{j,  k}^H
 \Tilde{\mathbf{H}}_{{r_2}, j} \boldsymbol{\Phi}_{r_2} \Tilde{\mathbf{g}}_{{r_2}, i}
 \Tilde{\mathbf{g}}_{{r_3}, i}^H \boldsymbol{\Phi}_{r_3}^H \Tilde{\mathbf{H}}_{{r_3}, h}^H
 \mathbf{A}_{h,  k} \Tilde{\mathbf{H}}_{{r_4}, h} \boldsymbol{\Phi}_{r_4} \Tilde{\mathbf{g}}_{{r_4}, l} \} \\\label{pc_cross_5}
 &= \eta_{{r_1}, j, l}^{(4)}  \eta_{{r_1}, j, i}^{(4)}  \eta_{{r_1}, h, i}^{(4)}  \eta_{{r_1}, h, l}^{(4)}M \mathrm{Tr}\{ \mathbf{A}_{j,  k}^H \} \mathrm{Tr}\{ \mathbf{A}_{h,  k} \}.
\end{align}

The next pilot contamination term is given by
\begin{align} \nonumber
    \mathbb{E}\{& \hat{\mathcal{Q}}_{r_1, j, l}^{(3)}(\Tilde{\mathbf{H}}_{r_1, j}) \mathcal{Q}_{r_2, j, i}^{(4)}(\Tilde{\mathbf{g}}_{r_2, i}, \Tilde{\mathbf{H}}_{r_2, j}) \mathcal{Q}_{r_3, h, i}^{(4)}(\Tilde{\mathbf{g}}_{r_3, i}, \Tilde{\mathbf{H}}_{r_3, h}) \hat{\mathcal{Q}}_{r_4, h, l}^{(3)}(\Tilde{\mathbf{H}}_{r_4, h}) \} \\\label{pc_cross_6}
 &= \eta_{{r_1}, j, l}^{(3)}  \eta_{{r_1}, j, i}^{(4)} \eta_{{r_1}, h, i}^{(4)} \eta_{{r_1}, h, l}^{(3)} M  \mathrm{Tr}\{ \mathbf{A}_{j,  k}^H\} \mathrm{Tr}\{ \mathbf{A}_{h,  k}\}.
\end{align}
The last part of the pilot contamination part is given as follows:
\begin{align} \nonumber
    \mathbb{E}\{& \hat{\mathcal{Q}}_{r_1, j, l}^{(4)}(\Tilde{\mathbf{g}}_{r_1, l}, \Tilde{\mathbf{H}}_{r_1, j}) \mathcal{Q}_{r_2, j, i}^{(3)}(\Tilde{\mathbf{H}}_{r_2, j}) \mathcal{Q}_{r_3, h, i}^{(3)}(\Tilde{\mathbf{H}}_{r_3, h}) \hat{\mathcal{Q}}_{r_4, h, l}^{(4)}(\Tilde{\mathbf{g}}_{r_4, l}, \Tilde{\mathbf{H}}_{r_4, h})\} \\ \label{pc_cross_7}
 &= \eta_{{r_1}, j, l}^{(4)} \eta_{{r_1}, j, i}^{(3)} \eta_{{r_1}, h, i}^{(3)}  \eta_{{r_1}, h, l}^{(4)} M
 \mathrm{Tr}\{ \mathbf{A}_{j,  k}^H\} \mathrm{Tr}\{ \mathbf{A}_{h,  k}\}.
\end{align}

Finally, by summing up the calculated terms we will arrive at
\begin{align} \label{E_jkjihkhi}
     \mathbb{E}\{& \hat{\mathbf{q}}_{j, k}^H \mathbf{q}_{j, i}  (\hat{\mathbf{q}}_{h, k}^H \mathbf{q}_{h, i})^H\} = \sum_{r_1=1}^R\sum_{r_2=1}^R\sum_{r_3=1}^R\sum_{r_4=1}^R E_{j,k,h,i}^{\mathrm{CROSS-BS}}(\boldsymbol{\Phi}) + \sum_{l \in \mathcal{P}_k \backslash \{k\}} E_{j,l,h,i}^{\mathrm{PC-CROSS-BS}}(\boldsymbol{\Phi}),
\end{align}
where
\begin{align} \nonumber
    E_{j,k,h,i}^{\mathrm{CROSS}}(\boldsymbol{\Phi}) &= \eta_{r_1, j, k}^{(1)} \eta_{r_2, j, i}^{(1)}  \eta_{r_3, h, i}^{(1)}   \eta_{r_4, h, k}^{(1)}
  \Bar{\mathbf{g}}_{r_1, k}^H \boldsymbol{\Phi}_{r_1}^H \Bar{\mathbf{H}}_{r_1, j}^H
  \Bar{\mathbf{H}}_{r_2, j}  \boldsymbol{\Phi}_{r_2}  \Bar{\mathbf{g}}_{r_2, i}
  \Bar{\mathbf{g}}_{r_3, i}^H \boldsymbol{\Phi}_{r_3}^H \Bar{\mathbf{H}}_{r_3, h}^H   
   \Bar{\mathbf{H}}_{r_4, h}  \boldsymbol{\Phi}_{r_4}  \Bar{\mathbf{g}}_{r_4, k} \\\nonumber
   &+ \eta_{r_1, j, k}^{(2)} \eta_{r_2, j, i}^{(1)} \eta_{r_3, h, i}^{(1)} \eta_{r_1, h, k}^{(2)} f_{r_2, j, i}(\boldsymbol{\Phi}_{r_2}) f_{r_3, h, i}^H(\boldsymbol{\Phi}_{r_3})  \mathrm{Tr}\{\Bar{\mathbf{H}}_{r_1, j}^H \mathbf{A}_{j,  k}^H \mathbf{a}_{N, r_2, j}  \mathbf{a}_{N, r_3, h}^H \mathbf{A}_{h,  k} \Bar{\mathbf{H}}_{r_1, h} \} \\\nonumber
   &+ \eta_{r_1, j, k}^{(3)} \eta_{r_1, j, i}^{(3)}  \eta_{r_3, h, i}^{(1)}    \eta_{r_4, h, k}^{(1)}  \mathrm{Tr}\{ \mathbf{A}_{j,  k} \} f_{r_4, h, k}(\boldsymbol{\Phi}_{r_4}) f_{r_3, h, i}^H(\boldsymbol{\Phi}_{r_3})
 \Bar{\mathbf{g}}_{r_1, k}^H  \Bar{\mathbf{g}}_{r_1, i}  \mathbf{a}_{N, r_3, h}^H  \mathbf{a}_{N, r_4, h} \\\nonumber
 &+ \eta_{r_1, j, k}^{(1)}  \eta_{r_2, j, i}^{(2)}  \eta_{r_2, h, i}^{(2)} \eta_{r_4, h, k}^{(1)} f_{r_1, j, k}^H(\boldsymbol{\Phi}_{r_1}) f_{r_4, h, k}(\boldsymbol{\Phi}_{r_4})
 \mathbf{a}_{N, r_1, j}^H  \Bar{\mathbf{H}}_{r_2, j} \Bar{\mathbf{H}}_{r_2, h} ^H
 \mathbf{a}_{N, r_4, h} \\\nonumber
 &+ \eta_{r_1, j, k}^{(1)} \eta_{r_2, j, i}^{(1)}  \eta_{r_3, h, i}^{(3)} \eta_{r_3, h, k}^{(3)} \mathrm{Tr}\{ \mathbf{A}_{h,  k} \} f_{r_1, j, k}^H(\boldsymbol{\Phi}_{r_1}) f_{r_2, j, i}(\boldsymbol{\Phi}_{r_2})
\mathbf{a}_{N, r_1, j}^H  \mathbf{a}_{N, r_2, j} 
   \Bar{\mathbf{g}}_{r_3, i}^H  \Bar{\mathbf{g}}_{r_3, k} \\\nonumber
   &+ \eta_{r_1, j, k}^{(3)} \eta_{r_1, j, i}^{(3)} \eta_{r_3, h, i}^{(3)} \eta_{r_3, h, k}^{(3)} \mathrm{Tr}\{ \mathbf{A}_{j,  k} \} \mathrm{Tr}\{ \mathbf{A}_{h,  k} \}
\Bar{\mathbf{g}}_{r_1, k}^H  \Bar{\mathbf{g}}_{r_1, i} \Bar{\mathbf{g}}_{r_3, i}^H \Bar{\mathbf{g}}_{r_3, k} \\\nonumber
&+ \eta_{{r_1}, j, k}^{(2)} \eta_{{r_2}, j, i}^{(2)} \eta_{{r_2}, h, i}^{(2)} \eta_{{r_1}, h, k}^{(2)} \mathrm{Tr} \{  \Bar{\mathbf{H}}_{r_1, j}^H  \mathbf{A}_{j,  k}^H \Bar{\mathbf{H}}_{{r_2}, j}   \Bar{\mathbf{H}}_{{r_2}, h}^H 
  \mathbf{A}_{h,  k}  \Bar{\mathbf{H}}_{{r_1}, h}  \} \\\nonumber
  &+ \eta_{{r_1}, j, k}^{(4)}  \eta_{{r_1}, j, i}^{(4)}  \eta_{{r_1}, h, i}^{(4)}  \eta_{{r_1}, h, k}^{(4)}M \mathrm{Tr}\{ \mathbf{A}_{j,  k}^H \} \mathrm{Tr}\{ \mathbf{A}_{h,  k} \} \\\nonumber
  &+ \eta_{{r_1}, j, k}^{(3)}  \eta_{{r_1}, j, i}^{(4)} \eta_{{r_1}, h, i}^{(4)} \eta_{{r_1}, h, k}^{(3)} M  \mathrm{Tr}\{ \mathbf{A}_{j,  k}^H\} \mathrm{Tr}\{ \mathbf{A}_{h,  k}\} \\ 
  &+ \eta_{{r_1}, j, k}^{(4)} \eta_{{r_1}, j, i}^{(3)} \eta_{{r_1}, h, i}^{(3)}  \eta_{{r_1}, h, k}^{(4)} M
 \mathrm{Tr}\{ \mathbf{A}_{j,  k}^H\} \mathrm{Tr}\{ \mathbf{A}_{h,  k}\},
\end{align}
and
\begin{align} \nonumber
    E_{j,l,h,i}^{\mathrm{PC-CROSS}}(\boldsymbol{\Phi}) &= \eta_{r_1, j, l}^{(2)} \eta_{r_2, j, i}^{(1)}  \eta_{r_3, h, i}^{(1)}  \eta_{r_1, h, l}^{(2)} f_{r_2, j, i}(\boldsymbol{\Phi}_{r_2}) f_{r_3, h, i}^H(\boldsymbol{\Phi}_{r_3})  \mathrm{Tr}\{\Bar{\mathbf{H}}_{r_1, j}^H \mathbf{A}_{j,  k}^H \mathbf{a}_{N, r_2, j}  \mathbf{a}_{N, r_3, h}^H \mathbf{A}_{h,  k} \Bar{\mathbf{H}}_{r_1, h} \} \\ \nonumber
    &+ \eta_{r_1, j, l}^{(3)} \eta_{r_1, j, i}^{(3)}  \eta_{r_3, h, i}^{(1)} \eta_{r_4, h, k}^{(1)} \mathrm{Tr}\{ \mathbf{A}_{j,  k} \} f_{r_4, h, k}(\boldsymbol{\Phi}_{r_4}) f_{r_3, h, i}^H(\boldsymbol{\Phi}_{r_3})
\Bar{\mathbf{g}}_{r_1, l}^H  \Bar{\mathbf{g}}_{r_1, i}
   \mathbf{a}_{N, r_3, h}^H  
    \mathbf{a}_{N, r_1, h} \\\nonumber
    &+ \eta_{r_1, j, k}^{(1)} \eta_{r_2, j, i}^{(1)}  \eta_{r_3, h, i}^{(3)} \eta_{r_3, h, l}^{(3)} \mathrm{Tr}\{ \mathbf{A}_{h,  k} \} f_{r_1, j, k}^H(\boldsymbol{\Phi}_{r_1}) f_{r_2, j, i}(\boldsymbol{\Phi}_{r_2})
 \mathbf{a}_{N, r_1, j}^H 
   \mathbf{a}_{N, r_2, j} 
   \Bar{\mathbf{g}}_{r_3, i}^H  \Bar{\mathbf{g}}_{r_3, l} \\\nonumber
   &+ \eta_{r_1, j, l}^{(3)} \eta_{r_1, j, i}^{(3)}  \eta_{r_3, h, i}^{(3)} \eta_{r_3, h, l}^{(3)} \mathrm{Tr}\{ \mathbf{A}_{j,  k} \} \mathrm{Tr}\{ \mathbf{A}_{h,  k} \}
\Bar{\mathbf{g}}_{r_1, k}^H  \Bar{\mathbf{g}}_{r_1, i} \Bar{\mathbf{g}}_{r_3, i}^H \Bar{\mathbf{g}}_{r_3, k} \\\nonumber
&+ \eta_{{r_1}, j, l}^{(4)}  \eta_{{r_1}, j, i}^{(4)}  \eta_{{r_1}, h, i}^{(4)}  \eta_{{r_1}, h, l}^{(4)}M \mathrm{Tr}\{ \mathbf{A}_{j,  k}^H \} \mathrm{Tr}\{ \mathbf{A}_{h,  k} \} \\\nonumber
&+ \eta_{{r_1}, j, l}^{(3)}  \eta_{{r_1}, j, i}^{(4)} \eta_{{r_1}, h, i}^{(4)} \eta_{{r_1}, h, l}^{(3)} M  \mathrm{Tr}\{ \mathbf{A}_{j,  k}^H\} \mathrm{Tr}\{ \mathbf{A}_{h,  k}\}\\ 
&+ \eta_{{r_1}, j, l}^{(4)} \eta_{{r_1}, j, i}^{(3)} \eta_{{r_1}, h, i}^{(3)}  \eta_{{r_1}, h, l}^{(4)} M
 \mathrm{Tr}\{ \mathbf{A}_{j,  k}^H\} \mathrm{Tr}\{ \mathbf{A}_{h,  k}\}.
\end{align}

\subsection{Derivation of $  \mathbb{E}\{\hat{\mathbf{q}}_{j, k}^H \mathbf{q}_{j, k}  (\hat{\mathbf{q}}_{h, k}^H \mathbf{q}_{h, k})^H\} $}
Next, we will focus on the derivation of $  \mathbb{E}\{\hat{\mathbf{q}}_{j, k}^H \mathbf{q}_{j, k}  (\hat{\mathbf{q}}_{h, k}^H \mathbf{q}_{h, k})^H\} $, i.e., the off-diagonal entries of (\ref{matrixE}) with $i = k$. The result obtained in (\ref{E_jkjihkhi}) is still valid for the case with $i = k$ plus some extra terms that are zero for the case when $i \neq k$ which are not zero when $i = k$. 

In what follows we will focus on non-zero terms of (\ref{offdiag}) with $i = k$ as follows. The new non-zero terms are listed as follows which the calculation is omitted here for bravery. 

\begin{align} \nonumber
    \mathbb{E}\{ & \hat{\mathcal{Q}}_{r_1, j, k}^{(2)}(\Tilde{\mathbf{g}}_{r_1, k}) \mathcal{Q}_{r_2, j, k}^{(2)}(\Tilde{\mathbf{g}}_{r_2, k}) \mathcal{Q}_{r_3, h, k}^{(1)}(.) \mathcal{Q}_{r_4, h, k}^{(1)}(.) \} \\ \label{self_term_new_1}
    &= (\eta_{r_1, j, k}^{(2)})^2  \eta_{r_3, h, k}^{(1)} \eta_{r_4, h, k}^{(1)}
M g_{r_1, j}
     \Bar{\mathbf{g}}_{r_3, k}^H \boldsymbol{\Phi}_{r_3}^H \Bar{\mathbf{H}}_{r_3, h}^H
     \Bar{\mathbf{H}}_{r_4, h}  \boldsymbol{\Phi}_{r_4}  \Bar{\mathbf{g}}_{r_4, k} ,
\end{align}
\begin{align} \nonumber
     \mathbb{E}\{& \mathcal{Q}_{r_1, j, k}^{(1)}(.) \mathcal{Q}_{r_2, j, k}^{(1)}(.)  \mathcal{Q}_{r_3, h, k}^{(2)}(\Tilde{\mathbf{g}}_{r_3, k}) \hat{\mathcal{Q}}_{r_4, h, k}^{(2)}(\Tilde{\mathbf{g}}_{r_4, k}) \} \\ \label{self_term_new_2}
     &= \eta_{r_1, j, k}^{(1)}  \eta_{r_2, j, k}^{(1)}  (\eta_{r_3, h, k}^{(2)})^2 M g_{r_3, h}
      \Bar{\mathbf{g}}_{r_1, k}^H \boldsymbol{\Phi}_{r_1}^H \Bar{\mathbf{H}}_{r_1, j} ^H
      \Bar{\mathbf{H}}_{r_2, j}  \boldsymbol{\Phi}_{r_2}  \Bar{\mathbf{g}}_{r_2, k},
\end{align}
\begin{align} \nonumber
      \mathbb{E}\{ & \hat{\mathcal{Q}}_{r_1, j, k}^{(4)}(\Tilde{\mathbf{g}}_{r_1, k}, \Tilde{\mathbf{H}}_{r_1, j}) \mathcal{Q}_{r_2, j, k}^{(4)}(\Tilde{\mathbf{g}}_{r_2, k}, \Tilde{\mathbf{H}}_{r_2, j})   \mathcal{Q}_{r_3, h, k}^{(1)}(.) \mathcal{Q}_{r_4, h, k}^{(1)}(.) \} \\\label{self_term_new_3}
      &= (\eta_{r_1, j, k}^{(4)})^2   \eta_{r_3, h, k}^{(1)} \eta_{r_4, h, k}^{(1)}
      M\mathrm{Tr}\{   \mathbf{A}_{j,  k}^H \}
     \Bar{\mathbf{g}}_{r_3, k}^H \boldsymbol{\Phi}_{r_3}^H \Bar{\mathbf{H}}_{r_3, h}^H
     \Bar{\mathbf{H}}_{r_4, h}  \boldsymbol{\Phi}_{r_4}  \Bar{\mathbf{g}}_{r_4, k} ,
\end{align}
\begin{align} \nonumber
     \mathbb{E}\{&  \mathcal{Q}_{r_1, j, k}^{(1)}(.) \mathcal{Q}_{r_2, j, k}^{(1)}(.) \mathcal{Q}_{r_3, h, k}^{(4)}(\Tilde{\mathbf{g}}_{r_3, k}, \Tilde{\mathbf{H}}_{r_3, h}) \hat{\mathcal{Q}}_{r_4, h, k}^{(4)}(\Tilde{\mathbf{g}}_{r_4, k}, \Tilde{\mathbf{H}}_{r_4, h}) \} \\ \label{self_term_new_4}
     &=  \eta_{r_1, j, k}^{(1)}  \eta_{r_2, j, k}^{(1)} (\eta_{r_3, h, k}^{(4)})^2  M\mathrm{Tr}\{   \mathbf{A}_{h,  k}^H \}
           \Bar{\mathbf{g}}_{r_1, k}^H \boldsymbol{\Phi}_{r_1}^H \Bar{\mathbf{H}}_{r_1, j} ^H
      \Bar{\mathbf{H}}_{r_2, j}  \boldsymbol{\Phi}_{r_2}  \Bar{\mathbf{g}}_{r_2, k}  ,
\end{align}
\begin{align} \nonumber
      \mathbb{E}\{& \hat{\mathcal{Q}}_{j, k}^{(5)}(\Tilde{\mathbf{h}}_{j, k}) \mathcal{Q}_{j, k}^{(5)}(\Tilde{\mathbf{h}}_{j, k}) \mathcal{Q}_{r_3, h, k}^{(1)}(.) \mathcal{Q}_{r_4, h, k}^{(1)}(.) \} \\\label{self_new_5}
      &= (\eta_{j, k}^{(5)})^2 \eta_{r_3, h, k}^{(1)} \eta_{r_4, h, k}^{(1)}
      \mathrm{Tr}\{ \mathbf{A}_{j,  k} \}
       \Bar{\mathbf{g}}_{r_3, k}^H \boldsymbol{\Phi}_{r_3}^H \Bar{\mathbf{H}}_{r_3, h}^H
     \Bar{\mathbf{H}}_{r_4, h}  \boldsymbol{\Phi}_{r_4}  \Bar{\mathbf{g}}_{r_4, k},
\end{align}
\begin{align} \nonumber
     \mathbb{E}\{& \mathcal{Q}_{r_1, j, k}^{(1)}(.)  \mathcal{Q}_{r_2, j, k}^{(1)}(.) \mathcal{Q}_{h, k}^{(5)}(\Tilde{\mathbf{h}}_{h, k})\hat{\mathcal{Q}}_{j, k}^{(5)}(\Tilde{\mathbf{h}}_{h, k}) \} \\ \label{self_term_6}
     &= \eta_{r_1, j, k}^{(1)}  \eta_{r_2, j, k}^{(1)}  (\eta_{h, k}^{(5)})^2 \mathrm{Tr}\{ \mathbf{A}_{h,  k} \}
      \Bar{\mathbf{g}}_{r_1, k}^H \boldsymbol{\Phi}_{r_1}^H \Bar{\mathbf{H}}_{r_1, j} ^H
      \Bar{\mathbf{H}}_{r_2, j}  \boldsymbol{\Phi}_{r_2}  \Bar{\mathbf{g}}_{r_2, k} ,
\end{align}
\begin{align} \nonumber
      \mathbb{E}\{& \hat{\mathcal{Q}}_{j, k}^{(5)}(\Tilde{\mathbf{h}}_{j, k}) \mathcal{Q}_{j, k}^{(5)}(\Tilde{\mathbf{h}}_{j, k}) \hat{\mathcal{Q}}_{j, k}^{(5)}(\Tilde{\mathbf{h}}_{h, k}) \hat{\mathcal{Q}}_{j, k}^{(5)}(\Tilde{\mathbf{h}}_{h, k}) \} \\ \label{self_term_7}
      &= (\eta_{j, k}^{(5)}  \eta_{h, k}^{(5)})^2
        \mathrm{Tr}\{   \mathbf{A}_{j,  k}^H \} \mathrm{Tr}\{   \mathbf{A}_{h,  k}\},
\end{align}
\begin{align} \nonumber
    \mathbb{E}\{& \hat{\mathcal{Q}}_{r_1, j, k}^{(2)}(\Tilde{\mathbf{g}}_{r_1, k})  \mathcal{Q}_{r_2, j, i}^{(2)}(\Tilde{\mathbf{g}}_{r_2, i})  \mathcal{Q}_{r_3, h, i}^{(4)}(\Tilde{\mathbf{g}}_{r_3, i}, \Tilde{\mathbf{H}}_{r_3, h}) \hat{\mathcal{Q}}_{r_4, h, k}^{(4)}(\Tilde{\mathbf{g}}_{r_4, k}, \Tilde{\mathbf{H}}_{r_4, h}) \} \\\label{self_term_8}
  &= ( \eta_{r_1, j, k}^{(2)}   \eta_{r_3, h, k}^{(4)})^2 M^2 \mathrm{Tr}\{   \mathbf{A}_{h,  k}\} g_{r_1, j} ,
\end{align}
\begin{align} \nonumber     
\mathbb{E}\{&  \hat{\mathcal{Q}}_{r_1, j, k}^{(4)}(\Tilde{\mathbf{g}}_{r_1, k}, \Tilde{\mathbf{H}}_{r_1, j}) \mathcal{Q}_{r_2, j, k}^{(4)}(\Tilde{\mathbf{g}}_{r_2, k}, \Tilde{\mathbf{H}}_{r_2, j}) \mathcal{Q}_{r_3, h, k}^{(2)}(\Tilde{\mathbf{g}}_{r_3, k}) \hat{\mathcal{Q}}_{r_4, h, k}^{(2)}(\Tilde{\mathbf{g}}_{r_4, k})  \} \\ \nonumber
    &= (\eta_{r_1, j, k}^{(4)}  \eta_{r_3, h, k}^{(2)} )^2
      \mathbb{E}\{ \Tilde{\mathbf{g}}_{r_3, k}^H \boldsymbol{\Phi}_{r_3}^H \Bar{\mathbf{H}}_{r_3, h} ^H \mathbf{A}_{h,  k}^H
      \Bar{\mathbf{H}}_{r_3, h} \boldsymbol{\Phi}_{r_3} \Tilde{\mathbf{g}}_{r_3, k} \} \\\label{self_term_9}
     &=  (\eta_{r_1, j, k}^{(4)}  \eta_{r_3, h, k}^{(2)} )^2 M^2 \mathrm{Tr}\{   \mathbf{A}_{j,  k}^H\} g_{r_3, j},
\end{align}
\begin{align} \nonumber
    \mathbb{E}\{& \hat{\mathcal{Q}}_{j, k}^{(5)}(\Tilde{\mathbf{h}}_{j, k}) \mathcal{Q}_{j, k}^{(5)}(\Tilde{\mathbf{h}}_{j, k}) \} \mathcal{Q}_{r_3, h, k}^{(2)}(\Tilde{\mathbf{g}}_{r_3, k}) \hat{\mathcal{Q}}_{r_4, h, k}^{(2)}(\Tilde{\mathbf{g}}_{r_4, k}) \\\nonumber
    &= \eta_{j, k}^{(5)}  \eta_{j, k}^{(5)}  \eta_{r_3, h, k}^{(2)} \eta_{r_4, h, k}^{(2)}  \mathbb{E}\{ \Tilde{\mathbf{h}}_{j, k}^H \mathbf{A}_{j,  k}^H \Tilde{\mathbf{h}}_{j, k} \}  \mathbb{E}\{    \Tilde{\mathbf{g}}_{r_3, k}^H \boldsymbol{\Phi}_{r_3}^H \Bar{\mathbf{H}}_{r_3, h} ^H \mathbf{A}_{h,  k}^H
      \Bar{\mathbf{H}}_{r_4, h} \boldsymbol{\Phi}_{r_4} \Tilde{\mathbf{g}}_{r_4, k} \} \\ \label{self_term_10}
      &=  (\eta_{j, k}^{(5)}  \eta_{r_3, h, k}^{(2)})^2  M g_{r_3, h}  \mathrm{Tr}\{ \mathbf{A}_{j,  k}^H \},
\end{align}
\begin{align} \nonumber
     \mathbb{E}\{ & \hat{\mathcal{Q}}_{r_1, j, k}^{(2)}(\Tilde{\mathbf{g}}_{r_1, k}) \mathcal{Q}_{r_2, j, k}^{(2)}(\Tilde{\mathbf{g}}_{r_2, k})  \mathcal{Q}_{h, k}^{(5)}(\Tilde{\mathbf{h}}_{h, k})\hat{\mathcal{Q}}_{j, k}^{(5)}(\Tilde{\mathbf{h}}_{h, k}) \}  \\ \label{self_term_11}
       &=  (\eta_{r_1, j, k}^{(2)} \eta_{h, k}^{(5)})^2 M g_{r_1, k} \mathrm{Tr}\{  \mathbf{A}_{h,  k} \},
\end{align}
\begin{align} \nonumber
   \mathbb{E}\{& \hat{\mathcal{Q}}_{r_1, j, k}^{(4)}(\Tilde{\mathbf{g}}_{r_1, k}, \Tilde{\mathbf{H}}_{r_1, j})  \mathcal{Q}_{r_2, j, i}^{(4)}(\Tilde{\mathbf{g}}_{r_2, k}, \Tilde{\mathbf{H}}_{r_2, j})  \mathcal{Q}_{h, k}^{(5)}(\Tilde{\mathbf{h}}_{h, k})  \hat{\mathcal{Q}}_{j, k}^{(5)}(\Tilde{\mathbf{h}}_{h, k}) \} \\ \label{self_term_12}
       &= ( \eta_{r_1, j, k}^{(4)} \eta_{h, k}^{(5)})^2 M \mathrm{Tr}\{ \mathbf{A}_{j,  k}^H \} \mathrm{Tr}\{ \mathbf{A}_{h,  k}^H \},
\end{align}
\begin{align} \nonumber
    \mathbb{E}\{ &  \hat{\mathcal{Q}}_{j, k}^{(5)}(\Tilde{\mathbf{h}}_{j, k}) \mathcal{Q}_{j, k}^{(5)}(\Tilde{\mathbf{h}}_{j, k}) \mathcal{Q}_{r_3, h, k}^{(4)}(\Tilde{\mathbf{g}}_{r_3, k}, \Tilde{\mathbf{H}}_{r_3, h}) \hat{\mathcal{Q}}_{r_4, h, k}^{(4)}(\Tilde{\mathbf{g}}_{r_4, k}, \Tilde{\mathbf{H}}_{r_4, h}) \}  \\  \label{self_term_12}
      &= ( \eta_{j, k}^{(5)}   \eta_{r_3, h, k}^{(4)})^2 M \mathrm{Tr}\{ \mathbf{A}_{j,  k}^H \} \mathrm{Tr}\{ \mathbf{A}_{h,  k} \},
\end{align}
\begin{align} \nonumber
    \mathbb{E}\{& \hat{\mathcal{Q}}_{j, k}^{(5)}(\Tilde{\mathbf{h}}_{j, k}) \mathcal{Q}_{j, k}^{(5)}(\Tilde{\mathbf{h}}_{j, k}) \mathcal{Q}_{r_3, h, k}^{(3)}(\Tilde{\mathbf{H}}_{r_3, h}) \hat{\mathcal{Q}}_{r_4, h, k}^{(3)}(\Tilde{\mathbf{H}}_{r_4, h})  \} \\\label{self_term_13}
      &=  (\eta_{j, k}^{(5)}   \eta_{r_3, h, k}^{(3)})^2 M \mathrm{Tr}\{ \mathbf{A}_{j,  k}^H \}  \mathrm{Tr}\{ \mathbf{A}_{h,  k} \},
\end{align}
\begin{align} \nonumber
    \mathbb{E}\{&  \hat{\mathcal{Q}}_{r_1, j, k}^{(3)}(\Tilde{\mathbf{H}}_{r_1, j}) \mathcal{Q}_{r_2, j, k}^{(3)}(\Tilde{\mathbf{H}}_{r_2, j}) \mathcal{Q}_{h, k}^{(5)}(\Tilde{\mathbf{h}}_{h, k}) \hat{\mathcal{Q}}_{j, k}^{(5)}(\Tilde{\mathbf{h}}_{h, k}) \} \\ \label{self_term_14}
    &=  (\eta_{h, k}^{(5)}   \eta_{r_3, j, k}^{(3)})^2 M \mathrm{Tr}\{ \mathbf{A}_{j,  k}^H \}  \mathrm{Tr}\{ \mathbf{A}_{h,  k} \},
\end{align}
\begin{align} \nonumber
    \mathbb{E}\{& \hat{\mathcal{Q}}_{r_1, j, k}^{(4)}(\Tilde{\mathbf{g}}_{r_1, k}, \Tilde{\mathbf{H}}_{r_1, j})  \mathcal{Q}_{r_2, j, k}^{(4)}(\Tilde{\mathbf{g}}_{r_2, k}, \Tilde{\mathbf{H}}_{r_2, j}) \mathcal{Q}_{r_3, h, k}^{(3)}(\Tilde{\mathbf{H}}_{r_3, h}) \hat{\mathcal{Q}}_{r_4, h, k}^{(3)}(\Tilde{\mathbf{H}}_{r_4, h})  \} \\ \label{self_term_15}
     &= (\eta_{r_1, j, k}^{(4)}  \eta_{r_3, h, k}^{(3)})^2 M^2
     \mathrm{Tr}\{ \mathbf{A}_{j,  k}^H \}  \mathrm{Tr}\{ \mathbf{A}_{h,  k} \} ,
\end{align}
\begin{align} \nonumber
    \mathbb{E}\{&  \hat{\mathcal{Q}}_{r_1, j, k}^{(3)}(\Tilde{\mathbf{H}}_{r_1, j}) \mathcal{Q}_{r_2, j, k}^{(3)}(\Tilde{\mathbf{H}}_{r_2, j}) \mathcal{Q}_{r_3, h, k}^{(4)}(\Tilde{\mathbf{g}}_{r_3, k}, \Tilde{\mathbf{H}}_{r_3, h}) \hat{\mathcal{Q}}_{r_4, h, k}^{(4)}(\Tilde{\mathbf{g}}_{r_4, k}, \Tilde{\mathbf{H}}_{r_4, h}) \} \\ \label{self_term_16}
    &= (\eta_{r_1, j, k}^{(3)}  \eta_{r_3, h, k}^{(4)})^2 M^2
     \mathrm{Tr}\{ \mathbf{A}_{j,  k}^H \}  \mathrm{Tr}\{ \mathbf{A}_{h,  k} \} ,
\end{align}
\begin{align} \nonumber
    \mathbb{E}\{& \hat{\mathcal{Q}}_{r_1, j, k}^{(2)}(\Tilde{\mathbf{g}}_{r_1, k})  \mathcal{Q}_{r_2, j, k}^{(2)}(\Tilde{\mathbf{g}}_{r_2, k}) \mathcal{Q}_{r_3, h, k}^{(3)}(\Tilde{\mathbf{H}}_{r_3, h}) \hat{\mathcal{Q}}_{r_4, h, k}^{(3)}(\Tilde{\mathbf{H}}_{r_4, h}) \} \\ \label{self_term_17}
     &=  (\eta_{r_1, j, k}^{(2)}  \eta_{r_3, h, k}^{(3)})^2
     M^2 \mathrm{Tr}\{ \mathbf{A}_{h,  k} \} g_{r_1, j},
\end{align}
\begin{align} \nonumber
    \mathbb{E}\{& \hat{\mathcal{Q}}_{r_1, j, k}^{(3)}(\Tilde{\mathbf{H}}_{r_1, j}) \mathcal{Q}_{r_2, j, i}^{(3)}(\Tilde{\mathbf{H}}_{r_2, j}) \mathcal{Q}_{r_3, h, k}^{(2)}(\Tilde{\mathbf{g}}_{r_3, k}) \hat{\mathcal{Q}}_{r_4, h, k}^{(2)}(\Tilde{\mathbf{g}}_{r_4, k}) \} \\ \label{self_term_18}
     &=  (\eta_{r_1, j, k}^{(3)}  \eta_{r_3, h, k}^{(2)})^2
     M^2 \mathrm{Tr}\{ \mathbf{A}_{j,  k} \} g_{r_3, h}.
\end{align}

\subsection*{The Pilot Contamination part of $ \mathbb{E}\{\hat{\mathbf{q}}_{j, k}^H \mathbf{q}_{j, k}  (\hat{\mathbf{q}}_{h, k}^H \mathbf{q}_{h, k})^H\}$ }

The pilot contamination terms of this case are given by
\begin{align} \nonumber
    \mathbb{E}\{& \hat{\mathcal{Q}}_{j, k}^{(5)}(\Tilde{\mathbf{h}}_{j, k}) \mathcal{Q}_{j, k}^{(5)}(\Tilde{\mathbf{h}}_{j, k}) \mathcal{Q}_{r_3, h, k}^{(3)}(\Tilde{\mathbf{H}}_{r_3, h}) \hat{\mathcal{Q}}_{r_4, h, l}^{(3)}(\Tilde{\mathbf{H}}_{r_4, h})  \} \\\label{self_term_13}
      &=  (\eta_{j, k}^{(5)})^2  \eta_{r_3, h, k}^{(3)} \eta_{r_3, h, l}^{(3)} M \mathrm{Tr}\{ \mathbf{A}_{j,  k}^H \}  \mathrm{Tr}\{ \mathbf{A}_{h,  k} \},
\end{align}
\begin{align} \nonumber
    \mathbb{E}\{&  \hat{\mathcal{Q}}_{r_1, j, l}^{(3)}(\Tilde{\mathbf{H}}_{r_1, j}) \mathcal{Q}_{r_2, j, k}^{(3)}(\Tilde{\mathbf{H}}_{r_2, j}) \mathcal{Q}_{h, k}^{(5)}(\Tilde{\mathbf{h}}_{h, k}) \hat{\mathcal{Q}}_{j, k}^{(5)}(\Tilde{\mathbf{h}}_{h, k}) \} \\ \label{self_term_14}
    &=  (\eta_{h, k}^{(5)})^2   \eta_{r_3, j, k}^{(3)} \eta_{r_3, j, l}^{(3)}  M \mathrm{Tr}\{ \mathbf{A}_{j,  k}^H \}  \mathrm{Tr}\{ \mathbf{A}_{h,  k} \},
\end{align}
\begin{align} \nonumber
    \mathbb{E}\{& \hat{\mathcal{Q}}_{r_1, j, k}^{(4)}(\Tilde{\mathbf{g}}_{r_1, k}, \Tilde{\mathbf{H}}_{r_1, j})  \mathcal{Q}_{r_2, j, k}^{(4)}(\Tilde{\mathbf{g}}_{r_2, k}, \Tilde{\mathbf{H}}_{r_2, j}) \mathcal{Q}_{r_3, h, k}^{(3)}(\Tilde{\mathbf{H}}_{r_3, h}) \hat{\mathcal{Q}}_{r_4, h, l}^{(3)}(\Tilde{\mathbf{H}}_{r_4, h})  \} \\ \label{self_term_15}
     &= (\eta_{r_1, j, k}^{(4)})^2  \eta_{r_3, h, k}^{(3)} \eta_{r_3, h, l}^{(3)}  M^2
     \mathrm{Tr}\{ \mathbf{A}_{j,  k}^H \}  \mathrm{Tr}\{ \mathbf{A}_{h,  k} \} ,
\end{align}
\begin{align} \nonumber
    \mathbb{E}\{&  \hat{\mathcal{Q}}_{r_1, j, l}^{(3)}(\Tilde{\mathbf{H}}_{r_1, j}) \mathcal{Q}_{r_2, j, k}^{(3)}(\Tilde{\mathbf{H}}_{r_2, j}) \mathcal{Q}_{r_3, h, k}^{(4)}(\Tilde{\mathbf{g}}_{r_3, k}, \Tilde{\mathbf{H}}_{r_3, h}) \hat{\mathcal{Q}}_{r_4, h, k}^{(4)}(\Tilde{\mathbf{g}}_{r_4, k}, \Tilde{\mathbf{H}}_{r_4, h}) \} \\ \label{self_term_16}
    &= \eta_{r_1, j, l}^{(3)} \eta_{r_1, j, k}^{(3)}   (\eta_{r_3, h, k}^{(4)})^2 M^2
     \mathrm{Tr}\{ \mathbf{A}_{j,  k}^H \}  \mathrm{Tr}\{ \mathbf{A}_{h,  k} \} ,
\end{align}
\begin{align} \nonumber
    \mathbb{E}\{& \hat{\mathcal{Q}}_{r_1, j, k}^{(2)}(\Tilde{\mathbf{g}}_{r_1, k})  \mathcal{Q}_{r_2, j, k}^{(2)}(\Tilde{\mathbf{g}}_{r_2, k}) \mathcal{Q}_{r_3, h, k}^{(3)}(\Tilde{\mathbf{H}}_{r_3, h}) \hat{\mathcal{Q}}_{r_4, h, l}^{(3)}(\Tilde{\mathbf{H}}_{r_4, h}) \} \\ \label{self_term_17}
     &=  (\eta_{r_1, j, k}^{(2)})^2  \eta_{r_3, h, k}^{(3)} \eta_{r_3, h, l}^{(3)}
     M^2 \mathrm{Tr}\{ \mathbf{A}_{h,  k} \} g_{r_1, j},
\end{align}
\begin{align} \nonumber
    \mathbb{E}\{& \hat{\mathcal{Q}}_{r_1, j, l}^{(3)}(\Tilde{\mathbf{H}}_{r_1, j}) \mathcal{Q}_{r_2, j, i}^{(3)}(\Tilde{\mathbf{H}}_{r_2, j}) \mathcal{Q}_{r_3, h, k}^{(2)}(\Tilde{\mathbf{g}}_{r_3, k}) \hat{\mathcal{Q}}_{r_4, h, k}^{(2)}(\Tilde{\mathbf{g}}_{r_4, k}) \} \\ \label{self_term_18}
     &=  \eta_{r_1, j, l}^{(3)} \eta_{r_1, j, k}^{(3)}  (\eta_{r_3, h, k}^{(2)})^2
     M^2 \mathrm{Tr}\{ \mathbf{A}_{j,  k} \} g_{r_3, h}.
\end{align}

Finally, by summing up all the terms, we will arrive at
\begin{align} \nonumber
    \mathbb{E}\{ \hat{\mathbf{q}}_{j, k}^H \mathbf{q}_{j, k}  (\hat{\mathbf{q}}_{h, k}^H \mathbf{q}_{h, k})^H\} &= \sum_{r_1=1}^R\sum_{r_2=1}^R\sum_{r_3=1}^R\sum_{r_4=1}^R E_{j,k,h,k}^{\mathrm{CROSS-BS}}(\boldsymbol{\Phi}) + E_{j,k,h}^{\mathrm{SELF-UE}}(\mathbf{r}) \\\label{term_jkjkhkhk}
    &+ \sum_{l \in \mathcal{P}_k \backslash \{k\}}  E_{j,l,h,k}^{\mathrm{PC-CROSS-BS}}(\boldsymbol{\Phi}) + E_{j,l,h,k}^{\mathrm{PC - SELF-UE}}(\mathbf{r}),
\end{align}
where
\begin{align}
    E_{j,k,h}^{\mathrm{SELF}}(\mathbf{r}) &= (\eta_{r_1, j, k}^{(2)})^2  \eta_{r_3, h, k}^{(1)} \eta_{r_4, h, k}^{(1)}
M g_{r_1, j}
     \Bar{\mathbf{g}}_{r_3, k}^H \boldsymbol{\Phi}_{r_3}^H \Bar{\mathbf{H}}_{r_3, h}^H
     \Bar{\mathbf{H}}_{r_4, h}  \boldsymbol{\Phi}_{r_4}  \Bar{\mathbf{g}}_{r_4, k} \\ \nonumber
     &+ \eta_{r_1, j, k}^{(1)}  \eta_{r_2, j, k}^{(1)}  (\eta_{r_3, h, k}^{(2)})^2 M g_{r_3, h}
      \Bar{\mathbf{g}}_{r_1, k}^H \boldsymbol{\Phi}_{r_1}^H \Bar{\mathbf{H}}_{r_1, j} ^H
      \Bar{\mathbf{H}}_{r_2, j}  \boldsymbol{\Phi}_{r_2}  \Bar{\mathbf{g}}_{r_2, k} \\\nonumber
      &+ (\eta_{r_1, j, k}^{(4)})^2   \eta_{r_3, h, k}^{(1)} \eta_{r_4, h, k}^{(1)}
      M\mathrm{Tr}\{   \mathbf{A}_{j,  k}^H \}
     \Bar{\mathbf{g}}_{r_3, k}^H \boldsymbol{\Phi}_{r_3}^H \Bar{\mathbf{H}}_{r_3, h}^H
     \Bar{\mathbf{H}}_{r_4, h}  \boldsymbol{\Phi}_{r_4}  \Bar{\mathbf{g}}_{r_4, k} \\\nonumber
     &+ \eta_{r_1, j, k}^{(1)}  \eta_{r_2, j, k}^{(1)} (\eta_{r_3, h, k}^{(4)})^2  M\mathrm{Tr}\{   \mathbf{A}_{h,  k}^H \}
           \Bar{\mathbf{g}}_{r_1, k}^H \boldsymbol{\Phi}_{r_1}^H \Bar{\mathbf{H}}_{r_1, j} ^H
      \Bar{\mathbf{H}}_{r_2, j}  \boldsymbol{\Phi}_{r_2}  \Bar{\mathbf{g}}_{r_2, k} \\ \nonumber
      &+ (\eta_{j, k}^{(5)})^2 \eta_{r_3, h, k}^{(1)} \eta_{r_4, h, k}^{(1)}
      \mathrm{Tr}\{ \mathbf{A}_{j,  k} \}
       \Bar{\mathbf{g}}_{r_3, k}^H \boldsymbol{\Phi}_{r_3}^H \Bar{\mathbf{H}}_{r_3, h}^H
     \Bar{\mathbf{H}}_{r_4, h}  \boldsymbol{\Phi}_{r_4}  \Bar{\mathbf{g}}_{r_4, k} \\ \nonumber
     &+ \eta_{r_1, j, k}^{(1)}  \eta_{r_2, j, k}^{(1)}  (\eta_{h, k}^{(5)})^2 \mathrm{Tr}\{ \mathbf{A}_{h,  k} \}
      \Bar{\mathbf{g}}_{r_1, k}^H \boldsymbol{\Phi}_{r_1}^H \Bar{\mathbf{H}}_{r_1, j} ^H
      \Bar{\mathbf{H}}_{r_2, j}  \boldsymbol{\Phi}_{r_2}  \Bar{\mathbf{g}}_{r_2, k} \\ \nonumber
      &+ (\eta_{j, k}^{(5)}  \eta_{h, k}^{(5)})^2
        \mathrm{Tr}\{   \mathbf{A}_{j,  k}^H \} \mathrm{Tr}\{   \mathbf{A}_{h,  k}\} \\\nonumber
        &+ ( \eta_{r_1, j, k}^{(2)}   \eta_{r_3, h, k}^{(4)})^2 M^2 \mathrm{Tr}\{   \mathbf{A}_{h,  k}\} g_{r_1, j} \\\nonumber
        &+ (\eta_{r_1, j, k}^{(4)}  \eta_{r_3, h, k}^{(2)} )^2 M^2 \mathrm{Tr}\{   \mathbf{A}_{j,  k}^H\} g_{r_3, j} \\\nonumber
        &+  (\eta_{j, k}^{(5)}  \eta_{r_3, h, k}^{(2)})^2  M g_{r_3, h}  \mathrm{Tr}\{ \mathbf{A}_{j,  k}^H \} \\ \nonumber
        &+ (\eta_{r_1, j, k}^{(2)} \eta_{h, k}^{(5)})^2 M g_{r_1, k} \mathrm{Tr}\{  \mathbf{A}_{h,  k} \} \\\nonumber
        &+ ( \eta_{r_1, j, k}^{(4)} \eta_{h, k}^{(5)})^2 M \mathrm{Tr}\{ \mathbf{A}_{j,  k}^H \} \mathrm{Tr}\{ \mathbf{A}_{h,  k}^H \} \\\nonumber
        &+ ( \eta_{j, k}^{(5)}   \eta_{r_3, h, k}^{(4)})^2 M \mathrm{Tr}\{ \mathbf{A}_{j,  k}^H \} \mathrm{Tr}\{ \mathbf{A}_{h,  k} \} \\ \nonumber
        &+ (\eta_{j, k}^{(5)}   \eta_{r_3, h, k}^{(3)})^2 M \mathrm{Tr}\{ \mathbf{A}_{j,  k}^H \}  \mathrm{Tr}\{ \mathbf{A}_{h,  k} \} \\ \nonumber
        &+ (\eta_{r_1, j, k}^{(4)}  \eta_{r_3, h, k}^{(3)})^2 M^2
     \mathrm{Tr}\{ \mathbf{A}_{j,  k}^H \}  \mathrm{Tr}\{ \mathbf{A}_{h,  k} \} \\ \nonumber
     &+ (\eta_{r_1, j, k}^{(3)}  \eta_{r_3, h, k}^{(4)})^2 M^2
     \mathrm{Tr}\{ \mathbf{A}_{j,  k}^H \}  \mathrm{Tr}\{ \mathbf{A}_{h,  k} \} \\\nonumber
     &+ (\eta_{r_1, j, k}^{(2)}  \eta_{r_3, h, k}^{(3)})^2
     M^2 \mathrm{Tr}\{ \mathbf{A}_{h,  k} \} g_{r_1, j} \\\nonumber
     &+ (\eta_{r_1, j, k}^{(3)}  \eta_{r_3, h, k}^{(2)})^2
     M^2 \mathrm{Tr}\{ \mathbf{A}_{j,  k} \} g_{r_3, h},
\end{align}
and
\begin{align}
     E_{j,l,h,k}^{\mathrm{PC - SELF}}(\mathbf{r}) &= (\eta_{j, k}^{(5)})^2  \eta_{r_3, h, k}^{(3)} \eta_{r_3, h, l}^{(3)} M \mathrm{Tr}\{ \mathbf{A}_{j,  k}^H \}  \mathrm{Tr}\{ \mathbf{A}_{h,  k} \} \\ \nonumber
     &+ (\eta_{h, k}^{(5)})^2   \eta_{r_3, j, k}^{(3)} \eta_{r_3, j, l}^{(3)}  M \mathrm{Tr}\{ \mathbf{A}_{j,  k}^H \}  \mathrm{Tr}\{ \mathbf{A}_{h,  k} \} \\ \nonumber
     &+  (\eta_{r_1, j, k}^{(4)})^2  \eta_{r_3, h, k}^{(3)} \eta_{r_3, h, l}^{(3)}  M^2
     \mathrm{Tr}\{ \mathbf{A}_{j,  k}^H \}  \mathrm{Tr}\{ \mathbf{A}_{h,  k} \} \\ \nonumber
     &+ \eta_{r_1, j, l}^{(3)} \eta_{r_1, j, k}^{(3)}   (\eta_{r_3, h, k}^{(4)})^2 M^2
     \mathrm{Tr}\{ \mathbf{A}_{j,  k}^H \}  \mathrm{Tr}\{ \mathbf{A}_{h,  k} \} \\ \nonumber
     &+ (\eta_{r_1, j, k}^{(2)})^2  \eta_{r_3, h, k}^{(3)} \eta_{r_3, h, l}^{(3)}
     M^2 \mathrm{Tr}\{ \mathbf{A}_{h,  k} \} g_{r_1, j} \\\nonumber
     &+ \eta_{r_1, j, l}^{(3)} \eta_{r_1, j, k}^{(3)}  (\eta_{r_3, h, k}^{(2)})^2
     M^2 \mathrm{Tr}\{ \mathbf{A}_{j,  k} \} g_{r_3, h}.
\end{align}

\subsection{Derivation of $  \mathbb{E}\{|\hat{\mathbf{q}}_{j, k}^H  \mathbf{q}_{j, i}|^2\}$}

We now move on the derivation of $  \mathbb{E}\{|\hat{\mathbf{q}}_{j, k}^H  \mathbf{q}_{j, i}|^2\}$. Some of them have been calculated in $\mathbb{E}\{ \hat{\mathbf{q}}_{j, k}^H \mathbf{q}_{j, k}  (\hat{\mathbf{q}}_{h, k}^H \mathbf{q}_{h, k})^H\}$ with $h = j$. The are some extra terms that need to be calculated separately. The approach for the derivation is the same as before and hence we omit it for brevity and the final result is as follows.  

\begin{align} \nonumber
      \mathbb{E}\{|\hat{\mathbf{q}}_{j, k}^H  \mathbf{q}_{j, i}|^2\} &=  E_{j, k, i}^{\mathrm{C}} +  \sum_{r_1=1}^R\sum_{r_2=1}^R\sum_{r_3=1}^R\sum_{r_4=1}^R \Big( E_{j,k,j,i}^{\mathrm{CROSS-BS}}(\boldsymbol{\Phi}) +E_{r, j, k, i}^{\mathrm{CROSS-USER}}(\boldsymbol{\Phi})+ E_{j, k, i}^{\mathrm{NOISE}}(\boldsymbol{\Phi}) +  \\\label{E_jkji}
      &+ \sum_{l \in \mathcal{P}_k \backslash \{k\}}   E_{j,l,j,i}^{\mathrm{PC-CROSS-BS}}(\boldsymbol{\Phi})  +  E_{j, l, k, i}^{\mathrm{CROSS-USER-PC}}(\boldsymbol{\Phi}) 
     \Big),
\end{align}
with 
\begin{align} \label{sub_1}
    E_{j, k, i}^{\mathrm{C}} = (\eta_{j, k}^{(5)} \eta_{j, i}^{(5)})^2 \mathrm{Tr}\{ \mathbf{A}_{j,  k}^H \mathbf{A}_{j,  k}\} + \frac{\sigma^2}{\rho \tau_p} (\eta_{ j, i}^{(5)})^2 \mathrm{Tr}\{ \mathbf{A}_{j,  k}^H  \mathbf{A}_{j,  k} \} +  \sum_{l \in \mathcal{P}_k \backslash \{k\}} \Big( (\eta_{j, l}^{(5)} \eta_{j, i}^{(5)})^2 M\mathrm{Tr}\{ \mathbf{A}_{j,  k}^H \mathbf{A}_{j,  k}\}  \Big),
\end{align}
\begin{align} \nonumber
E_{r, j, k, i}^{\mathrm{CROSS-USER}}(\boldsymbol{\Phi}) &=  \eta_{r_1, j, k}^{(1)} (\eta_{r_2, j, i}^{(3)})^2  \eta_{r_4, j, k}^{(1)}  M f_{r_1, j, k}^H(\boldsymbol{\Phi}_r) f_{r_4, j, k}(\boldsymbol{\Phi}_{r_4})   \mathbf{a}_{N, r_1, j}^H  \mathbf{a}_{N, r_4, j}  \\ \nonumber
      &+ \eta_{r_1, j, k}^{(1)} (\eta_{r_2, j, i}^{(4)})^2  \eta_{r_4, j, k}^{(1)} f_{r_1, j, k}^H(\boldsymbol{\Phi}_r) f_{r_4, j, k}(\boldsymbol{\Phi}_{r_4}) \mathbf{a}_{N, r_1, j}^H \mathbf{a}_{N, r_4, j} \\ \nonumber
     &+ \eta_{r_1, j, k}^{(1)} (\eta_{j, i}^{(5)})^2  \eta_{r_4, j, k}^{(1)} f_{r_1, j, k}^H(\boldsymbol{\Phi}_r) f_{r_4, j, k}(\boldsymbol{\Phi}_{r_4})  \mathbf{a}_{N, r_1, j}^H  \mathbf{a}_{N, r_4, j} \\\nonumber
    &+ (\eta_{r_1, j, k}^{(2)} \eta_{r_2, j, i}^{(3)})^2   M^2 h_{r_1,j,k} +  (\eta_{r_1, j, k}^{(2)} \eta_{r_2, j, i}^{(4)})^2  M h_{r_1,j,k} + (\eta_{r_1, j, k}^{(2)} \eta_{j, i}^{(5)})^2 M h_{r_1,j,k} \\ \nonumber
    &+   (\eta_{r_1, j, k}^{(3)})^2 \eta_{r_2, j, i}^{(1)}  \eta_{r_3, j, i}^{(1)} M \mathrm{Tr}\{ \mathbf{A}_{j,  k}^H \Bar{\mathbf{H}}_{r_2, j}  \boldsymbol{\Phi}_{r_2} \Bar{\mathbf{g}}_{r_2, i} \Bar{\mathbf{g}}_{r_3, i}^H \boldsymbol{\Phi}_{r_3}^H \Bar{\mathbf{H}}_{r_3, j}^H  \mathbf{A}_{j,  k} \}  + (\eta_{r_1, j, k}^{(3)} \eta_{r_2, j, i}^{(2)})^2     M^2 h_{r_2, j, k} \\\nonumber
    &+ (\eta_{r_1, j, k}^{(3)} \eta_{j, i}^{(5)})^2 M \mathrm{Tr}\{ \mathbf{A}_{j,  k}^H \mathbf{A}_{j,  k}  \}  \\ \nonumber
     &+(\eta_{r_1, j, k}^{(4)})^2 \eta_{r_2, j, i}^{(1)} \eta_{r_3, j, i}^{(1)}  M \mathrm{Tr}\{ \mathbf{A}_{j,  k}^H \Bar{\mathbf{H}}_{r_2, j}  \boldsymbol{\Phi}_{r_2} \Bar{\mathbf{g}}_{r_2, i} \Bar{\mathbf{g}}_{r_3, i}^H \boldsymbol{\Phi}_{r_3}^H \Bar{\mathbf{H}}_{r_3, j}^H  \mathbf{A}_{j,  k}  \} \\\nonumber
     &+  (\eta_{r_1, j, k}^{(4)} \eta_{r_2, j, i}^{(2)})^2   M^2 h_{r_2,j,k} + (\eta_{r_1, j, k}^{(4)} \eta_{j, i}^{(5)})^2  M \mathrm{Tr}\{ \mathbf{A}_{j,  k}^H  \mathbf{A}_{j,  k} \} \\
     &+  (\eta_{j, k}^{(5)})^2   \eta_{r_2, j, i}^{(1)} \eta_{r_3, j, i}^{(1)} \mathrm{Tr}\{\mathbf{A}_{j,  k}^H  \Bar{\mathbf{H}}_{r_2, j}  \boldsymbol{\Phi}_{r_2}  \Bar{\mathbf{g}}_{r_2, i} \Bar{\mathbf{g}}_{{r_3}, i}^H \boldsymbol{\Phi}_{r_3}^H \Bar{\mathbf{H}}_{{r_3}, j}^H \mathbf{A}_{j,  k} \} \\ \nonumber
     &+  (\eta_{j, k}^{(5)} \eta_{r_1, j, i}^{(2)})^2 M h_{r_1,j,k}   \\ \nonumber
    &+ (\eta_{j, k}^{(5)} \eta_{r_1, j, i}^{(3)})^2  M\mathrm{Tr}\{ \mathbf{A}_{j,  k}^H \mathbf{A}_{j,  k}\} + (\eta_{j, k}^{(5)} \eta_{r_1, j, i}^{(4)})^2 \mathrm{Tr}\{ \mathbf{A}_{j,  k}^H \mathbf{A}_{j,  k}\} \Big)   \Big)  \Big),
\end{align}
and
\begin{align} \nonumber
    E_{r, j, k, i}^{\mathrm{NOISE}}(\boldsymbol{\Phi}) &= 
     \Big(
         \frac{\sigma^2}{\rho \tau_p} \eta_{r_1, j, i}^{(1)} \eta_{r_2, j, i}^{(1)} f_{r_1, j, i}(\boldsymbol{\Phi}_r) f_{r_2, j, i}^H(\boldsymbol{\Phi}_{r_2})   \mathrm{Tr}\{ \mathbf{A}_{j,  k}^H
     \mathbf{a}_{N, r_1, j} \mathbf{a}_{N, r_2, j}^H \mathbf{A}_{j,  k} \} \\ \nonumber
     &+ \frac{\sigma^2}{\rho \tau_p} (\eta_{r_1, j, i}^{(2)})^2 M h_{i, j} + \frac{\sigma^2}{\rho \tau_p} (\eta_{r_1, j, i}^{(3)})^2 M \mathrm{Tr}\{ \mathbf{A}_{j,  k}^H  \mathbf{A}_{j,  k} \} \\ \label{z_10_f}
      &+ \frac{\sigma^2}{\rho \tau_p} (\eta_{r_1, j, i}^{(4)})^2  \mathrm{Tr}\{ \mathbf{A}_{j,  k}^H  \mathbf{A}_{j,  k} \}  \Big),
\end{align}
\begin{align} \nonumber
   E_{j, l, k, i}^{\mathrm{CROSS-USER-PC}}(\boldsymbol{\Phi}) &= 
    (\eta_{r_1, j, l}^{(2)}  \eta_{r_2, j, i}^{(3)})^2 M^2 h_{r_1,j,k}  + (\eta_{r_1, j, l}^{(2)}  \eta_{r_2, j, i}^{(4)})^2 M h_{r_1,j,k} + (\eta_{r_1, j, l}^{(2)} \eta_{j, i}^{(5)})^2  M h_{r_1,j,k} \\ \nonumber
     &+  (\eta_{r_1, j, l}^{(3)})^2 \eta_{r_2, j, i}^{(1)}  \eta_{r_3, j, i}^{(1)} M \mathrm{Tr}\{ \mathbf{A}_{j,  k}^H \Bar{\mathbf{H}}_{r_2, j}  \boldsymbol{\Phi}_{r_2} \Bar{\mathbf{g}}_{r_2, i} \Bar{\mathbf{g}}_{r_3, i}^H \boldsymbol{\Phi}_{r_3}^H \Bar{\mathbf{H}}_{r_3, j}^H  \mathbf{A}_{j,  k} \}  \\\nonumber
    &+ (\eta_{r_1, j, l}^{(3)} \eta_{r_2, j, i}^{(2)})^2     M^2 h_{r_2, j, k} +  (\eta_{r_1, j, l}^{(3)} \eta_{r_2, j, i}^{(3)})^2    M^2 \mathrm{Tr}\{ \mathbf{A}_{j,  k}^H \mathbf{A}_{j,  k}  \} \\ \nonumber
    &+  (\eta_{r_1, j, l}^{(3)} \eta_{j, i}^{(5)})^2 M \mathrm{Tr}\{ \mathbf{A}_{j,  k}^H \mathbf{A}_{j,  k}  \} \\ \nonumber
    &+ (\eta_{r_1, j, l}^{(4)})^2 \eta_{r_2, j, i}^{(1)} \eta_{r_3, j, i}^{(1)}  M \mathrm{Tr}\{ \mathbf{A}_{j,  k}^H \Bar{\mathbf{H}}_{r_2, j}  \boldsymbol{\Phi}_{r_2} \Bar{\mathbf{g}}_{r_2, i} \Bar{\mathbf{g}}_{r_3, i}^H \boldsymbol{\Phi}_{r_3}^H \Bar{\mathbf{H}}_{r_3, j}^H  \mathbf{A}_{j,  k}  \} \\\nonumber
     &+  (\eta_{r_1, j, l}^{(4)} \eta_{r_2, j, i}^{(2)})^2   M^2 h_{r_2,j,k} \\ \nonumber
    &+  (\eta_{r_1, j, l}^{(4)} \eta_{r^\prime, j, i}^{(4)})^2  M^2 \mathrm{Tr}\{ \mathbf{A}_{j,  k}^H  \mathbf{A}_{j,  k} \} + (\eta_{r_1, j, l}^{(4)} \eta_{j, i}^{(5)})^2  M \mathrm{Tr}\{ \mathbf{A}_{j,  k}^H  \mathbf{A}_{j,  k} \} \\ \nonumber
    &+ (\eta_{j, l}^{(5)})^2   \eta_{r_2, j, i}^{(1)} \eta_{r_3, j, i}^{(1)} \mathrm{Tr}\{\mathbf{A}_{j,  k}^H  \Bar{\mathbf{H}}_{r_2, j}  \boldsymbol{\Phi}_{r_2}  \Bar{\mathbf{g}}_{r_2, i} \Bar{\mathbf{g}}_{{r_3}, i}^H \boldsymbol{\Phi}_{r_3}^H \Bar{\mathbf{H}}_{{r_3}, j}^H \mathbf{A}_{j,  k} \} \\ \nonumber
    &+  (\eta_{j, l}^{(5)} \eta_{r_1, j, i}^{(2)})^2 M h_{r_1,j,k}   + (\eta_{j, l}^{(5)} \eta_{r, j, i}^{(3)})^2  M\mathrm{Tr}\{ \mathbf{A}_{j,  k}^H \mathbf{A}_{j,  k}\} \\ \nonumber
     &+ (\eta_{j, l}^{(5)} \eta_{r_1, j, i}^{(4)})^2 \mathrm{Tr}\{ \mathbf{A}_{j,  k}^H \mathbf{A}_{j,  k}\} .
\end{align}

\subsection{Derivation of $  \mathbb{E}\{|\hat{\mathbf{q}}_{j, k}^H  \mathbf{q}_{j, k}|^2\}$}
The next part is the case when  $h = j$ which is equivalent to $  \mathbb{E}\{|\hat{\mathbf{q}}_{j, k}^H  \mathbf{q}_{j, k}|^2\}$. The result of this case equals to the terms obtained for the case $  \mathbb{E}\{\hat{\mathbf{q}}_{j, k}^H \mathbf{q}_{j, k}  (\hat{\mathbf{q}}_{h, k}^H \mathbf{q}_{h, k})^H\} $ plus the terms obtained for the case $  \mathbb{E}\{|\hat{\mathbf{q}}_{j, k}^H  \mathbf{q}_{j, i}|^2\}$ since it covers all the possible cases. Hence,  $  \mathbb{E}\{|\hat{\mathbf{q}}_{j, k}^H  \mathbf{q}_{j, k}|^2\}$ is given by

\begin{align} \nonumber
      \mathbb{E}\{&|\hat{\mathbf{q}}_{j, k}^H  \mathbf{q}_{j, i}|^2\} =  E_{j, k, j}^{\mathrm{C}} +  \sum_{r_1=1}^R\sum_{r_2=1}^R\sum_{r_3=1}^R\sum_{r_4=1}^R \Big( E_{j,k,j,k}^{\mathrm{CROSS-BS}}(\boldsymbol{\Phi}) +E_{r, j, k, k}^{\mathrm{CROSS-USER}}(\boldsymbol{\Phi})+ E_{j,k,k}^{\mathrm{SELF-UE}}(\mathbf{r}) \\\label{E_jkjk}
      &+E_{j, k, k}^{\mathrm{NOISE}}(\boldsymbol{\Phi})  + \sum_{l \in \mathcal{P}_k \backslash \{k\}}   E_{j,l,j,k}^{\mathrm{PC-CROSS-BS}}(\boldsymbol{\Phi})  +  E_{j, l, k, k}^{\mathrm{CROSS-USER-PC}}(\boldsymbol{\Phi})  +  E_{j,l,j,k}^{\mathrm{PC - SELF-UE}}(\mathbf{r}).
\end{align}

The only remaining part to be calculated is $\mathbf{V}_k = \mathrm{diag}\left( \mathbb{E}\{\rVert \mathbf{v}_{1, k}\rVert^2\}, \dots, \mathbb{E}\{\rVert\mathbf{v}_{J, k}\rVert^2 \} \right)$ which will be calculated next. The $j$th entry of $\mathbf{V}_k$ could be written as
\begin{align}
     \mathbb{E}\{\rVert\mathbf{v}_{j, k}\rVert^2 \} &=  \mathbb{E}\{\rVert  \hat{\mathbf{q}}_{j,  k} \rVert^2 \} = \mathbb{E}\{ \hat{\mathbf{q}}_{j,  k}^H \hat{\mathbf{q}}_{j,  k} \}.
\end{align}
According to the orthogonality property of the LMMSE estimator, we  will have $\mathbb{E}\{ \hat{\mathbf{q}}_{j,  k}^H \hat{\mathbf{q}}_{j,  k} \} = \mathbb{E}\{ \hat{\mathbf{q}}_{j,  k}^H \mathbf{q}_{j,  k} \}$ where $\mathbb{E}\{ \hat{\mathbf{q}}_{j,  k}^H \mathbf{q}_{j,  k} \}$ is calculated in (\ref{exp_qhq}). Thus, all the terms in (\ref{SINR_here}) are calculated and this completes the proof.

\newpage
\ifCLASSOPTIONcaptionsoff
  \newpage
\fi



%



\bibliographystyle{IEEEtran}


\end{document}